\documentclass[12pt]{article}
\usepackage{amsmath,mathrsfs, eufrak, mathtools,amsthm}
\usepackage{amssymb}

\usepackage{bbm}
\usepackage{nicefrac}
\usepackage{newtxmath,newtxtext}
\usepackage{graphics,framed, wrapfig}
\usepackage{pgf}
\usepackage{tikz,pgfplots}
\usepackage{mhchem}
\usepackage{mathtools}

\usepackage{microtype}
\usepackage{titlesec}
\usepackage[top=1.0in,bottom=1.0in,left=1.5in,right=1in]{geometry}

\titleformat{\subsubsection}[runin]{\normalfont\bfseries}{\thesubsubsection}{1em}{}

\usepackage{pstricks,pst-plot,psfrag}

\usepackage{graphicx,subfigure,xspace,bm,tikz,pgfplots}
\usepackage[lined,linesnumbered,algoruled,noend]{algorithm2e}

\newtheorem{theorem}{Theorem}[section]
\newtheorem{definition}{Definition}

\newtheorem{lemma}{Lemma}

\newtheorem{remark}{Remark}

\newtheorem{example}{Example}

\newtheorem{proposition}[theorem]{Proposition}  
\renewcommand{\t}{\mathfrak{t}}
\renewcommand{\a}{\mathfrak{a}}
\newcommand{\pp}{\mathfrak{p}}

\newcounter{para}[section]

\newcommand{\ad}{\operatorname{ad}}
\newcommand{\Ad}{\operatorname{Ad}}

\newcommand{\h}{\mathfrak{h}}

\newcommand{\oprocendsymbol}{\hbox{$\bullet$}}
\newcommand{\oprocend}{\relax\ifmmode\else\unskip\hfill\fi\oprocendsymbol}

\newcommand{\g}{\mathfrak{g}}

\newcommand{\rp}{{\rm p}}

\newcommand{\R}{\mathbb{R}}
\newcommand{\N}{\mathbb{N}}

\renewcommand{\cal}{\mathcal}

\newcommand{\dep}{\operatorname{dep}}

\newcommand{\SP}{\operatorname{Sp}}
\newcommand{\SL}{\operatorname{SL}}
\newcommand{\SO}{\operatorname{SO}}
\newcommand{\SU}{\operatorname{SU}}

\renewcommand{\sl}{\mathfrak{sl}}
\newcommand{\so}{\mathfrak{so}}

\newcommand{\tr}{\text{Tr}}

\renewcommand{\tr}{\operatorname{tr}}
\newcommand{\C}{\mathbb{C}}

\renewcommand{\span}{\operatorname{Span}}

\definecolor{BBlue}{cmyk}{.98,0.10,0,.25}

\begin{document}

\title{Structure Theory for Ensemble Controllability, Observability, and Duality}
%\rule{\textwidth}{0.4mm}}
\date{}
\maketitle

\vspace{-2cm}
\begin{flushright}
{\small {\bf Xudong Chen\footnote[1]{X. Chen is with the ECEE Dept., CU Boulder. {\em Email: xudong.chen@colorado.edu}.} 
}}
\end{flushright}

\begin{abstract}
Ensemble control deals with the problem of using a finite number of control inputs to simultaneously steer a large population (in the limit, a continuum) of control systems. Dual to the ensemble control problem, ensemble estimation deals with the problem of using a finite number of measurement outputs to estimate the initial state of every individual system in the ensemble.   
We introduce in the paper a novel class of ensemble systems, termed {\em distinguished ensemble systems}, and establish sufficient conditions for  controllability and observability of such systems. 

Every distinguished ensemble system has two key components, namely a set of {\em distinguished control vector fields} and a set of {\em codistinguished observation functions}. Roughly speaking, a set of vector fields is distinguished if it is closed (up to scaling) under Lie bracket, and moreover, every vector field in the set can be obtained by a Lie bracket of two vector fields in the same set. Similarly, a set of functions is codistinguished  to a set of vector fields if the Lie derivatives of the functions along the given vector fields yield (up to scaling)  the same set of functions.  We demonstrate in the paper that the structure of a distinguished ensemble system can significantly simplify the analysis of ensemble controllability and observability. Moreover, such a structure can be used as a guiding principle for ensemble system design.  

We further address in the paper the problem about existence of a distinguished ensemble system for a given manifold. We provide an affirmative answer for the case where the manifold is a connected semi-simple Lie group. Specifically, we show that every such Lie group admits a set of distinguished vector fields, together with a set of codistinguished functions. The proof is constructive, leveraging the structure theory of semi-simple real Lie algebras and representation theory. Examples will be provided along the presentation of the paper illustrating key definitions and main results.     
\end{abstract}


\section{Introduction}\label{sec:introduction}

We address in the paper controllability  and observability of a continuum ensemble of control systems. Roughly speaking, ensemble control deals with the problem of using a {\em finite} number of control inputs to simultaneously steer a large population (in the limit, a {\em continuum}) of control systems. These individual control systems may be structurally identical, but show variations in their tuning parameters. Dual to ensemble control, ensemble estimation deals with the problem of estimating the state of {\em every} individual control system in the ensemble using only a {\em finite} number of measurement outputs. We refer the reader to Fig.~\ref{fig:opening} for an illustration of a continuum ensemble of control systems indexed by a parameter of a two-dimensional surface. 
Note that any finite ensemble of control systems can be viewed as a proper subsystem of the continuum ensemble. Controllability (or observability) of the continuum ensemble will guarantee the controllability (or observability) of any such finite subsystem of it. 

\begin{wrapfigure}[21]{r}{7.5cm}\vspace{-1.2cm}
\begin{center}
\includegraphics[width = 0.45\textwidth]{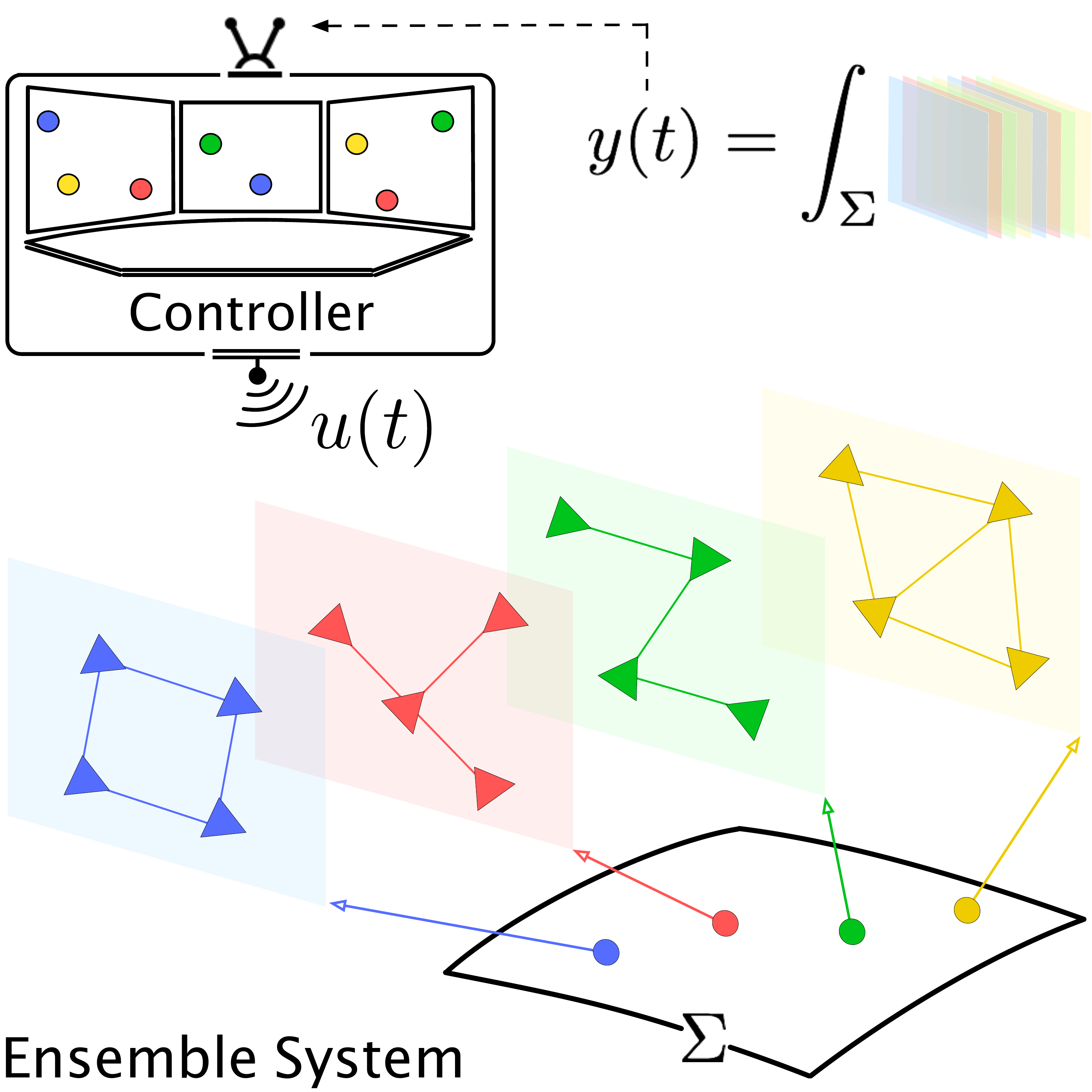}
\caption{\small  A continuum ensemble of systems indexed by a parameter of a surface~$\Sigma$. %i.e., each point of $\Sigma$ corresponds to a single networked system. 
A controller broadcasts a signal $u(t)$ as a common control input to  steer every individual system in the ensemble. Meanwhile, it receives a measurement output $y(t)$ integrating the information of individual states of all the systems.}\label{fig:opening}
\end{center}
\end{wrapfigure}

The framework of ensemble control and estimation naturally has many applications across various disciplines in engineering and science. The individual control systems in the ensemble can be used to model, for example, spin dynamics that are controlled by a magnetic field~\cite{glaser1998unitary,appelbaum2016spin}, molecules that respond to external stimuli such as light~\cite{yu2003photomechanics} and heat~\cite{taniguchi2018walking}, or micro-robotics that are steered by a broadcast control signal~\cite{becker2012approximate}.   
We further note that an individual control system does not necessarily have only one single physical entity, but rather it can comprise multiple interacting components (or agents). In this case, every individual control system is itself a networked control system (or a multi-agent system). For example, a mathematical model for a continuum ensemble of multi-agent formation systems has recently been proposed and investigated in~\cite{chen2018controllability}.

Many existing ensemble control and estimation theories deal only with linear ensembles (i.e., ensembles of linear control systems). For nonlinear ensembles, the literature is relatively sparse on controllability, and much less on observability. There is also a lack of methodologies for designing nonlinear dynamics of individual control systems so that an  ensemble of such systems is controllable and observable. 
To address the above issues, we introduce in the paper a novel class of nonholonomic ensemble systems, termed {\em distinguished ensembles}. Every such ensemble system has two key components: a set of finely structured control vector fields, termed {\em distinguished vector fields}, and a set of co-structured observations functions, termed {\em codistinguished functions}. Details about the structure of a distinguished ensemble will be provided below. 
We will demonstrate that controllability and/or observability of a distinguished ensemble system can be easily fulfilled under some mild assumption. The first half of the paper is devoted to establishing the fact.     
For the second half, we will investigate the problem about existence of a distinguished ensemble. We focus on the case where the state space of every individual system is a Lie group or its homogeneous space. We leverage tools from structure theory of semi-simple real Lie algebras and representation theory to construct explicitly distinguished vector fields and codistinguished functions on those spaces.  

We introduce below the model of a distinguished ensemble system in details. We also provide literature review and outline the contributions of the paper.      

\subsection{Mathematical models for ensemble control and estimation} 
The model of an ensemble system considered in the paper comprises two parts, namely ensemble control and ensemble estimation. We introduce these two parts subsequently.  
  
{\em Model for ensemble control.} 
We consider a continuum ensemble of control systems indexed by a parameter~$\sigma\in \Sigma$, where~$\Sigma$ is the {\bf parametrization space}. We assume in the paper that~$\Sigma$ is compact, analytic, and path-connected. If an individual control system in the ensemble is associated with index~$\sigma$, then we call it {\bf system-$\sigma$}. The state space of each individual system is the same, which we denote by~$M$. We assume that~$M$ is analytic. 
Further, let $x_\sigma(t)\in M$ be the state of system-$\sigma$ at time~$t$. Then, 
in general, the control model of an ensemble system can be described by the following differential equation:  
\begin{equation}\label{eq:generalensemblemodel}
\dot x_\sigma(t) := \frac{\partial x_\sigma(t)}{\partial t}=  f(x_\sigma(t), \sigma, u(t)), \quad x_\sigma\in M \mbox{ for all } \sigma \in \Sigma,  
\end{equation}
where $u(t)$ is a finite dimensional control input common to {\em all} of the individual control systems, and~$f$ is an analytic vector field.

Let $x_\Sigma(t):= \{x_\sigma(t)\mid \sigma\in \Sigma\}$ be the collection of system states. We call $x_\Sigma(t)$ a {\bf profile}. Let ${\rm C}^\omega(\Sigma, M)$ be the space of analytic functions from~$\Sigma$ to~$M$. We assume that for any given~$t$, the profile $x_\Sigma(t)$ belongs to ${\rm C}^\omega(\Sigma, M)$. We call ${\rm C}^\omega(\Sigma, M)$ the {\bf profile space}.  
Ensemble controllability is then about the ability of using the common control input $u(t)$ to steer ensemble system~\eqref{eq:generalensemblemodel} from an arbitrary initial profile~$x_\Sigma(0)$ to any target profile~$\hat x_\Sigma$ at any given time~$T > 0$. A precise definition will be provided in Def.~\ref{def:ensemblecontrollability}, Section~\S\ref{sec:controlobserverdual}. %.  
%Although in the paper we mainly deal with analytic ensemble systems, the results can be extended to the smooth case with a few more technical details of arguments. 
%++++ Cauchy-Kowalevski theorem +++++

%The controllability problem for the general ensemble system~\eqref{eq:generalensemblemodel} is hard. 
We focus in the paper on a special class of ensemble systems, namely systems such that the vector fields $f$ are separable in state~$x_\sigma(t)$, the parameter~$\sigma$, and the control input~$u(t)$. Specifically,  we consider the following type of ensemble system:    
\begin{equation}\label{eq:controlonly}
\dot x_\sigma(t) = f_0(x_\sigma(t), \sigma) + \sum^m_{i = 1} \sum^{r}_{s = 1} u_{i,s}(t)\rho_{s}(\sigma) f_i (x_\sigma(t)), \quad  x_\sigma\in M \mbox{ and } \sigma\in \Sigma, 
\end{equation}
where $f_0$ is a drifting term, the $f_i$'s are {\bf control vector fields} depending only on~$x_\sigma(t)$, the $\rho_i$'s are {\bf parametrization functions} defined on~$\Sigma$, and the $u_{i,s}$'s are scalar control inputs. We assume in the paper that all the vector fields and parametrization functions are analytic in their variables. All the control inputs are integrable functions over any finite time interval.  For convenience, we let $u(t)$ be the collection of all the $u_{i,s}(t)$'s.

{\em Model for ensemble estimation.} %We next describe the dual ensemble estimation model. 
%For the dual ensemble estimation problem, 
We assume that there are~$l$ (scalar) measurement outputs $y^j(t)$, for $j = 1, \ldots, l$, at our disposal. % the problem of ensemble observability is about whether or not one can estimate the entire initial  profile $x_\Sigma(0)$.     
Each~$y^j(t)$ is a certain average of an observation function~$\phi^j(x_\sigma(t))$ over the parametrization space~$\Sigma$. Specifically, we let $\Sigma$ be equipped with a positive Borel measure, and each $\phi^j$, for $j = 1,\ldots, l$, be an analytic function defined on~$M$. Then, the measurement outputs $\{y^j(t)\}^l_{j = 1}$ are described by% the following integrals:    
$$
y^j(t) =\displaystyle\int_\Sigma \phi^j(x_\sigma(t)) d\sigma,     \quad j = 1,\ldots, l.
$$
For convenience, let $y(t)$ be the collection of the $y^j(t)$'s. 
%Such type of measurement output arises naturally in physics. Consider, for example, the temperature or pressure in classical thermodynamics, the NMR spectrum of an ensemble of spin systems, the  far-field intensity in Laser beam combination. 
The question related to ensemble observability is to ask whether one can use a certain control input $u(t)$ to excite system~\eqref{eq:controlonly} and then, estimate  $x_\Sigma(0)$ from $y(t)$. A precise definition will be provided in Section~\S\ref{sec:controlobserverdual}.

{\em Model for an ensemble system.} Combining the above two parts,  we arrive at the following mathematical model of an ensemble system:  
\begin{equation}\label{eq:ensemblesystem0}
\left\{
\begin{array}{lr}
\dot x_\sigma(t) = f_0(x_\sigma(t), \sigma) + \sum^m_{i = 1}\sum^r_{s = 1}u_{i,s}(t)\rho_{s}(\sigma) f_i (x_\sigma(t)),  & \forall \sigma\in \Sigma, \vspace{3pt}\\
  y^j(t)  =\displaystyle\int_\Sigma \phi^j(x_\sigma(t)) d\sigma,    & \forall j = 1,\ldots, l. \\
 \end{array}
 \right. 
\end{equation}
Examples of the above system will be given along the  presentation.

\subsection{Distinguished structure and examples}

A major contribution of the paper is to introduce a novel class of 
 nonholonomic ensemble systems~\eqref{eq:ensemblesystem0}, termed {\em distinguished ensembles}. Every such ensemble system has two key components: a set of distinguished control vector fields $\{f_i\}^m_{i = 1}$ and a set of codistinguished observation functions $\{\phi^j\}^l_{j = 1}$.  
%structure on the control vector fields and the observation functions of  system~\eqref{eq:ensemblesystem0}, %termed as {\em distinguished ensemble system}, which is comprised of two key components,  
Roughly speaking,  
a set of  vector fields $\{f_i\}^m_{i = 1}$ is said to be {\em distinguished} if the Lie bracket of any two vector fields $f_i$ and $f_j$ is, up to scaling, another vector field~$f_k$, i.e., $[f_i,f_j]= \lambda f_k$ for $\lambda$ a constant,  and conversely, any vector field $f_k$ in the set can be obtained in this way. Similarly, a set of functions $\{\phi^j\}^l_{j = 1}$ is said to be {\em codistinguished} to the vector fields $\{f_i\}^m_{i = 1}$ if the Lie derivative of any $\phi^j$ along any $f_i$ is, up to scaling, another function~$\phi^k$, i.e., $f_i\phi^j = \lambda \phi^k$ for $\lambda$ a constant, and conversely, any function $\phi^k$ in the set can be obtained in this way    
(see Def.~\ref{def:distvecf} and Def.~\ref{def:codisth},  Section~\S\ref{ssec:keydefinition} for details). %We call system~\eqref{eq:ensemblesystem0} a {\em distinguished ensemble system} if 
%We demonstrate in the paper that distinguished vector fields and codistinguished functions are key components in establishing ensemble controllability and observability of~\eqref{eq:controlonly}. %Furthermore, the structure of a distinguished ensemble system can serve as a principle for designing a controllable and observable ensemble system.   

%{\em Examples of distinguished sets of Lie algebras.} 
We note here that although the notion of a ``distinguished set'' of a Lie algebra appears to be new, such set arises naturally in many places. Here are a few examples: 

{\em 1)} When dealing with the rigid motions of a three dimensional object with a fixed center, we have that the infinitesimal motions of rotations around three axes of an orthonormal frame~$\Theta \in \SO(3)$ are given by 
$$
f_1(\Theta) = \Theta\Omega_{23}, 
%\begin{bmatrix}
% 0 & 0 & 0 \\
% 0 & 0 & 1\\
% 0 & -1 & 0 
%\end{bmatrix}
\qquad
f_2(\Theta) := \Theta \Omega_{31},
%\begin{bmatrix}
% 0 & 0 & -1 \\
% 0 & 0 & 0\\
% 1 & 0 & 0 
%\end{bmatrix} 
\qquad
 f_3(\Theta) := 
\Theta \Omega_{12},
%\begin{bmatrix}
% 0 & 1 & 0 \\
% -1 & 0 & 0\\
% 0 & 0 & 0 
%\end{bmatrix}, 
$$
where each $\Omega_{ij}$ is a skew-symmetric matrix with $1$ on the~$ij$th entry, $-1$ on the $ji$th entry, and~$0$ elsewhere. 
By computation, $[f_i, f_j] = f_k$ where $(i,j,k)$ is any cyclic rotation of $(1,2,3)$. Thus, the above vector fields form a distinguished set.

{\em 2)} In quantum mechanics, the Pauli spin matrices are used to represent angular momentum operators. We recall that they are given by 
$$
\sigma_1 := 
\begin{bmatrix}
0 & 1 \\
1 & 0
\end{bmatrix} \qquad
\sigma_2 := \begin{bmatrix}
0 & -\mathrm{i} \\
\mathrm{i} & 0 
\end{bmatrix}\qquad
\sigma_3 := \begin{bmatrix}
1 & 0\\
0 & -1
\end{bmatrix},
$$ 
where $\mathrm{i}$ is the imaginary unit. Similarly, if $(i,j,k)$ is a cyclic rotation of $(1,2,3)$, then $[\sigma_i, \sigma_j ]_{\rm M} = 2\mathrm{i}\sigma_k$ where $[\cdot, \cdot]_{\rm M}$ denotes the matrix commutator. Although the constant~$2\mathrm{i}$ is not real, one can multiple all the three matrices by~$\mathrm{i}$ so that the new set $\{\mathrm{i}\sigma_i\}^3_{i = 1}$ now satisfies 
$[\mathrm{i}\sigma_i,  \mathrm{i}\sigma_j ]_{\rm M} = -2\mathrm{i}\sigma_k$. Note that the set of matrices $\{\mathrm{i}\sigma_i\}^3_{i = 1}$ belongs to $\mathfrak{su}(2)$ i.e., the Lie algebra associated with the special unitary group ${\rm SU}(2)$. However, we shall note that $\mathfrak{su}(2)$ is isomorphic to $\so(3)$.  

{\em 3)} We also note that the ladder operators represented by the following matrices in the special linear Lie algebra $\sl(2,\R)$: 
$$
H := 
\begin{bmatrix}
1 & 0 \\
0 & -1
\end{bmatrix} \qquad
X :=
\begin{bmatrix}
0 & 1 \\
0 & 0
\end{bmatrix} \qquad
Y:= 
\begin{bmatrix}
0 & 0 \\
1 & 0
\end{bmatrix}
$$ 
satisfy the desired property: $[H, X]_{\rm M} = 2X$, $[H, Y]_{\rm M} = -2Y$, and $[X, Y]_{\rm M} = H$.  
 
%{\em 4)} We further recall that the standard unicycle model (Dubins car) is given by
%$$
%\begin{bmatrix}
%\dot x_1 \\
%\dot x_2 \\
%\dot \theta
%\end{bmatrix} = 
%u_1 
%\begin{bmatrix}
%\cos\theta \\
%\sin\theta \\
%0
%\end{bmatrix} + u_2 
% \begin{bmatrix}
%0\\
%0 \\
%1
%\end{bmatrix}. 
%$$
%Although the set of the two control vector fields is not closed under Lie bracket:
%$$  
%\left [f_1 := 
%\begin{bmatrix}
%\cos\theta \\
%\sin\theta \\
%0
%\end{bmatrix}, 
%f_2: = 
% \begin{bmatrix}
%0\\
%0 \\
%1
%\end{bmatrix}\right ] = f_3:= 
%\begin{bmatrix}
%\sin\theta \\
%-\cos\theta \\
%0
%\end{bmatrix},
%$$
%the above three of them are; indeed, $[f_1, f_3] = 0$ and $[f_2, f_3] = f_1$. However, $\{f_i\}^3_{i = 1}$ is {\em not} qualified as a distinguished set because~$f_2$ cannot be generated by a Lie bracket of any two vector fields in the set. %Nevertheless, a (continuum) ensemble of unicycles is  
%

The examples given above demonstrate the existence of distinguished sets in Lie algebras $\mathfrak{so}(3) \approx \mathfrak{su}(2)$ and $\sl(2, \R)$. In fact,  we have shown in~\cite{chen2017distinguished} that every semi-simple {\em real} Lie algebra has a distinguished set. We review such a fact in Subsection~\S\ref{ssec:distinguishedsetofliealgebra}. %We will 

%Note that for codistinguished functions, the definitio naturally involves Lie algebra representation. We give examples along the presentation.    

\subsection{Literature review}
Amongst related works about controllability of nonlinear ensembles, we first mention~\cite{li2006control,JSL-NK:09} by Li and Khaneja  in which the authors establish the controllability of  an ensemble of Bloch equations parametrized by a pair of scalar parameters $(\sigma_1,\sigma_2)$ over $[a_1,b_1]\times [a_2, b_2]$ in $\R^2$:   
$$
\dot x(t) = (\sigma_1 \Omega_{12} +  u_1(t) \sigma_2 \Omega_{13} +  u_2(t)\sigma_2 \Omega_{23} )x(t). 
$$
%Note that the above model is a typical example of~\eqref{eq:controlonly} where $\omega \Omega_{12} x(t)$ is the drifting term, $\Omega_{13}x(t)$ and $\Omega_{23}x(t)$ are the control vector fields. %and there is a single parametrization function given by $\rho(\sigma) = \sigma$.   
Ensemble control of Bloch equations has also been addressed in~\cite{beauchard2010controllability} using tools from functional analysis.  
We further  note that the controllability of a general ensemble of control-affine systems has been recently addressed in~\cite{agrachev2016ensemble}, in which the authors established an ensemble version of Rachevsky-Chow theorem via a Lie algebraic method.  
We do not to intend to reproduce in the paper the results established there, but rather our contribution related to ensemble controllability is to demonstrate that if the set of control vector fields $\{f_i\}^m_{i = 1}$  is distinguished, then the ensemble version of Rachevsky-Chow criterion can be easily verified in analysis and fulfilled in system design. 
For ensemble control of linear systems, we refer the reader to~\cite{li2011ensemble,helmke2014uniform},~\cite[Ch.~12]{fuhrmann2015mathematics} and references therein.  We further refer the reader to~\cite{brockett2007optimal,chen2016optimal1,chen2016optimal2} for optimal control of probability distributions evolving along  linear systems.   

%{\em Ensemble observability.} %To the best of author's knowledge, 
Observability of a continuum ensemble system has been mostly addressed within the class of linear systems. We first refer the reader to~\cite[Ch.~12]{fuhrmann2015mathematics} 
where the following ensemble of linear systems is investigated: $$\dot x_\sigma(t)  = A(\sigma) x_\sigma(t) \in \R^n, \quad y(t) =  \int_\Sigma C(\sigma) x_\sigma(t) d\sigma\in \R^l.$$ 
The authors addressed the observability of the above ensemble system using the duality between controllability and observability of infinite-dimensional linear systems~\cite{curtain2012introduction}.  
We also refer the reader to~\cite{zeng2016ensemble} for a related observability problem about estimating the probability distribution of the initial state. Specifically, the authors there considered a {\em single} time-invariant linear system:
$\dot x(t) = A x(t) + Bu(t)$ and  $y(t) = C x(t)$.  
An initial probability distribution~$p_0$ of~$x\in \R^n$ induces a distribution $\bar p_t$ of $y(t)$ for a given control input~$u(t)$. The observability problem addressed there is whether one is able to estimate $p_0$ given that the entire distributions $\bar p_t$, for all $t \ge 0$, are known. 
%For observability of nonlinear ensemble systems, we refer the reader to~\cite{} where we focussed on a specific ensemble of Bloch equations and investigated observability and system identifiability.  
We further refer the reader to~\cite{hermann1977nonlinear,gauthier1981observability,van1982observability,gauthier1994observability,de2000observability} for the study of observability of a {\em single} nonlinear system using the so-called observability codistribution.

\subsection{Outline of contribution and organization of the paper}

The technical contribution of the paper is two-fold: {\em 1)} We establish a structure theory for controllability and observability of a distinguished ensemble system. {\em 2)} We prove the existence of distinguished ensemble systems over semi-simple Lie groups. 

%The first item has been addressed in the previous subsection.  
%We provide below more details for the last two items.  

%{\em (i). Introduction of distinguished ensemble systems.} We have address this item in the previous subsection. 

{\em Structure theory.} We establish in Section~\S\ref{sec:controlobserverdual} a sufficient condition for controllability and observability of a distinguished (and pre-distinguished) ensemble system. In particular, 
we demonstrate how distinguished vector fields and codistinguished functions can simplify the analysis and lead to ensemble controllability and observability. %Specifically, we show that having such a structure simplifies the analyses by transposing the controllability (or observability) problem to a problem of finding a separating set of ${\rm L}^2(\Sigma)$, i.e., the Hilbert space of square-integrable functions on~$\Sigma$. 
 The structure theory established in the paper also provides a solution to the problem of ensemble system design---i.e., the problem of co-designing the control vector fields $f_i$'s, the observations functions $\phi^j$'s, and the parametrization functions $\rho_s$'s so that system~\eqref{eq:ensemblesystem0} is controllable and/or observable. In particular, it divides the  problem into two independent subproblems---one is about finding a set of distinguished vector fields $\{f_i\}^m_{i = 1}$ and a set of codistinguished function $\{\phi^j\}^l_{j = 1}$ over the given manifold~$M$ while the other is about finding a set of parametrization functions $\{\rho_s\}^r_{s =1}$ that separates points of the parametrization space~$\Sigma$. %${\rm L}^2(\Sigma)$, i.e., the space of all square-integrable functions. %Thus, the structure theory can be viewed as a guiding principle for designing controllable and observable ensemble systems. 

{\em  Existence of distinguished ensembles.}  We prove in Section~\S\ref{sec:onliegroup} that every semi-simple Lie group~$G$ admits a set of  distinguished vector fields, together with a set of codistinguished functions. %over a given manifold~$M$. We provide an affirmative answer for the case where~$M$ is an arbitrary semi-simple Lie group~$G$. %A  few commonly seen examples include special unitary group $\SU(n)$, special orthogonal group $\SO(n)$, special linear group $\SL(n)$, symplectic group $\SP(2n, \R)$, etc.  
The proof of the existence result is constructive: {\em 1)}~For distinguished vector fields, we leverage the result established in~\cite{chen2017distinguished} where we have shown how to construct a distinguished set on the Lie algebra level. %A sketch of the construction is provided in Subsection~\S\ref{ssec:distinguishedsetofliealgebra}. One then 
We then identify the distinguished set with the corresponding set of left- (or right-) invariant  vector fields over the group~$G$.  
{\em 2)}~For codistinguished functions, we show how to generate these functions using representation theory. In particular, we show in Section~\S\ref{ssec:matrixcoeffi} that a selected set of matrix coefficients associated with a finite dimensional Lie group representation could be used as a set of codistinguished functions (with respect to a set of left-invariant vector fields). Then, in Section~\S\ref{ssec:adjointrepresentationcodist}, we focus on a special representation, namely the adjoint representation. We show,  in this case, that there indeed exists a set of matrix coefficients as codistinguished functions. In particular, if~$G$ is a matrix Lie group, then these matrix coefficients are simply given by $\phi^{ij}(g) = \tr(g X_j g^{-1} X_i^\top)$ where $X_i$ and $X_j$ are selected matrices out of the Lie algebra~$\g$ of~$G$. We further address, in Section~\S\ref{ssec:homogeneousspace}, the existence problem for homogeneous spaces.    

We provide key definitions and notations in Section~\S\ref{sec:definitionsandnotations} and conclusions  at the end. 

%{\em Organization of the paper.} In Section~\S\ref{sec:definitionsandnotations}, we introduce key definitions and notations used in the paper. Next, in Section~\S\ref{sec:controlobserverdual}, we introduce formally distinguished vector fields and codistinguished functions. We also establish in the section a sufficient condition for controllability and observability of a distinguished ensemble system. Then, in Section~\S\ref{sec:onliegroup}, we address the existence of distinguished vector fields and codistinguished functions on Lie groups and their homogeneous spaces. We provide conclusions at the end.     

\section{Definitions and Notations}\label{sec:definitionsandnotations}

% We introduce in the section key definitions and notations used in the paper. We also recall a few known facts in differential geometry and Lie groups/algebras.  
 
 \subsection{Geometry, topology, and algebra}
{\em 1. Manifolds.} Let $M$ be a real analytic manifold. For a point $x\in M$, let $T_xM$ be the tangent space and $T^*_x M$ be the cotangent space of~$M$ at~$x$. Let $TM:= \cup_{x\in M} T_xM$  be the tangent bundle and $T^*M := \cup_{x\in M} T^*_xM$ be the cotangent bundle. 
     
Let ${\rm C}^\omega(M)$ be the set of real analytic functions on~$M$. Denote by ${\bf 1}_M\in {\rm C}^\omega(M)$ the constant function whose value is $1$ everywhere.  Let $\mathfrak{X}(M)$ be the set of analytic vector fields over $M$. Let $\phi\in {\rm C}^\omega(M)$ and $f\in \mathfrak{X}(M)$. Denote by $f\phi\in {\rm C}^\omega(M)$ the Lie derivative of~$\phi$ along~$f$. 
If we embed $M$ into a Euclidean space, then $f\phi$ is simply given by
$$(f\phi)(x) := \lim_{\epsilon \to 0} \frac{\phi(x + \epsilon f(x)) - \phi(x)}{\epsilon}, \quad \forall x\in M.$$

For any $\phi\in {\rm C}^\omega(M)$, we let $d\phi\in T^*M$ be a one-form defined as follows:  Let $d\phi_x\in T^*_x M$ be the evaluation of $d\phi$ at~$x$. Then, for any $f\in \mathfrak{X}(M)$, we have that $d\phi_x(f(x)) = (f\phi)(x)$. For two vector fields $f_i, f_j\in \mathfrak{X}(M)$, we let $[f_i,f_j]$ be the Lie bracket, which is defined such that 
$[f_i, f_j] \phi = f_i f_j \phi  - f_j f_i \phi$  for all $\phi\in {\rm C}^\omega(M)$.  

Let $\{f_i\}^m_{i = 1}$ be a subset of $\mathfrak{X}(M)$. Let ${\rm w} = w_1\cdots w_k$ be a {\em word} over the alphabet $\{1,\ldots, m\}$ of length~$k$. For a function $\phi\in {\rm C}^\omega(M)$, we define $f_{\rm w}\phi := f_{w_1}\cdots f_{w_k}\phi$. If ${\rm w} = \varnothing$, i.e., an empty word (of zero length), then we set $f_{\rm w}\phi := \phi$.

Let $\eta: M\to N$ be a diffeomorphism. Denote by $\eta_*: TM\to TN$ the derivative of~$\eta$. For a vector field $f\in \mathfrak{X}(M)$, let $\eta_*f\in \mathfrak{X}(N)$ be the {\em pushforward} defined as 
$(\eta_*f)(y) := \eta_*(f(\eta^{-1}y))$ for all $y\in N$.  
For a function $\phi\in {\rm C}^\omega(N)$, let $\eta^*\phi \in {\rm C}^\omega(M)$ be the {\em pullback} defined as $(\eta^*\phi)(x) := \phi(\eta(x))$ for all $x\in M$.

{\em 2. Whitney topologies.}  Let $M$ be equipped with a Riemannian metric. Denote by $\operatorname{d}_M(x_1, x_2)$ the distance between two points $x_1$ and $x_2$ in~$M$. 
Let $\Sigma$ be an analytic, compact manifold, and ${\rm C}^\omega(\Sigma, M)$ be the space of analytic functions from $\Sigma$ to~$M$. The {\em Whitney ${\rm C}^0$-topology} on ${\rm C}^\omega(\Sigma, M)$ can be defined by a basis of open sets: First, recall that a profile $x_\Sigma= \{x_\sigma \mid \sigma\in \Sigma \}$ can be viewed as a function from $\Sigma$ to~$M$. Given a profile $x_\Sigma$ and a positive number $\epsilon$, we define an open set as follows:   
$$
\{\bar x_\Sigma \in {\rm C}^\omega(\Sigma, M) \mid  \operatorname{d}_M( x_\sigma, \bar x_\sigma) <  \epsilon, \quad \forall \sigma\in \Sigma \}.
$$
Then, a basis of open sets can be obtained by letting $x_\Sigma$  vary over ${\rm C}^\omega(\Sigma, M)$ and letting~$\epsilon$ vary over the set of all positive real numbers. Generally, one can also define the Whitney ${\rm C}^k$-topology for $0\le k\le \infty$; for that, one needs to introduce the notion of jet space. We omit here the details and refer the reader to~\cite[Ch.~2-Sec.~2]{golubitsky2012stable}.

{\em 3. Algebra of functions.} Let $\Sigma$ be an analytic, compact manifold and $\{\rho_s\}^r_{s = 1}$  
be a set of real-valued functions  on~$\Sigma$. %We call $\{\rho_s\}^r_{s = 1} $ a {\em separating set} if for any two distinct points $\sigma, \sigma' \in \Sigma$, there exists a function $\rho_s$ out of the set such that $\rho_s(\sigma) \neq \rho_s(\sigma')$. %A function $\rho_s$ is everywhere nonzero if $\rho_s(\sigma) \neq 0$ for all $\sigma\in \Sigma$. 
For any $k \ge 0$, let $\rho^k_s(\sigma) := \rho_s(\sigma)^k$. Note, in particular, that $\rho^0_s ={\bf 1}_\Sigma$. If $\rho_s$ is everywhere nonzero, then $\rho^k_s$ is defined for all $k\in \mathbb{Z}$. 
We call  $\prod^r_{s = 1}\rho^{k_s}_s$, for $k_s \ge 0$, a {\em monomial}. Its degree is define by $k:= \sum^r_{s = 1}k_s$. Let ${\cal P}$ be the collection of all monomials. We decompose ${\cal P} = \sqcup_{k \ge 0} {\cal P}(k)$, where ${\cal P}(k)$ is comprised of monomials of degree~$k$.   
Denote by ${\cal S}$ the subalgebra generated by the set of functions $\{\rho_s\}^r_{s = 1}$. It is defined such that if $\rp\in {\cal S}$, then~$\rp$ can be expressed as a linear combination of a finite number of monomials with real coefficients.

 \subsection{Lie groups, Lie algebras, and representations}
 
{\em 1. Lie groups and Lie algebras.} Let $G$ be a Lie group with $e$ the identity element. Let~$\g$ be the associated Lie algebra, and $[\cdot, \cdot]$ be the Lie bracket. We identify each element $X\in \g$ with a {\em left-invariant vector field} $L_X$ over~$G$, i.e., $L_X(g) = gX$ for any $g\in G$. Thus, $L_{[X, Y]} = [L_X, L_Y]$.  
Note that to each $X\in \g$, there also corresponds a right-invariant vector field $R_X$. For any $X, Y\in \g$, we have $R_{[X, Y]} = - [R_X, R_Y]$.

 A {\em subalgebra} $\h$ of $\g$ is a vector subspace closed under Lie bracket, i.e., $[\h, \h] \subseteq \h$. An {\em ideal} $\mathfrak{i}$ of $\g$ is a subalgebra such that $[\mathfrak{i}, \g] \subseteq \mathfrak{i}$. We say that $\g$ is {\em simple}  if it is non-abelian, and moreover, the only ideals of~$\g$ are~$0$ and itself. A {\em semi-simple} Lie algebra is a direct sum of simple Lie algebras. A {\em Cartan subalgebra} $\h$ of $\g$ is maximal among the abelian subalgebras $\h'$ of~$\g$ such that the adjoint representation $\ad(X):= [\cdot, X]$ are simultaneously diagonalizable (over $\C$) for all $X\in \h'$. 
 
Simple real Lie algebras have been completely classified (up to isomorphism) by \'Elie Cartan. One can assign to each simple real Lie algebra a  Vogan diagram~\cite[Ch.~VI]{knapp2013lie} or a Satake diagram~\cite{satake1960representations,araki1962root}, depending on whether a  maximally compact or a maximally non-compact Cartan subalgebra is used.  %for classification. %We refer the reader to for details. 
A few commonly seen examples include special unitary Lie algebra $\mathfrak{su}(n)$, special linear Lie algebra $\sl(n, \R)$,  
special orthogonal Lie algebra $\so(n)$, symplectic Lie algebra $\mathfrak{sp}(2n, \R)$, indefinite special orthogonal Lie algebra $\so(p,q)$ (e.g. $\so(1,3)$ is the Lie algebra of the Lorentz group $\operatorname{O}(1,3)$). A complete list of (non-complex) simple real Lie algebras can be found in~\cite[Thm.~6.105]{knapp2013lie}. More details about the structure theory of semi-simple real Lie algebras will be provided along the presentation of the paper whenever needed.

{\em 2. Lie group and Lie algebra representation.} 
%Let $G$ be a Lie group and $\g$ be the Lie algebra. We denote by $e$ the identity element of $G$. 
Let $V$ be a finite dimensional vector space over $\R$. Let $\operatorname{Aut}(V)$ and $\operatorname{End}(V)$ be the sets of automorphisms and endomorphisms of~$V$, respectively. 
%If a basis of $V$ is chosen, then each automorphism/endomorphism of $V$ can be realized as a square (real) matrix.  
A {\em representation} $\pi$ of $G$ on $V$, is a group homomorphism $\pi:G \to \operatorname{Aut}(V)$, i.e., $\pi(e)$ is the identity map and $\pi(gh) = \pi(g)\pi(h)$.    

Let $\langle \cdot, \cdot\rangle$ be an inner-product on $V$. 
 We say that the representation~$\pi$ is ${\rm C}^k$ (i.e., $k$th continuously differentiable) if the map $\pi: (g, v) \in G\times V \mapsto \pi(g)v\in V$ is~${\rm C}^k$. A {\em matrix coefficient}  is any ${\rm C}^k$-function on $G$ defined as
 $ \langle v_i, \pi(g)v_j \rangle$ where $v_i, v_j$ belong to~$V$.  In particular, if the $v_i$'s form an orthonormal basis of $V$, then 
 $ \langle v_i, \pi(g)v_j  \rangle$ is exactly the $ij$-th entry of the matrix $\pi(g)$ with respect to the given basis.

 A group representation $\pi$ induces a Lie algebra homomorphism $\pi_*: \g \to \operatorname{End}(V)$, where $\pi_*$ is the derivative of $\pi$ at the identity~$e\in G$. It satisfies the following condition:  $$\pi_*([X, Y]) = \pi_*(Y)\pi_*(X) - \pi_*(X)\pi_*(Y),\quad \forall X, Y\in \g.$$ 
 %defined by $\pi_{*}(X) := \lim_{t\to 0 } \nicefrac{\pi(\exp(Xt)) - \pi(e)}{t}$.  
 We call $\pi_*$ a {\em representation} of $\g$ on $V$, or simply a Lie algebra representation.   
  
Let $\Ad: G\to \operatorname{Aut}(\g)$ be the {\em adjoint representation}, i.e.,  
 %$$
 %\Ad(g)(X):= \lim_{t\to 0} \frac{ g \exp(Xt)g^{-1} - e}{t},
 %$$
 for each $g\in G$,  $\Ad(g): T_e G\to T_e G$ is the derivative of the conjugation $h\in G \mapsto ghg^{-1}\in G$ at the identity~$e$. 
Denote by $\ad:\g\to \operatorname{End}(\g)$ the induced Lie algebra representation of~$\Ad$, which is given by $\ad(X)(\cdot) = [\cdot, X]$ for all $X\in \g$.  

%If $G$ happens to be a matrix Lie group embedded in $\R^{n\times n}$, then the {\em standard representation} of~$G$ on $\R^n$ is simply given by $(g, v)\in G\times \R^n \mapsto gv\in \R^n$. 

 {\em 3. Lie products.} Let $A:= \{X_1,\ldots, X_k\}$ be a set of free generators, and $\mathfrak{L}(A)$ be the associated free Lie algebra. For a Lie product $\xi\in \mathfrak{L}(A)$, let  $\dep(\xi)$ be the {\em depth} of $\xi$ defined as the number Lie brackets in~$\xi$. For example, the depth of $[X_{i_1}, [X_{i_2},X_{i_3}]]$ is~$2$. We denote by ${\cal L}_A$ the collection of Lie products of the~$X_i$'s in $A$.   
% The {\em P. Hall basis}~\cite{serre2009lie} of $\mathfrak{L}(A)$ is a sequence of Lie products ${\cal L}_A := \{\xi_i\}^\infty_{i = 1}$ which satisfies the following three conditions: {\em (i).} The first $k$ Lie products are the $X_i$'s, i.e., ${\xi_i} = X_i$ for all $i = 1,\ldots, k$; {\em (ii).} If $\dep(\xi_i) < \dep(\xi_j)$, then $i < j$; {\em (iii).} Each $[\xi_i, \xi_j]$ belongs to ${\cal L}_A$ if and only if $\xi_i, \xi_j\in {\cal L}_A$ with $i < j$, and either $\xi_j = X_{j'}$ for some~$j'$ or $\xi_j = [\xi_l,\xi_r]$ with $\xi_l, \xi_r\in {\cal L}_A$ and $l \le i$. For example~\cite{murray1993steering}, if $k = 3$, then the Lie products $\xi$ in the P. Hall basis with $\dep(\xi) \le 2$ are given by
% $$
% \begin{array}{l}
% \,X_1 \quad X_2 \quad X_3 \quad [X_1, X_2] \quad [X_1, X_3] \quad [X_2, X_3] \\
% \,[X_1, [X_1, X_2]] \quad [X_1, [X_1, X_3]] \quad [X_2, [X_1, X_2]] \quad [X_2, [X_1, X_3]] \\
% \, [X_2, [X_2, X_3]] \quad  [X_3, [X_1, X_2]] \quad [X_3, [X_1, X_3]] \quad [X_3, [X_2, X_3]]. 
% \end{array}
% $$
We further decompose ${\cal L}_A =\sqcup_{k \ge 0} {\cal L}_A(k)$ where
 ${\cal L}_A(k)$ is comprised of Lie products of depth~$k$. %An application of P.~Hall basis  is in nonholonomic motion planning~\cite{murray1993steering,sussmann1993lie} via the technical of Lie extension. %of nonholonomic system via the approach of Lie extension (see, for example,). We review such an approach in Subsection~\S\ref{ssec:paec}.

 \subsection{Miscellaneous}
Let $\{e_1,\ldots, e_n\}$ be the standard basis of~$\R^n$. We denote by $\det(e_{i_1},\ldots, e_{i_n})$ the determinant of a matrix whose $j$-th column is $e_{i_j}$ for $i_j \in \{1,\ldots, n\}$. 

 For a vector $v = (v_1,\ldots, v_n)\in \R^n$, we let $\|v\|_1:= \sum^n_{i = 1}|v_i|$ be the one-norm of~$v$.

 Let $V$ be a vector space over $\R$. We denote by $V^*$ the dual space, i.e., it is the collection of all linear functions from~$V$ to~$\R$.
 
 Let  $V'$ and~$V''$ be two subsets (not necessarily subspaces) of the vector space~$V$. The two subsets $V'$ and $V'$ are said to be {\bf projectively identical} if for any $v'\in V'$, there exists a $v''\in V''$ and a constant $c\in \R$ such that $v' = cv''$, and vice versa. We write $V' \equiv V''$ to indicate such equivalence relation. 

Let $S$ be an arbitrary set with an operation ``$*$'' defined so that $s_1*s_2$ belongs to~$S$ for all $s_1, s_2\in S$. For any two subsets $S'$ and $S''$ of $S$, we let $S'*S''$ be the subset of $S$ comprised of the elements $s'*s''$ for all $s'\in S'$ and $s''\in S''$. 
Here are two examples in which such a notation will be used:  {\em (i)} If $S$ is a vector space and ``$*$'' is the addition ``$+$'', then we write $S' + S''$.  %{\em (ii)} If $W = \g$ is a Lie algebra and ``$*$'' is the Lie bracket ``$[\cdot, \cdot]$'', then we write $[W', W'']$; 
{\em (ii)} If $S$ is the commutative algebra of analytic functions ${\rm C}^\omega(\Sigma)$ and ``$*$'' is the pointwise multiplication, then we simply write $S' S''$.  

However, we note that the above notation does not apply to $[\g_1, \g_2]$ for $\g_1$ and $\g_2$ two subsets of a Lie algebra $\g$. By convention, $[\g_1,\g_2]$ is the {\em linear span} of all $[X_1, X_2]$ with $X_1\in\g_1$ and $X_2\in \g_2$. We adopt such a convention in the paper as well.  

For $\g$ a Lie algebra, we let $[\cdot, \cdot]$ be the Lie bracket associated with~$\g$. If $\g$ is comprised of matrices, then, to avoid confusion,  we denote by $[\cdot, \cdot]_{\rm M}$ the matrix commutator, which differs from $[\cdot, \cdot]$ by a negative sign, i.e., $[X, Y] = - [X, Y]_{\rm M}$.

For a general control system $\dot x(t) = f(x(t), u(t))$, we denote by $u[0,T]$ the control input $u(t)$ for the time interval $[0,T]$ for $T > 0$, and correspondingly, $x[0,T]$ the trajectory generated by the system over  $[0,T]$.

\section{Distinguished Ensemble Systems}\label{sec:controlobserverdual}

\subsection{Distinguished vector fields and codistinguished functions}\label{ssec:keydefinition}

We introduce in the section the class of (pre-)distinguished ensemble systems and establish
controllability and observability of any such ensemble system. We start by introducing two key components of the system, namely distinguished vector fields and codistinguished functions. We first have the following definition:  

\begin{definition}[Distinguished vector fields]\label{def:distvecf}
A set of vector fields  $\{f_i\}^m_{i = 1}$ over an analytic manifold~$M$ is {\bf distinguished}  if the following hold: 
\begin{enumerate}
\item[\em 1)] For any $x\in M$,  the set $\{f_i(x)\}^m_{i = 1}$ spans $T_x M$.  
\item[\em 2)] For any two $f_i$ and $f_j$, there exist an $f_k$ and a real number $\lambda$ such that
\begin{equation}\label{eq:distvecf}
[f_i, f_j] = \lambda f_k;
\end{equation}
conversely, for any $f_k$, there exist $f_i$ and $f_j$ and a {\em nonzero} $\lambda$ such that~\eqref{eq:distvecf} holds.
\end{enumerate} \,
\end{definition}

Recall that $\mathfrak{X}(M)$ is the Lie algebra of analytic vector fields over~$M$, which is infinite dimensional.  
However, if $F:= \{f_i\}^m_{i = 1}$ is distinguished, then by item~2 of Def.~\ref{def:distvecf}, the $\R$-span of the~$f_i$'s, which we denote by~$\mathbb{L}_F$,   
is a {\em finite dimensional} subalgebra of $\mathfrak{X}(M)$. We note here that $\mathbb{L}_F$ is {\em perfect}, i.e., $[\mathbb{L}_F, \mathbb{L}_F] =  \mathbb{L}_F$.  
%We denote it by 
%$$
%F := \left \{\sum^m_{i = 1} c_i f_i  \mid c_i \in \R\right \}.
%$$  

Let $N$ be any manifold diffeomorphic to~$M$, and $\eta: M\to N$ be the diffeomorphism. Recall that for a vector field~$f$ over $M$, we denote by $\eta_* f$ the pushforward of~$f$ as a vector field over $N$. We have the following fact: 

\begin{lemma}\label{lem:topologicalinv}  
If $\{f_i\}^m_{i = 1}$ is distinguished over~$M$, then $\{\eta_* f_i\}^m_{i = 1}$ is distinguished over~$N$.  
\end{lemma}

\begin{proof}If $[f_i,f_j] = \lambda f_k$, then 
%It simply follows from the fact that the pushforward  commutes with Lie bracket, i.e., 
$[\eta_* f_i, \eta_* f_j] = \eta_* [f_i,f_j] = \lambda  \eta_* f_k$.  
\end{proof}

We next introduce the definition of codistinguished functions:   

\begin{definition}[Codistinguished functions]\label{def:codisth} 
 A set of functions $\{\phi^j\}^l_{j = 1}$ on $M$ is {\bf codistinguished} to a set of vector fields $\{f_i\}^m_{i = 1}$ if the following hold:   
\begin{enumerate}
\item[\em 1)] For any $x\in M$, the set of (exact) one-forms  $\{d\phi^j_x\}$ spans $T^*_x M$.
\item[\em 2)] For any $f_i$ and any $\phi^j$,  there exist a $\phi^k$ and a real number $\lambda$ such that
\begin{equation}\label{eq:codistfunc}
f_i  \phi^j  = \lambda \phi^k;
\end{equation}
conversely, for any $\phi^k$, there exist $f_i$, $\phi^j$, and a {\em nonzero} $\lambda$ such that~\eqref{eq:codistfunc} holds.
\item[\em 3)] For $x, x'\in M$, if $\phi^j(x) = \phi^j(x')$ for all $j = 1,\ldots, l$, then $x = x'$.
\end{enumerate}
If $\{\phi^j\}^l_{j = 1}$ satisfies only {\em 1)} and {\em 2)}, then it is {\bf weakly codistinguished} to~$\{f_i\}^m_{i = 1}$. 
\end{definition}

%Note that in the above definition, we do {\em not} require the set of vector fields $\{f_i\}^m_{i = 1}$ to be distinguished. We also have the following remark:
%
%\begin{remark}
%By the Cartan's magic formula, $f d\phi = d (f\phi)$ for any~$f\in \mathfrak{X}(M)$ and $\phi\in {\rm C}^\omega(M)$.
%In particular, if $f_i\phi^j = \lambda\phi^k$, then 
%$
%f_i d\phi^j= d (f_i \phi^j) = \lambda d \phi^k 
%$.     
%Thus, we can say that the set of one-forms $\{d\phi^j\}^l_{j = 1}$ is {\em (weakly) codistinguished} to $\{f_i\}^m_{i= 1}$.  
%\end{remark}

Let $\tilde\eta: N\to M$ be a diffeomorphism. Recall that for a function $\phi$ on~$M$, we denote by~$\tilde \eta^* \phi$ the pullback of~$\phi$ as a function on~$N$. We have the following fact:

\begin{lemma}\label{lem:diffeocodist}
If $\{\phi^j\}^l_{j = 1}$ on~$M$ is codistinguished to $\{f_i\}^m_{i = 1}$, then $\{\tilde\eta^*\phi^j\}^l_{j = 1}$ on~$N$ is codistinguished to $\{\tilde\eta_*^{-1}f_i\}^m_{i = 1}$.   
\end{lemma}

\begin{proof}
	If $f_i \phi^j = \lambda\phi^k$, then
$
(\tilde \eta_*^{-1} f_i)(\tilde\eta^*\phi^j) = \tilde\eta^* (f_i \phi^j) = \lambda  \tilde\eta^*\phi^k.
$
\end{proof}

We say that a set of vector fields $F:=\{f_i\}^m_{i = 1}$  and a set of functions $\Phi:=\{\phi^j\}^l_{j = 1}$ are {\bf (weakly) jointly distinguished} if $F$ is distinguished and $\Phi$ is (weakly) codistinguished to~$F$. 
Note that Lemmas~\ref{lem:topologicalinv} and~\ref{lem:diffeocodist} imply that the property of having a set of (weakly) jointly distinguished pair $(F, \Phi)$ is topologically invariant. Let $F$ and $\Phi$ be (weakly) disjoined. Recall that $\mathbb{L}_F$ is a finite-dimensional Lie algebra spanned by~$F$ (since $F$ is distinguished). Let $\mathbb{L}_\Phi$ be the $\R$-span of $\Phi$. Then, by the second item of Def.~\ref{def:codisth}, the following map: 
$$(f, \phi)\in \mathbb{L}_F\times \mathbb{L}_\Phi\mapsto f\phi\in \mathbb{L}_\Phi$$ 
is a finite dimensional Lie algebra representation of~$\mathbb{L}_F$ on~$\mathbb{L}_\Phi$. 

%\begin{definition}
%%Let $M$ be a smooth manifold,  $\{f_i\}^m_{i = 1}$ be vector fields over $M$, and $\{\phi^j\}^l_{j = 1}$ be functions on $M$. 
%\end{definition}

%Now, let $F:= \{f_i\}^m_{i = 1}$ and $\Phi:=\{\phi^j\}^l_{j = 1}$ be (weakly) jointly distinguished. Recall that $\mathbb{L}_F$ is the $\R$-span of~$F$. Similarly, we let $\mathbb{L}_\Phi$ be the $\R$-span of $\Phi$, which is a finite dimensional subspace of ${\rm C}^\omega(M)$.   
%%$$
%%\mathbb{L}_\Phi:= \left \{\sum^l_{j = 1} c_j \phi^j  \mid c_j \in \R\right \}.
%%$$
%We now have the following fact:  
%%$$
%%\Phi:= \left \{\sum^l_{j = 1} c_j \phi^j  \mid c_j \in \R\right \}.
%%$$
%
%
%
%\begin{lemma}\label{lem:representation} 
%Let $F:=\{f_i\}^m_{i =1}$ and $\Phi:=\{\phi^j\}^l_{j = 1}$ be (weakly) jointly distinguished. 
%Then, the following action:  $$(f, \phi)\in \mathbb{L}_F\times \mathbb{L}_\Phi \mapsto f \phi\in \mathbb{L}_\Phi$$   
%is a Lie algebra representation of~$\mathbb{L}_F$ on~$\mathbb{L}_\Phi$. 
%\end{lemma}
%
%\begin{proof}
%It follows from the fact that for any two $f_i, f_{i'}$ out of $\{f_i\}^m_{i = 1}$ and any~$\phi^j$ out of $\{\phi^j\}^l_{j = 1}$, we have $[f_i, f_{i'}] \phi = f_i  f_{i'} \phi^j - f_{i'}f_i \phi^j \in \mathbb{L}_\Phi$.  \end{proof}

%We will use this fact in Section~\S\ref{sec:onliegroup} to generate co-distinguished functions. 

For the remainder of the subsection, we provide an example about jointly distinguished vector fields and functions on $\SO(3)$. These vector fields and functions will be further generalized in Section~\S\ref{sec:onliegroup} so that they exist on {\em any}  semi-simple Lie group.

\begin{example}\label{emp:example1}
{\em 
Let $\SO(3)$ be the matrix Lie group of $3\times 3$ special orthogonal matrices, and $\so(3)$ be the associated Lie algebra. We define a basis $\{X_i\}^3_{i = 1}$ of $\so(3)$ as follows: 
%\begin{equation}\label{eq:f123SO3}
%L_{X_1}(g) := gX_{1}, \quad L_{X_2}(g) := g X_{2}, \quad L_{X_3}(g) := gX_{3}.   
%\end{equation}
$$ X_i:= e_je^\top_k - e_k e_j^\top  \,\,\mbox{ where } \det(e_i,e_j,e_k) = 1, \quad \forall i = 1,2,3.$$ 
Let $\{L_{X_i}\}^3_{i = 1}$ be the corresponding left-invariant vector fields.  By computation, 
\begin{equation}\label{eq:exampleijk}
[L_{X_i}, L_{X_j}] = \det(e_i,e_j,e_k)  L_{X_k}, \quad \forall i\neq j. 
\end{equation}  
It follows that $\{L_{X_i}\}^3_{i = 1}$ is distinguished.

Denote by $\tr(\cdot)$ the trace of a square matrix. We next define functions $\{\phi^{ij}\}^3_{i,j= 1}$ on $\SO(3)$ as follows:  
$$
\phi^{ij}(g) := \operatorname{tr}(g X_j g^\top X_i^\top), \quad 1\le i, j \le 3.$$  
 We show below that $\{\phi^{ij}\}^3_{i, j = 1}$ is codistinguished to $\{L_{X_i}\}^3_{i = 1}$. 
First, for any left-invariant vector field $L_X$ with $X\in \so(3)$, we obtain by computation that   
\begin{equation}\label{eq:sssssss}
d\phi^{ij}_g(L_X(g)) = (L_X \phi^{ij})(g)  = \tr(g[X_j, X]g^\top X_i^\top ).
\end{equation}
%Note that the Lie bracket $[\cdot, \cdot]$ differs from the matrix commutator $[\cdot, \cdot]_{\rm M}$ by a negative sign.
We now prove that the three items of Def.~\ref{def:codisth} are satisfied for $\{\phi^{ij}\}^3_{i, j = 1}$ and $\{L_{X_i}\}^3_{i = 1}$:    
\begin{enumerate}
\item[\em 1)] 
We fix an arbitrary group element $g\in \SO(3)$ and show that $\{d\phi_g^{ij}\}^3_{i, j = 1}$ spans $T^*_g\SO(3)$. For convenience, let $\hat X_{ij}:=[g^\top X_i^\top g, X_j]$.  Then, by~\eqref{eq:sssssss}, we obtain that
$$
d\phi^{ij}_g(L_X(g)) = \tr(X [g^\top X_i^\top g, X_j]) =  \tr(X \hat X_{ij}).
$$
Note that $\tr(\cdot, \cdot)$ is negative definite on $\so(3)$. Thus, $\{d\phi_g^{ij}\}^3_{i, j = 1}$ spans $T^*_g\SO(3)$ if and only if $\{\hat X_{ij}\}^3_{i,j = 1}$ spans $\so(3)$. It now suffices to show that $\{\hat X_{ij}\}^3_{i,j = 1}$ spans $\so(3)$. But, this holds because  both $\{X_j\}^3_{j = 1}$ and $\{g^\top X_i^\top g\}^3_{i = 1}$ span $\so(3)$ and, moreover, $\so(3)$ is simple so that $[\so(3), \so(3)] =\so(3)$.

%Since $\tr(\cdot, \cdot)$ is negative definite on $\so(3)$, $\{d\phi^{ij}_g\}^3_{i,j = 1}$ spans $T^*_g\SO(3)$ for all $g\in \SO(3)$. %Thus, the first item of Def.~\ref{def:codisth} is satisfied.

\item[\em 2)] For the second item, we combine~\eqref{eq:exampleijk} and~\eqref{eq:sssssss} to obtain the following:   
%$
%(L_{X_i}\phi^{i'j})(g)  =  \tr(g[X_{j}, X_i]g^\top X_{i'}^\top).
%$
%It then follows from that 
$$
L_{X_i}\phi^{i'j}=
\left\{
\begin{array}{ll}
 -\det(e_i,e_j,e_k)\phi^{i'k}, & \mbox{ if } i \neq j, \\
 0, & \mbox{ otherwise.}
 \end{array}
 \right. 
$$
%and $L_{X_i}\phi^{i'j} = 0$ if $i = j$. 
\item[\em 3)] Finally, let $g$ and $g'$ be such that $\phi^{ij}(g) = \phi^{ij}(g')$ for all $1 \le i,j \le 3$:  
$$
\tr(g X_j g^\top X_i^\top) = \tr(g' X_j g'^\top X_i^\top), \quad \forall i = 1,2,3.
$$
Because $\{X_i\}^3_{i =1}$ spans $\so(3)$ and $\tr(\cdot, \cdot)$ is negative definite on $\so(3)$, we have that
$
g X_j g^\top = g' X_j g'^\top  
$. Since this holds for all $X_j$, it follows that $g^{\top}g'$ belongs to the center of $\SO(3)$. But the center is trivial. We thus conclude that $g= g'$.   
\end{enumerate}
}
\end{example}

\subsection{Controllability and observability of distinguished ensemble system}\label{ssec:eco}

We establish in the subsection a sufficient condition for controllability and observability of the ensemble system~\eqref{eq:ensemblesystem0}. For convenience, we reproduce below the mathematical model of the ensemble system introduced in Section~\S\ref{sec:introduction}:
\begin{equation}\label{eq:ensemblesystem}
\left\{
\begin{array}{lr}
\dot x_\sigma(t) = f_0(x_\sigma(t), \sigma) + \sum^m_{i = 1} \sum^{r}_{s = 1} u_{i, s}(t)\rho_{s}(\sigma) f_i (x_\sigma(t)),  & \forall \sigma\in \Sigma, 
\vspace{3pt}\\
 y^j(t) =  \displaystyle\int_\Sigma \phi^j(x_\sigma(t)) d\sigma,  &  \quad \forall j = 1,\ldots, l.\\
 \end{array}
 \right. 
\end{equation}
We recall that the common state space~$M$ is an analytic manifold, equipped with a Riemannian metric~$\operatorname{d}_M(\cdot,\cdot)$. The parametrization space~$\Sigma$ is analytic, compact, and path-connected. It is also equipped with a positive measure. All vector fields and parametrization functions are analytic. %The drifting vector field~$f_0$ belongs to ${\rm C}^\omega(M\times \Sigma)$. The parametrization functions $\{\rho_{i,s}\}^r_{s =1}$ and the control vector fields $\{f_i\}^m_{i = 1}$'s belong to~${\rm C}^\omega(\Sigma)$ and $\mathfrak{X}(M)$, respectively.  The functions $\{\phi^j\}^l_{j = 1}$ belong to~${\rm C}^\omega(M)$. 
The control inputs $u_{i,s}(t)$ are integrable functions over any time interval $[0,T]$ for $T > 0$. We denote by~$u(t)$ the collection of the $u_{i,s}(t)$ and $y(t)$ the collection of the $y^j(t)$.

\begin{wrapfigure}[13]{r}{7.5cm}\vspace{-0.40cm}
\centering
\includegraphics[width=.42\textwidth]{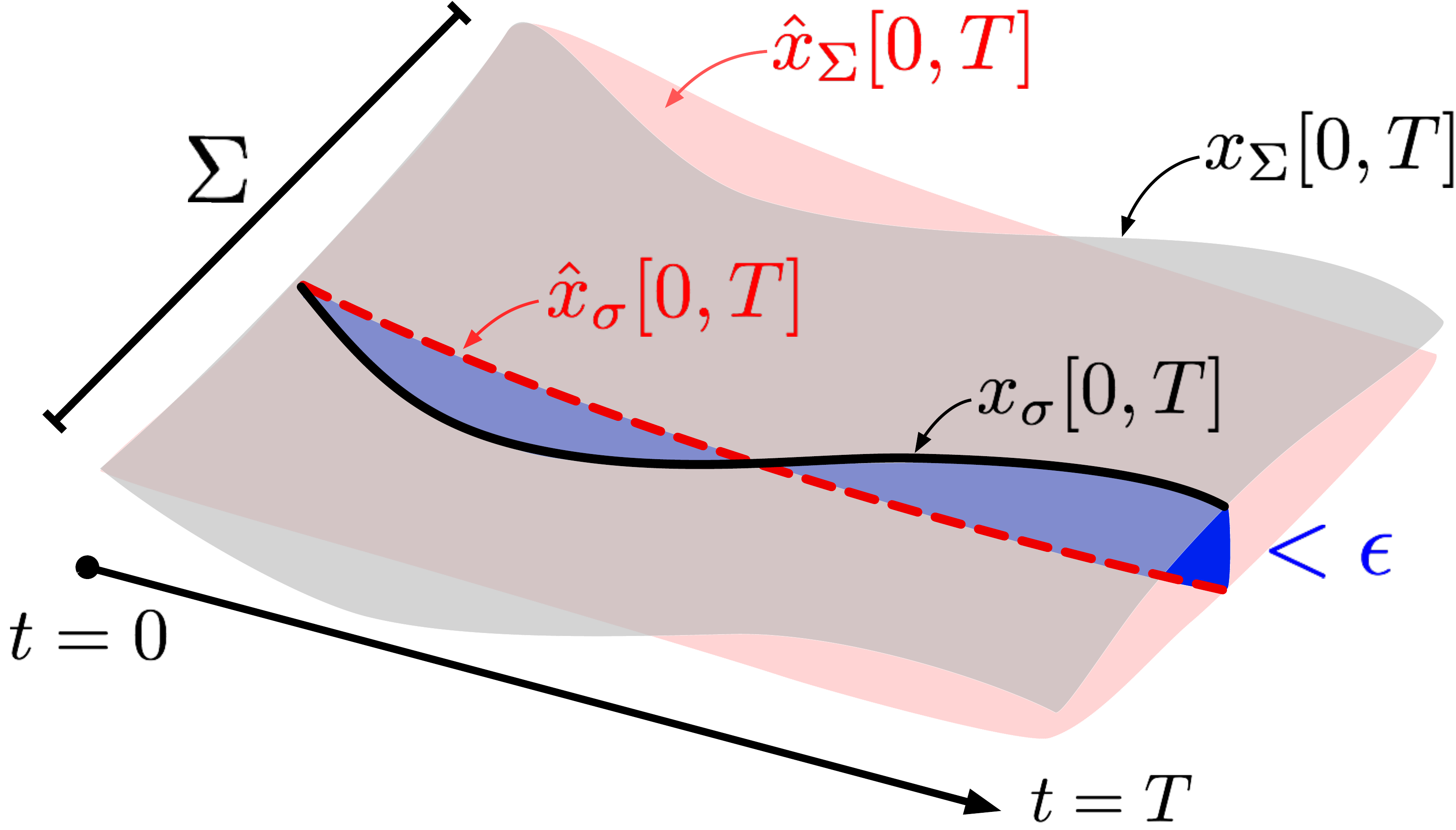}
\caption{The red surface $\hat x_\Sigma[0,T]$ is a trajectory of profiles we want the system to follow. The grey surface  $x_\Sigma[0,T]$, within an $\epsilon$-tubular neighborhood of $\hat x_\Sigma[0,T]$, is generated by a control input. %The red dashed (resp. black solid) curve is $\hat x_\sigma[0,T]$ (resp. $x_\sigma[0,T]$) associated with the individual system-$\sigma$.
}\label{fig:controllable}
\end{wrapfigure}

We also recall that $x_\Sigma(t)$ is the profile of system~\eqref{eq:ensemblesystem} at time~$t$, 
and ${\rm C}^\omega(\Sigma, M)$ is the profile space equipped with the Whitney ${\rm C}^0$-topology. Further, we let $$x_\Sigma[0,T]:=\{x_\sigma[0,T] \mid \sigma\in \Sigma\}.$$ 
We call $x_\Sigma[0,T]$ a {\em trajectory of profiles}.  
We introduce below precise definitions of  ensemble controllability and observability of system~\eqref{eq:ensemblesystem}.  

We first have the definition of ensemble controllability. An illustration of the definition is provided in Fig.~\ref{fig:controllable}.

\begin{definition}[Ensemble controllability]\label{def:ensemblecontrollability}
%System~\eqref{eq:ensemblesystem} is {\bf approximately ensemble controllable} 
%if for {\em 1)} any initial profile $x_\Sigma(0)\in {\rm C}^\omega(\Sigma, M)$ and any target profile $\hat x_\Sigma\in {\rm C}^\omega(\Sigma, M)$, {\em 2)} any time $T > 0$, and {\em 3)} any error tolerance $\epsilon > 0$, there exists an integrable function $u(t)$ as a control input such that the solution $x_\sigma(t)$ of~\eqref{eq:ensemblesystem} satisfies 
%\begin{equation*}\label{eq:defcontrollability}
%{\rm d}_M(x_\sigma(T), \hat x_\sigma) < \epsilon, \quad \forall \sigma\in \Sigma.
%\end{equation*}
System~\eqref{eq:ensemblesystem} is {\bf approximately ensemble path-controllable} if for any initial profile $x_\Sigma(0)$, any target trajectory of profiles $\hat x_\Sigma[0,T]$ with $\hat x_\Sigma(0) = x_\Sigma(0)$, and any error tolerance $\epsilon > 0$, there is a control input $u(t)$ such that the trajectory $x_\sigma[0,T]$ generated by the control input satisfies  
$$
{\rm d}_M(x_\sigma(t), \hat x_\sigma(t)) < \epsilon, \quad \forall (t, \sigma)\in [0,T]\times \Sigma.
$$ 
\end{definition}

%Note that since $\Sigma$  is compact and ${\rm d}_M(x_\sigma(T), \hat x_\sigma)$ is continuous in $\sigma$, condition~\eqref{eq:defcontrollability} is equivalent with the ${\rm L}^p$-norm:  
%$$
%\left ( \int_\Sigma \left ( {\rm d}_M(x_\sigma(T), \hat x_\sigma) \right )^p d\sigma \right )^{1/p}< \epsilon, \quad \forall 1\le p\le \infty
%$$  
%We now introduce the definition of ensemble observability:  

We next introduce the definition of ensemble observability. To proceed, we first have the following one: 

\begin{definition}[Output equivalence] Two initial profiles $x_\Sigma(0)$ and $\bar x_\Sigma(0)$ of system~\eqref{eq:ensemblesystem} are {\bf output equivalent}, which we denote by $x_\Sigma(0)\sim \bar x_\Sigma(0)$,  if for any $T> 0$ and any integrable function $u:[0,T]\to \R^m$ as a control input, the following hold:  
$$
\int_\Sigma \phi^{j}(x_\sigma(t)) d\sigma = \int_\Sigma \phi^{j}(\bar x_\sigma(t)) d\sigma,
$$
for all $t\in [0,T]$ and for all $j = 1,\ldots, l$.
\end{definition}

\begin{wrapfigure}[13]{r}{7.5cm}\vspace{-0.0cm}
\centering
\includegraphics[width=.41\textwidth]{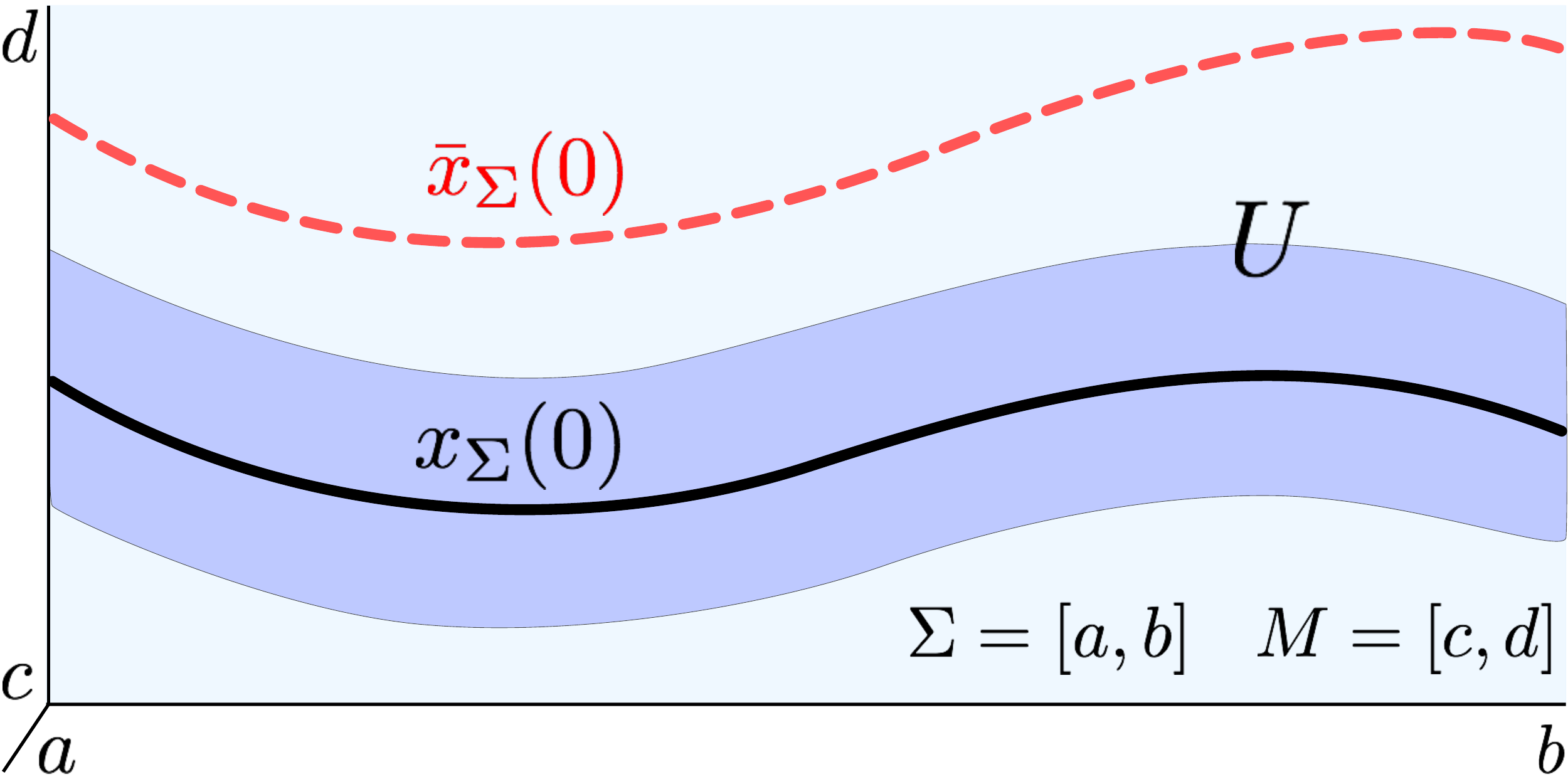}
\caption{Let $\Sigma := [a, b]$ and $M := [c, d]$ be closed intervals. Weak ensemble observability requires that there is a neighborhood~$U$ of~$x_\Sigma(0)$ such that if  $\bar x_\Sigma(0)\sim x_\Sigma(0)$ and $\bar x_\Sigma(0) \neq x_\Sigma(0)$, then $\bar x_\Sigma(0)$ cannot intersect~$U$. %Ensemble observability requires that $U = [a, b]\times [c, d]$.
}\label{fig:observable}
\end{wrapfigure}

For a given $x_\Sigma(0)$, we let $O(x_\Sigma(0))$ be the collection of all initial profiles in ${\rm C}^{\omega}(\Sigma, M)$ that are output equivalent to $x_\Sigma(0)$, i.e., 
\begin{equation}\label{eq:outputequivalence}
O(x_\Sigma(0)):=  \left \{\bar x_\Sigma(0) \mid  \bar x_\Sigma(0)\sim x_\Sigma(0) \right \}.
\end{equation} 
The set $O(x_\Sigma(0))$ can be viewed as a ``measure of ambiguity'' for the ensemble estimation problem. Note that $O(x_\Sigma(0))$ always contains $x_\Sigma(0)$ itself. The set $O(x_\Sigma(0))$ could be finite, infinite but discrete, or even a continuum.  
With the above definition of output equivalence, we now introduce the definition of ensemble observability. An illustration of the definition is provided in Fig.~\ref{fig:observable}.

\begin{definition}[Ensemble observability]\label{def:ensembleobservability}
System~\eqref{eq:ensemblesystem} is {\bf weakly ensemble observable} if for any profile $x_\Sigma(0)\in {\rm C}^\omega(\Sigma, M)$, there exists an open neighborhood~$U$ of $x_\Sigma(0)$ in ${\rm C}^{\omega}(\Sigma,M)$ such that any profile $\bar x_\Sigma(0)$ that is output equivalent to $x_\Sigma(0)$ cannot intersect~$U$. 
Further, system~\eqref{eq:ensemblesystem} is {\bf ensemble observable} if for any profile $x_\Sigma(0)$, $O(x_\Sigma(0)) = \{x_\Sigma(0)\}$. 
\end{definition}

We establish below a sufficient condition for ensemble controllability and observability of system~\eqref{eq:ensemblesystem}. To state the condition, we need a few more preliminaries.  

First, we say that the set of parametrization functions $\{\rho_s\}^r_{s = 1}$ defined on $\Sigma$ is a {\bf separating set} if for any two distinct points $\sigma, \sigma' \in \Sigma$, there exists a function $\rho_s$, for some $s \in \{1,\ldots, r\}$, such that $\rho_s(\sigma) \neq \rho_s(\sigma')$. Note that by Stone-Weierstrass theorem~\cite[Chp.~7]{rudin1976principles}, if $\{\rho_s\}^r_{s = 1}$ separates point and contains an everywhere nonzero function, then the subalgebra generated by $\{\rho_s\}^r_{s = 1}$ is dense in the space ${\rm C}^0(\Sigma)$ of continuous functions on $\Sigma$. 

Next, for convenience, we let $\phi:= (\phi^1,\ldots, \phi^l)$ be a vector-valued function on~$M$. For a given~$x\in M$, we let $[x]_\phi$ be the pre-image of $\phi(x)$, i.e., $[x]_\phi$ is the collection of all points $x'$ in $M$ such that $\phi(x') = \phi(x)$. Note that if the set of one-forms $\{d\phi^j_x\}^l_{j = 1}$ spans $T^*_xM$ for all $x\in M$, then $[x]_\phi$ is a discrete set. Let $\chi_\phi$ be defined as follows:
$$
\chi_\phi:= \sup_{x\in M} \left |[x]_\phi \right |. 
$$  
If $\chi_\phi$ is unbounded, then we set $\chi_\phi :=\infty$. We have the following fact:

\begin{lemma}\label{lem:finiteambiguity}
If $M$ is compact and the one-forms $\{d\phi^j_x\}^l_{j = 1}$ span $T^*_xM$ for all $x\in M$, then $\chi_\phi$ is a finite number. 
\end{lemma}

\begin{proof}
First, note that for any $x\in M$, $|[x]_\phi|$ is a finite number. This holds because otherwise $[x]_\phi$ contains an accumulation point $x_*$ and, moreover, the one-forms $\{d\phi^j_{x_*}\}^l_{j = 1}$ {\em cannot} span $T^*_{x_*}M$, which is a contradiction. In fact, since the one-forms  $\{d\phi^j_x\}^l_{j = 1}$ span $T^*_xM$ for all $x\in M$, there is an open ball $B_{\epsilon(x)}(x)$ centered at $x$ with radius $\epsilon(x)$ such that $|[x']_\phi| = |[x]_\phi|$ for all $x'\in B_{\epsilon(x)}(x)$. Note that $\{B_{\epsilon(x)}(x)\}_{x\in M}$ is an open cover of $M$. Since $M$ is compact, there is a finite cover $\{B_{\epsilon(x_i)}(x_i)\}^N_{i = 1}$. It then follows that $\chi_\phi:= \max^N_{i = 1} |[x_i]_\phi |$.  	
\end{proof}

We are now in a position to state the first main result of the paper. The result establishes connections between the ``distinguished'' structure introduced in the previous subsection and ensemble controllability/observability of system~\eqref{eq:ensemblesystem}:   

\begin{theorem}\label{thm:mthm1}
Consider the ensemble system~\eqref{eq:ensemblesystem}. 
Suppose that $\{\rho_s\}^r_{s = 1}$ is a separating set and contains an everywhere nonzero function; then, the following hold: 
\begin{enumerate}
 \item[1)] If the set of control vector fields $\{f_i\}^m_{i = 1}$ is distinguished, then system~\eqref{eq:ensemblesystem} is approximately ensemble path-controllable.
 \item[2)] If the set of observation functions $\{\phi^j\}^l_{j = 1}$ is (weakly) codistinguished to $\{f_i\}^m_{i = 1}$, then system~\eqref{eq:ensemblesystem} is (weakly) ensemble observable. If, further, $M$ is compact, then for any $x_\Sigma(0)$, the set $O(x_\Sigma(0))$ defined in~\eqref{eq:outputequivalence} is finite and $|O(x_\Sigma(0))|\le \chi_\phi$.
 \end{enumerate}
\end{theorem}

Following the above theorem, we introduce the following definition: 

\begin{definition}\label{def:distinguishedensemblesys}
An ensemble system~\eqref{eq:ensemblesystem} is {\bf distinguished} if {\em 1)} the set  of parametrization functions $\{\rho_s\}^r_{s = 1}$ separates points and contains an everywhere nonzero function, and {\em 2)} the set of control vector fields $\{f_i\}^m_{i = 1}$  and the set   of observation functions $\{\phi^j\}^l_{j = 1}$ are (weakly) jointly distinguished.  
\end{definition} 

By Theorem~\ref{thm:mthm1}, a distinguished ensemble  system is approximately ensemble path-controllable and (weakly) ensemble observable. We provide below an example of a distinguished ensemble system:  
%We address briefly here the issue of designing a separating sets of functions $\{\rho_s\}^r_{s = 1}$. First, 

\begin{example}{\em  
Recall that in Example~\ref{emp:example1}, we have introduced jointly distinguished left-invariant vector fields $\{L_{X_i}\}^3_{i = 1}$ and functions $\{\tr(g X_j g^\top X^\top_i)\}^3_{i,j = 1}$ on $\SO(3)$. Now, consider a continuum ensemble of control systems defined on $\SO(3)$, parametrized by a scalar parameter $\sigma$ over a closed interval $[a, b]$ with $0 < a < b$. Let   $\rho(\sigma) := \sigma$ be the parametrization function. The singleton $\{\rho\}$ is a separating set and $\rho$ is everywhere nonzero. Thus, the following ensemble system is distinguished:   
$$
\left\{
\begin{array}{lr}
\dot g_\sigma(t) = f_0(g_\sigma(t), \sigma) + \sum^3_{i = 1}u_i(t)\sigma L_{X_i}(g_\sigma(t)), & \sigma\in [a, b], \vspace{3pt}\\
 y^{ij}(t) =  \displaystyle\int_\Sigma \tr(g_\sigma(t) X_j g_\sigma^\top(t) X^\top_i) d\sigma,    & 1\le i, j \le 3.\\
 \end{array}
 \right. 
$$ 
Thus, it is approximately ensemble path-controllable and ensemble observable.
} \end{example}

We have the following remark on the set of parametrization functions:  

\begin{remark}\label{rmk:existenceofrho1}{\em 
For any analytic manifold~$\Sigma$, there exists a set of separating set. By the Nash  embedding theorem~\cite{nash1966analyticity,greene1971analytic}, the manifold~$\Sigma$ can be isometrically embedded into a Euclidean space~$\R^N$. We write $\sigma = (\sigma_1,\ldots,\sigma_N)$ as the coordinate of a point $\sigma\in \Sigma$.         
Now, let $\rho_s(\sigma) := \sigma_s$, for $s = 1,\ldots, N$, be the standard coordinate functions (more precisely,  the restrictions of the coordinate functions to~$\Sigma$).  Further, let $\rho_{N + 1} := {\bf 1}_\Sigma$ be the unit function. Then, $\{\rho_s\}^{N + 1}_{s = 1}$ satisfies the assumption of Theorem~\ref{thm:mthm1}. }
%Of course, depending on the type of the parametrization space $\Sigma$, one can choose different separating sets.  Consider, for example, the $N$-torus $\Sigma = T^N$ (i.e., the $N$-copy of a circle $S^1$). We can express a point $\sigma\in T^N$ by $(\sigma_1,\ldots, \sigma_N)$ with $ \sigma_k\in [0, 2\pi)$. In this case,  the trigonometric functions $\cos(\sigma_k), \sin(\sigma_k)$ for $ k = 1,\ldots, N$ plus~${\bf 1}_{T^N}$ can be used to form a desired separating set with ${\bf 1}_{T^N}$ the everywhere nonzero function.    
\end{remark} 
 
We establish below Theorem~\ref{thm:mthm1}.  
The proof will be divided into two parts: we deal with ensemble controllability and ensemble observability separately. The proofs will be given in Subsections~\ref{ssec:paec} and~\ref{ssec:peo}, respectively.

\subsection{Proof of approximate ensemble path-controllability}\label{ssec:paec} 
We establish here the first item of Theorem~\ref{thm:mthm1}. The proof relies on the use of the technique of Lie extension. We review such a technique below. To proceed, we first recall that for an arbitrary {\em single} control-affine system: 
\begin{equation}\label{eq:singlecontroaffine}
\dot x(t) = f_0(x(t)) + \sum^m_{i = 1}u_i(t)f_i(x(t)),
\end{equation} 
the first order Lie extension of the system is a new control-affine system given by %with additional control vector fields and corresponding control inputs: 
$$\dot x(t) = f_0(x(t)) + \sum^m_{i = 1}u_i(t) f_i(x(t)) + \sum^m_{i,  j = 1}u_{ij}(t)[f_i,f_j](x(t)).$$ 
%where the control inputs $u_{ij}$'s are independent of the existing $u_i$'s.
By repeatedly applying Lie extensions, we obtain a family of control-affine systems with an increasing number of control vector fields. All of these control vector fields can be expressed as Lie products involving the $f_i$'s in~\eqref{eq:singlecontroaffine}. We make the statement precise below.  

First, recall that for the given set of vector fields $F:= \{f_i\}^m_{i = 1}$, we use ${\cal L}_F$ to denote the collection of Lie products generated by~$F$ in which the $f_i$'s are treated as if they were ``free''  generators. For ease of notation, we will simply write ${\cal L}$ by omitting the subindex~$F$. Decompose ${\cal L}:= \sqcup_{k \ge 0}{\cal L}(k)$ where each ${\cal L}(k)$ is comprised of Lie products of depth~$k$. Then, the $k$th order Lie extension of~\eqref{eq:singlecontroaffine} is a control-affine system given by
$$
\dot x(t) = f_0(x(t)) + \sum^k_{l = 0}\sum_{\xi \in {\cal L}(l)}u_\xi(t)\xi(x(t)).
$$  
By increasing the order~$k$, we obtain an infinite family of Lie extended systems. 
It is known that the original control-affine system~\eqref{eq:singlecontroaffine} is approximately path-controllable if and only if any of its Lie extended systems is. 
%Besides its use in proving controllability, the technique of Lie extension has been leveraged for addressing nonholonomic motioning planning problems. 
In fact, Sussmann and Liu showed in~\cite{sussmann1993lie,liu1997approximation} how to construct control inputs $u_i(t)$ using a finite number of sinusoidal inputs of appropriate frequencies to approximate a desired trajectory generated by a given (but arbitrary) Lie extended system. The same technique has also been used in~\cite{agrachev2016ensemble} for proving approximate ensemble controllability. We further refer the reader to~\cite{murray1993steering,sussmann1993lie,leonard1995motion} for the use of Lie extension in nonholonomic motion planning.  %in which a variant of Rashevsky-Chow theorem is established for approximate ensemble controllability.  
 
We now apply the technique of Lie extension to the ensemble system~\eqref{eq:ensemblesystem}. For convenience, we reproduce below the control part of the system:  
\begin{equation}\label{eq:controlpart}
\dot x_\sigma(t) = f_0(x_\sigma(t), \sigma) + \sum^m_{i = 1} \sum^{r}_{s = 1} u_{i,s}(t)\rho_{s}(\sigma) f_i (x_\sigma(t)),  \quad \forall \sigma\in \Sigma.
\end{equation}
In this case, we have that for any individual system-$\sigma$, the control vector fields are $\rho_s(\sigma) f_i(x_\sigma)$, for $1\le s \le r$ and $1\le i \le m$. Note that the Lie bracket of any two of these control vector fields is given by
$[\rho_s(\sigma) f_i, \rho_{s'}(\sigma) f_j] =\rho_s(\sigma)\rho_{s'}(\sigma) [f_i, f_j]$.   
Thus, the first order Lie extension of~\eqref{eq:controlpart} is given by
\begin{multline*}
\dot x_\sigma(t)=  f_0(x_\sigma(t), \sigma) + \sum^m_{i = 1} \sum^{r}_{s = 1} u_{i,s}(t)\rho_{s}(\sigma) f_i (x_\sigma(t)) + \\  \sum^m_{i, j =1}\sum^r_{s,s' = 1} u_{ij, ss'}(t) \rho_s(\sigma)\rho_{s'}(\sigma)[f_{i}, f_{j}](x_\sigma(t)), \quad \forall\sigma\in \Sigma. 
\end{multline*}
%Note that for any fixed pair $(i,j)$,  one can combine all the ``coefficients'' associated with the Lie bracket $[f_i,f_j](x_\sigma(t))$, namely $u_{ij,ss'}(t)\rho_s(\sigma)\rho_{s'}(\sigma)$ and $u_{ij, s's}(t)\rho_{s'}(\sigma)\rho_s(\sigma)$,  so that 
The last term of the above expression can be simplified as follows: 
$$
\sum_{\xi \in {\cal L}(1)} \sum_{\rp\in {\cal P}(2)} u_{\xi, \rp}(t) \rp(\sigma) \xi(x_\sigma(t)),  
$$
where ${\cal P}(2)$ is the collection of monomials $\rho_s\rho_{s'}$ of degree~$2$. In general, we obtain the following $k$th order Lie extension of~\eqref{eq:controlpart}:         
\begin{equation}\label{eq:lieensmebleextend}
\dot x_\sigma(t) = f_0(x_\sigma(t), \sigma) + \sum^{k}_{l = 0}\sum_{\xi\in {\cal L}(l)} \sum_{\rp\in {\cal P}(l + 1)} u_{\xi, \rp}(t) \rp(\sigma) \xi(x_\sigma(t)), \quad \forall\sigma\in \Sigma.   
\end{equation}
%Note that all the scalar control inputs $u_{\xi,\rp}$ are independent of each other.   
%We note again that all the scalar control inputs $u_{\xi, \rp}(t)$ are independent with each other.  

%Further, by repeatedly applying Lie extensions, we obtain that
%\begin{multline}\label{eq:lieensmebleextend}
%\dot x_\sigma = f_0(x_\sigma, \sigma) + \sum^m_{i = 1} u_i  \sigma f_i(x_\sigma)  + \sum_{1\le i < j \le m} u_{ij} \sigma^2[f_i,f_j](x_\sigma)  \\ 
%+ \sum_{i,j,k} u_{ijk} \sigma^3 [f_i, [f_j, f_k]](x_\sigma) +\cdots
%\end{multline}
%where the $k$-th summation in~\eqref{eq:lieensmebleextend} is over all Lie products of depth~$k$ in the P.~Hall basis generated by the $f_i$'s (treated as free generators). The $u_i$'s, $u_{ij}$'s, and $u_{ijk}$'s, etc, are all independent control inputs. 

Recall that two arbitrary sets of vector fields $\{f_i\}^m_{i = 1}$ and $\{f'_{i'}\}^{m'}_{i'= 1}$ over~$M$ are said to be projectively identical, which we denote by $\{f_i\}^m_{i = 1} \equiv \{f'_{i'}\}^{m'}_{i' = 1}$, if for any $f_i$,  there exist an $f'_{i'}$ and a real number~$\lambda$ such that $f_i = \lambda f'_{i'}$, and vice versa.  
We will use such an equivalence relation in the following way: In the original ensemble control system~\eqref{eq:controlpart}, the set of control vector fields  $\{f_i\}^m_{i = 1}$ is, by assumption, distinguished. Thus, by the second item of Def.~\ref{def:distvecf}, if we evaluate the Lie products in each ${\cal L}(k)$, then 
\begin{equation}\label{eq:skequiv}
{\cal L}(k) \equiv \left \{f_i \right \}^m_{i = 1}, \quad \forall k\ge 0.
\end{equation}
Since every control vector field~$f$ in~\eqref{eq:lieensmebleextend} is  obtained by evaluating a Lie product involving the $f_i$'s, by using the above fact, we can simplify the Lie extended  system~\eqref{eq:lieensmebleextend} as follows:
\begin{equation}\label{eq:simplifiedlieensembleextend}
\dot x_\sigma(t) =f_0(x_\sigma(t), \sigma) + \sum^m_{i = 1}\sum^{k}_{l = 0}\sum_{\rp\in {\cal P}(l + 1)} \left (u_{i,\rp}(t) \rp(\sigma) \right ) f_i(x_\sigma(t)), \quad \forall \sigma\in \Sigma. 
\end{equation}
The control inputs $u_{i,\rp}(t)$ in the above expression are defined such that 
%as follows: For $k = 1$, $u^{[1]}_i:= u_i$ for all $i = 1,\ldots, m$. For $k> 1$, we have
$$
u_{i,\rp}(t):= \sum_{\xi} \lambda_{\xi} \, u_{\xi, \rp}(t),
$$
where the summation is over Lie products~$\xi$ of depth~$(\deg(\rp) - 1)$ such that 
$
\xi  = \lambda_{\xi}  f_{i}
$. 
% Note that all the control inputs $u^{[k]}_i$'s are independent of each other.    
 
 We will now establish approximate ensemble path-controllability of system~\eqref{eq:simplifiedlieensembleextend} for a certain order~$k$.   
 Let $\hat x_\Sigma[0,T]$ be the trajectory of profiles we want the system~\eqref{eq:simplifiedlieensembleextend} to approximate. By the first item of Def.~\ref{def:distvecf}, we have that the set $\{f_i(x)\}^m_{i = 1}$ spans $T_xM$ for all $x\in M$. Thus, there are smooth functions $c_i(t, \sigma)$ (smooth in both $t$ and $\sigma$), for $i = 1,\ldots, m$,  such that the following hold:
 $$
 \frac{\partial \hat x_\sigma(t)}{\partial t} - f_0(\hat x_\sigma(t), \sigma) = 
 \sum^m_{i = 1} c_i(t, \sigma)f_i(\hat x_\sigma(t)), \quad \forall (t, \sigma)\in [0,T]\times \Sigma.  
 $$
Note, in particular, that if there exist an order~$k\ge 0$ and a set of 
control inputs $u_{i,\rp}$, for $i = 1, \ldots, m$ and $\rp$ a monomial with $1\le \deg(\rp) \le k + 1$, such that 
\begin{equation}\label{eq:equalitiesneedtosatisfy}
c_i(t, \sigma) = \sum^{k}_{l = 0}\sum_{\rp\in {\cal P}(l + 1)} u_{i,\rp}(t) \rp(\sigma), \quad \forall (t, \sigma)\in [0,T]\times \Sigma \, \mbox{ and } \, \forall i = 1,\ldots, m,  
\end{equation}
then the trajectory of profiles $x_\Sigma[0,T]$ generated by system~\eqref{eq:simplifiedlieensembleextend},  
%to the following differential equation: 
%\begin{equation}\label{eq:desiredtraj}
%\dot x_\sigma(t) =  f_0(x_\sigma, \sigma) + \sum^m_{i = 1} c_i(t, \sigma)f_i(\hat x_\sigma(t)), 
%\end{equation}
with $x_\Sigma(0) = \hat x_\Sigma(0)$, will be exactly $\hat x_\Sigma[0,T]$. Said in another way, if~\eqref{eq:equalitiesneedtosatisfy} holds, then one can steer the $k$th order Lie extended system~\eqref{eq:simplifiedlieensembleextend} to follow  the given target trajectory~$\hat x_\Sigma[0, T]$.

  But, in general, the equality~\eqref{eq:equalitiesneedtosatisfy} cannot be satisfied by a finite sum. Nevertheless, we show below that the two sides of the expression can be made arbitrarily close to each other provided that~$k$ is sufficiently large, i.e.,  
% Specifically, we will show that for any $\delta > 0$, there exist a~$k\ge 0$, a set of $\rp$ with $1\le \deg(\rp) \le k + 1$  
\begin{equation}\label{eq:whatwereallywant}
\left |  \sum^{k}_{l = 0}\sum_{\rp\in {\cal P}(l + 1)} u_{i,\rp}(t) \rp(\sigma) -  c_i(t, \sigma)  \right | < \delta, 
\end{equation}
for all $(t, \sigma)\in [0,T]\times \Sigma$ and for all $i = 1,\ldots, m$. 
 This essentially follows from the Stone-Weierstrass theorem.  We provide details below. Note that if~\eqref{eq:whatwereallywant} holds for any given $\delta > 0$, then the trajectory of profiles $x_\Sigma[0,T]$ generated by~\eqref{eq:simplifiedlieensembleextend} can be made uniformly and arbitrarily close to~$\hat x_\Sigma[0,T]$. %so that~\eqref{eq:ensemblecontrollabilityqipashuo} in Def.~\ref{def:ensemblecontrollability} is satisfied. 

By the assumption of Theorem~\ref{thm:mthm1}, the set $\{\rho_s\}^r_{s = 1}$ is a separating set and contains an everywhere nonzero function. Without loss of generality, we let $\rho_1$ be such a function, i.e., $\rho_1(\sigma) \neq 0$ for all $\sigma\in \Sigma$. It follows %from the Stone-Weierstrass theorem 
that the subalgebra generated by the set $\{\rho_s\}^r_{s = 1}$ is dense in ${\rm C}^0(\Sigma)$. %i.e., the space of continuous functions on~$\Sigma$.
Specifically, for any given $\delta' > 0$, there exist an integer~$k \ge 0$ and a set of smooth functions $u'_{i, \rp'}: [0,T]\to \R$, for $i = 1,\ldots, m$ and~$\rp'$ a monomial with $0\le \deg(\rp') \le k$, such that the following holds: %partition of unity over [0,T].
\begin{equation}\label{eq:whatwereallywant11}
\left |\sum^{k}_{l = 0}\sum_{\rp'\in {\cal P}(l)} u'_{i,\rp'}(t) \rp'(\sigma) - \rho^{-1}_1(\sigma) c_i(t, \sigma) \right | < \delta', 
\end{equation}       
 for all $(t, \sigma)\in [0,T]\times \Sigma$ and for all $i = 1,\ldots, m$.

 Let $\gamma:=\max\{ |\rho^{-1}_1(\sigma)| \mid \sigma\in \Sigma\}$ > 0. Note that $\gamma$ exists because $\rho_1$ is everywhere nonzero and $\Sigma$ is compact.  
 Now, given an arbitrary $\delta > 0$, we define $\delta':= \delta / \gamma$ and let the inequality~\eqref{eq:whatwereallywant11} be satisfied. By the definition of~$\gamma$, we have that  
$$
 \left |\sum^{k}_{l = 0}\sum_{\rp'\in {\cal P}(l)} u'_{i,\rp'}(t) (\rho_1(\sigma)\rp'(\sigma)) -  c_i(t, \sigma) \right | < \gamma \delta' = \delta,  
$$
for all $(t, \sigma)\in [0,T]\times \Sigma$ and for all $i = 1,\ldots, m$. 
 Note that each $\rho_1\rp'$ in the above expression is a monomial and $1\le \deg(\rho_1\rp') \le k + 1$. Next, for any $i = 1,\ldots, m$ and any monomial $\rp$ with $1\le \deg(\rp) \le k + 1$, we let the corresponding control input $u_{i,\rp}(t)$ be defined such that for any $t\in [0,T]$,    
 $$u_{i,\rp}(t):= 
 \left\{
 \begin{array}{ll}
 u'_{i,\rp'}(t) & \mbox{ if } \rp = \rho_1\rp' \mbox{ with } 0\le \deg(\rp') \le k, \\
 0 & \mbox{ otherwise}. 
 \end{array}
 \right. 
 $$  
With the above-defined control inputs $u_{i,\rp}(t)$, we conclude that~\eqref{eq:whatwereallywant} is satisfied.   

\subsection{Proof of ensemble observability}\label{ssec:peo}
We will now establish the second item of Theorem~\ref{thm:mthm1}. % which is about ensemble observability of system~\eqref{eq:ensemblesystem}. 
Let a profile $\bar x_\Sigma(0)$ be chosen such that it is output equivalent to $x_\Sigma(0)$. 
The majority of effort will be devoted to proving the following fact: 
If $\{\phi^j\}^l_{j = 1}$ is weakly codistinguished to $\{f_i\}^m_{i = 1}$, then there is an open neighborhood~$U$ of $x_\Sigma(0)$ in ${\rm C}^\omega(\Sigma,M)$ such that if $\bar x_\Sigma(0)$ intersects~$U$, then $\bar x_\Sigma(0) = x_\Sigma(0)$. After proving the fact, we will show that if, further, the common state space $M$ is compact, then $O(x_\Sigma(0))$ is a finite set and $|O(x_\Sigma(0))|\le \chi_\phi$. Finally, we will show that if $\{\phi^j\}^l_{j = 1}$ is codistinguished to $\{f_i\}^m_{i = 1}$, then $O(x_\Sigma(0)) = \{x_\Sigma(0)\}$.

To proceed, we first introduce a few key notations that will be used in the proof. For an arbitrary 
 differential equation $\dot x(t) = f(x(t))$, we denote by $e^{tf} x(0)$ the solution of the equation at time~$t$ with~$x(0)$ the initial condition. We will use such a notation to denote a solution $x_\sigma(t)$, for any $\sigma\in \Sigma$, of system~\eqref{eq:ensemblesystem}. Next, we recall that $u(t)$ is the collection of the control inputs $u_{i,s}(t)$, for $1 \le i \le m$ and $1\le s \le r$, in system~\eqref{eq:ensemblesystem}. We introduce a  notation for a piecewise constant control input~$u(t)$ over $[0, T]$ as follows:  
%Recall that $\{e_1,\ldots, e_m\}$ is the standard basis of $\R^m$. We introduce the following notation for a piecewise constant control input:  
\begin{equation}\label{eq:piecewiseconstant}
u[0, T] := (i_{1}, s_1, \nu_1, t_1)\cdots  (i_{k},  s_k, \nu_k, t_k), 
\end{equation} 
where  $0 < t_1  < \cdots  < t_{k} = T$ is an increasing sequence of switching times,  $\nu_p$'s are real numbers, and $(i_p, s_p)$'s are pairs of indices chosen out of $\{1,\ldots, m\} \times \{1,\ldots, r\}$. 
The piecewise constant control input $u[0,T]$ is defined such that if $t \in [t_{p - 1}, t_{p})$, then 
$$u_{i,s}(t) = \left\{
\begin{array}{ll}
\nu_p & \mbox{ if } (i, s) = (i_p, s_p),\\
0 & \mbox{ otherwise}. 
\end{array}\right.
$$  
Note, in particular, that at any time $t\in [0,T]$,  there is at most one nonzero scalar control input~$u_{i,s}(t)$ in~$u(t)$. %It should be clear that any such $u(t)$ is right-continuous. 

We will now apply the piecewise constant control input~\eqref{eq:piecewiseconstant} to excite system~\eqref{eq:ensemblesystem}. For convenience, we define $\tau_p:= t_{p} - t_{p-1}$ for all $p = 1, \ldots, k$ (with $t_0 := 0$).     
We further define a set of vector fields $\{\tilde f_p\}^k_{p = 1}$ as follows:  
$$
\tilde f_{p} := \nu_p \rho_{s_p} f_{i_p} + f_0, \quad \forall p = 1,\ldots, k, 
$$ 
where we have omitted all the arguments in the expression.    
Since $x_\Sigma(0)\sim \bar  x_\Sigma(0)$, we have that for all $j = 1,\ldots, l$,
$$
\int_\Sigma \phi^j\left (e^{\tau_k \tilde f_{k}}\cdots e^{\tau_1\tilde f_{1}}x_\sigma(0) \right ) d\sigma =  \int_\Sigma \phi^j\left (e^{\tau_k \tilde f_{k}}\cdots e^{\tau_1\tilde f_{1}}\bar x_\sigma(0)\right ) d\sigma. 
$$
Moreover, the above equality holds for any~$\tau_p$ and~$\nu_p$, with $p = 1,\ldots, k$.

We next take the partial derivative $\nicefrac{\partial^{k}}{\partial \tau_1\cdots\partial \tau_{k} }$ on both sides of the above expression and evaluate the derivatives at~$\tau_1 = \cdots = \tau_k =0$. By computation, we obtain that
$$
\int_\Sigma \left (\tilde f_{1} \cdots \tilde f_{k} \phi^j \right )(x_\sigma(0))   d\sigma = \int_\Sigma \left (\tilde f_{1} \cdots \tilde f_{k} \phi^j \right )(\bar x_\sigma(0))  d\sigma.
$$
We further take the partial derivative $\nicefrac{\partial^k}{\partial \nu_1\cdots \partial \nu_k}$ and evaluate  
at $\nu_1 = \cdots  = \nu_k = 0$. By computation, we obtain that
\begin{equation}\label{eq:meirenyu}
\int_\Sigma (f_{\rm w} \phi^j)(x_\sigma(0)) \rp(\sigma)  d\sigma = \int_\Sigma ( f_{\rm w} \phi^j)(\bar x_\sigma(0))\rp(\sigma)  d\sigma,
\end{equation}
where ${\rm w} := i_1\cdots i_k$ is a word and $\rp := \rho_{s_1}\cdots \rho_{s_k}$ is a monomial.   

Note that $\{\phi^j\}^l_{j = 1}$ is (weakly) codistinguished to $\{f_i\}^m_{i = 1}$. By the second item of Def.~\ref{def:codisth}, we have that for any $j' = 1,\ldots, l$, there exist a word ${\rm w}$ over the alphabet $\{1,\ldots, m\}$ of length~$k$, a function $\phi^{j}$, and a {\em nonzero}~$\lambda$ such that $f_{\rm w}\phi^{j} = \lambda\phi^{j'}$. Since~\eqref{eq:meirenyu} holds for all words ${\rm w}$ of length~$k$ and all monomials $\rp$ of degree~$k$ for~$k$ arbitrary, we obtain that for all $j = 1,\ldots, l$, 
\begin{equation}\label{eq:foranyk}
\int_\Sigma  \phi^j(x_\sigma(0)) \rp(\sigma)  d\sigma =  \int_\Sigma  \phi^j(\bar x_\sigma(0))\rp(\sigma)  d\sigma, %\quad \forall  j = 1,\ldots, l, 
\end{equation}
which holds for all monomials $\rp$.  
%In general, we consider a control law $u= \delta_1 e_{i_1}[t_0,t_1) \cdots \delta_k e_{i_k}[t_{k-1}, t_{k}]$. Since $ x_\Sigma$ and $\bar x_\Sigma$ are indistinguishable, for any nonnegative $s_i$'s and any real numbers $\delta_i$'s, we have 
%$$
%\int_\Sigma \phi^j(e^{s_{k} \tilde f_{i_{k}}} \cdots e^{s_1\tilde f_{i_1}}x_\sigma) d\sigma =  \int_\Sigma \phi^j(e^{s_{k} \tilde f_{i_{k}}} \cdots e^{s_1\tilde f_{i_1}}\bar x_\sigma) d\sigma. 
%$$ 
%Taking the partial derivative $\nicefrac{\partial^{2k}}{\partial s_1\cdots\partial s_{k} \partial \delta_1\cdots\partial \delta_{k} }$ of the above expression and then evaluating at $s_i =  \delta_i = 0$ for all $i = 1,\ldots, k$, we obtain that
%$$
%\int_\Sigma ( f_{i_1}\cdots  f_{i_k}\phi^j)(x_\sigma) \sigma^k d\sigma =  \int_\Sigma ( f_{i_1}\cdots  f_{i_k} \phi^j)(\bar x_\sigma) \sigma^k d\sigma.
%$$
%Since $\{\phi^j\}^l_{j = 1}$ is codistinguished to $\{f_i\}^m_{i = 1}$, there exist a word ${\rm w}$  of length $k$, a function $\phi^{j'}$, and a {\em nonzero}~$\lambda$ such that $f_{\rm w} \phi^{j'} = \lambda \phi^j$,  and hence
%\begin{equation}\label{eq:foranyk}
%\int_\Sigma  \phi^j(x_\sigma) \sigma^{k}  d\sigma =  \int_\Sigma  \phi^j(\bar x_\sigma)\sigma^{k}  d\sigma, \quad \forall j = 1,\ldots, l. 
%\end{equation}
%and moreover, the expression above holds for all $k\in \mathbb{N}$.

We now let ${\rm L}^2(\Sigma)$ be the Hilbert space space of all square-integrable functions on~$\Sigma$, where  
%i.e., $${\rm L}^2(\Sigma):= \left \{ {\rm q} \mid \int_\Sigma {\rm q}^2(\sigma) d\sigma < \infty \right \}.$$   
the inner-product is defined as follows:  
$$ 
\langle {\rm q}_1, {\rm q}_2 \rangle_{{\rm L}^2} : = \int_\Sigma {\rm q}_1(\sigma){\rm q}_2(\sigma) d\sigma, \quad \forall {\rm q}_1, {\rm q}_2\in {\rm L}^2(\Sigma).  
$$
Note that $\Sigma$ is compact. 
By the assumption of Theorem~\ref{thm:mthm1}, the set of parametrization functions $\{\rho_s\}^r_{s = 1}$ is separates points and contains an everywhere-nonzero function, so the subalgebra generated by the set is dense in ${\rm L}^2(\Sigma)$.   
Thus, if there is a function ${\rm q}\in {\rm L}^2(\Sigma)$ such that $\langle {\rm q}, \rp\rangle_{{\rm L}^2} = 0$ for all monomials $\rp\in {\cal P}$, then ${\rm q}$ is zero almost everywhere (it differs from the identically-zero function over a set of measure zero).  
In the case here, we define for each $j = 1,\ldots, l$ the following function:
$${\rm q}^j(\sigma) := \phi^j(x_\sigma(0)) - \phi^j(\bar x_\sigma(0)).$$ Then, one can re-write~\eqref{eq:foranyk} as follows: 
$$
\langle {\rm q}^j, \rp  \rangle_{{\rm L}^2} = 0, \quad \forall \rp\in {\cal P} \,\, \mbox{ and }\,\, \forall j = 1,\ldots, l.
$$
Because $x_\sigma(0), \bar x_\sigma(0)$ are analytic in~$\sigma$ and each $\phi^j(x)$ is analytic in~$x$, we have that each ${\rm q}^j(\sigma)$ is analytic in~$\sigma$. 
Furthermore, since $\Sigma$ is equipped with a strictly positive measure, we have that each~${\rm q}^j$ is identically zero, i.e.,   
\begin{equation}\label{eq:concludingtheproof}
\phi^j(x_\sigma(0)) =  \phi^j(\bar x_\sigma(0)), \quad \forall \sigma \in \Sigma \,\,\mbox{ and }\,\, \forall j = 1,\ldots,l.
\end{equation}

Since $\{\phi^j\}^l_{j = 1}$ is (weakly) codistinguished to $\{f_i\}^m_{i = 1}$, by the first item of Def.~\ref{def:codisth}, the set of one-forms $\{d\phi^j_{x}\}^l_{j = 1}$ spans the cotangent space $T^*_{x}M$ for all $x\in M$. It follows that for any $x\in M$, there is an open ball $B_{\epsilon(x)}(x)$ centered at~$x$ with radius $\epsilon(x) > 0$ such that if $\bar x\in B_{\epsilon(x)}(x)$ and $\phi^j(x) = \phi^j(\bar x)$ for all $j = 1,\ldots, l$, then $\bar x = x$. Furthermore, since each $\phi^j$ is analytic, for any fixed $x\in M$, the radius~$\epsilon(x)$ of the open ball can be chosen such that it is locally continuous around~$x$.     
Since the initial profile $x_\Sigma(0)$ is analytic in~$\sigma$, the above arguments have the following implication: for each $\sigma\in \Sigma$, there is an open neighborhood $V_\sigma$ of $\sigma$ in $\Sigma$ and a positive number~$\epsilon_\sigma$ such that if $\sigma'\in V_\sigma$ and $\bar x_{\sigma'}(0)$ belongs to the open ball $B_{\epsilon_\sigma}(x_{\sigma'}(0))$ with
$\phi^j(x_{\sigma'}(0)) = \phi^j(\bar x_{\sigma'}(0))$ for all $j = 1,\ldots, l$, then $\bar x_{\sigma'}(0) = x_{\sigma'}(0)$. 

 The collection of the above open sets $\{V_\sigma\}_{\sigma\in \Sigma}$ is an open cover of~$\Sigma$. Since~$\Sigma$ is compact, there exists a finite subcover  $\{V_{\sigma_{i}}\}^N_{i =1} $ of~$\Sigma$. We then let $$\epsilon:= \min\left \{\epsilon_{\sigma_i} \mid i = 1,\ldots, N\right \}> 0.$$ 
 We further let~$U$ be an open neighborhood of $x_\Sigma(0)$ in ${\rm C}^\omega(\Sigma, M)$ defined as follows:
$$
U:= \left \{x'_\Sigma \in {\rm C}^\omega(\Sigma, M) \mid {\rm d}_M(x_\sigma(0), x'_\sigma) < \epsilon, \quad \forall \sigma\in \Sigma \right \}. 
$$
We show below that if $\bar x_\Sigma(0)$ intersects~$U$, then $\bar x_\Sigma(0) = x_\Sigma(0)$. Note that if this is the case, then weak ensemble observability of system~\eqref{eq:ensemblesystem} is established. 

To establish the fact, we first assume that there is a certain $\sigma \in \Sigma$ such that ${\rm d}_M(x_\sigma(0), \bar x_\sigma(0)) < \epsilon$ and, hence, $\bar x_\Sigma(0)$ intersects~$U$. Then, by the definition of~$\epsilon$, we have that $\bar x_\sigma(0) =  x_\sigma(0)$. Now, let $\sigma'$ be any other point of~$\Sigma$. We need to show that $\bar x_{\sigma'}(0) = x_{\sigma'}(0)$. Because~$\Sigma$ is path-connected, there is a continuous path $p:[0,1] \to \Sigma$ with $p(0) = \sigma$ and $p(1) = \sigma'$. Again, by the definition of~$\epsilon$, we have that for any $\lambda\in [0,1]$, there are only two cases: Either $\bar x_{p(\lambda)}(0) = x_{p(\lambda)}(0)$ or ${\rm d}_M(x_{p(\lambda)}, \bar x_{p(\lambda)}) \ge \epsilon$. On the other hand, the profile $\bar x_\Sigma(0)$ is analytic in~$\sigma$ and $p(\lambda)$ is continuous in~$\lambda$, so $\bar x_{p(\lambda)}(0)$ is continuous in~$\lambda$ as well. But then, since $\bar x_{p(0)}(0) = x_{p(0)}(0)$, it must hold that $\bar x_{p(\lambda)}(0) = x_{p(\lambda)}(0)$ for all $\lambda\in [0,1]$. In particular, $\bar x_{\sigma'}(0) = x_{\sigma'}(0)$. %We have thus shown that $\bar x_\Sigma(0) = x_\Sigma(0)$. 
We have thus shown that if $\{\phi^j\}^l_{j = 1}$ is weakly codistinguished to $\{f_i\}^m_{i = 1}$, then system~\eqref{eq:ensemblesystem} is weakly ensemble observable.

We now show that if, further, $M$ is  compact, then $|O(x_\Sigma(0))| \le \chi_\phi$. Recall that $\phi:= (\phi^1,\ldots,\phi^l)$ and $[x]_\phi$ is the pre-image of $\phi(x)$.   
Because any two different profiles in $O(x_\Sigma(0))$ are completely disjoint, it suffices to show that $|[x_\sigma(0)]_\phi|\le \chi_\phi$ for some (and, hence, any) $\sigma\in \Sigma$. But, this follows from  
the definition of $\chi_\phi$ and 
Lemma~\ref{lem:finiteambiguity}.

Finally, note that if $\{\phi^j\}^l_{j = 1}$ is codistinguished to $\{f_i\}^m_{i = 1}$ (and, hence, the third item of Def.~\ref{def:codisth} is satisfied), then by~\eqref{eq:concludingtheproof}, $x_\sigma(0) = \bar x_\sigma(0)$ for all $\sigma\in \Sigma$, i.e., $O(x_\Sigma(0)) = \{ x_\Sigma(0) \}$.  Thus, system~\eqref{eq:ensemblesystem} is ensemble observable. This completes the proof.

\subsection{Pre-distinguished ensemble system}\label{ssec:preensemble}
We consider in the subsection a more challenging scenario where the set of control vector fields $\{f_i(x)\}^m_{i = 1}$ 
(resp. the set of one-forms $\{d\phi^j(x)\}^l_{j = 1}$) in system~\eqref{eq:ensemblesystem} does not necessarily span the tangent space $T_xM$ (resp. the cotangent space $T^*_xM$). Nevertheless, the two sets $\{f_i\}^m_{i = 1}$ and $\{\phi^j\}^l_{j = 1}$  together can ``generate'' (weakly) jointly distinguished vector fields and functions. We make the statement precise below.

To proceed, we first introduce a few definitions and notations. 
Let $F:= \{f_i\}^m_{i = 1}$ and ${\cal L}_F$ be the collection of Lie products generated by~$F$ (the $f_i$'s are treated as ``free'' generators).  We say that ${\cal L}_F$ is {\em projectively finite} if there is a finite set of vector fields $\bar F:= \{\bar f_i\}^{\bar m}_{i = 1}$ over~$M$ such that if one evaluates the Lie products in ${\cal L}_F$, then ${\cal L}_F\equiv \bar F$.  

Next, let ${\cal W}$ be the collection of all words over the alphabet $\{1,\ldots, m\}$. Recall that for a given word ${\rm w} = i_1\cdots i_k$ and an analytic function $\phi$ on~$M$, we use $f_{\rm w} \phi$  to denote $f_{i_1}\cdots f_{i_k}\phi$. If ${\rm w} = \varnothing$, then $f_{\rm w}\phi = \phi$.  
Given a set function $\Phi:=\{\phi^j\}^l_{j = 1}$ on~$M$ and the set of vector fields $F$, we define
$$
F_{\cal W} \Phi := \{f_{\rm w} \phi^j \mid  {\rm w} \in {\cal W} \mbox{ and } j = 1,\ldots, l \}. 
$$
%be defined as the collection of all $f_{\rm w}\phi^j$ for $j = 1,\ldots, l$ and ${\rm w}$ a word over the alphabet $\{1,\ldots, m\}$.
Similarly, we say that $F_{\cal W}\Phi$ is {\em projectively finite} if there is a finite subset $\bar \Phi :=\{\bar \phi^j\}^{\bar l}_{j = 1}$ of ${\rm C}^\omega(M)$ such that $F_{\cal W}\Phi \equiv \bar \Phi$. %Further, ${\cal  D}$  is said to be {\em (weakly) codistinguished} to~$F$ if $\bar \Phi$ is (weakly) codistinguished to~$F$.   
Note, in particular, that $F$ and $\Phi$ are, up to scaling, subsets of $\bar F$ and $\bar \Phi$, respectively. 
%We also note that if the two given sets $F$ are $\Phi$ are (weakly) jointly distinguished, then ${\cal L}_F \equiv F$ and $F_{\cal W} \Phi\equiv \Phi$.  
We now have the following definition:

%${\rm Gen}_F(\Phi) \subset {\rm C}^\omega(M)$ be comprised of  $f_{\rm w}\phi^j$ for all $j = 1,\ldots, l$ and for all words~${\rm w}$ over the alphabet $\{1,\ldots, m\}$, and $\overline \Phi_F$ be a minimal subset of ${\rm C}^\omega(M)$ which is equivalent to ${\rm Gen}_F(\Phi)$. 

\begin{definition}\label{def:predistinguishedvecf}
A set of vector fields $F := \{f_i\}^m_{i = 1}$ over~$M$ is {\bf pre-distinguished} if there exists a distinguished set $\bar F$ of vector fields such that $ {\cal L}_F \equiv \bar F$. Similarly, a set of functions $\Phi := \{\phi^j\}^l_{j = 1}$ on~$M$ is {\bf (weakly) pre-codistinguished} to~$F$ 
if there exists a finite set $\bar \Phi$ of functions, (weakly) codistinguished to~$F$, such that $F_{\cal W} \Phi\equiv \bar \Phi$.    %The two subsets $F$ and $\Phi$ are {\bf (weakly) jointly pre-dist.} if $\overline F$ and $\overline \Phi_F$ are (weakly) jointly distinguished. 
\end{definition}

Note that given a pair of jointly distinguished sets $F$ and $\Phi$, one can look for (proper) subsets $F'\subseteq F$ and $\Phi' \subseteq \Phi$ so that ${\cal L}_{F'} \equiv F$ and ${F'}_{\cal W}\Phi' \equiv \Phi$, i.e., $F'$ and $\Phi'$ are {\em jointly pre-distinguished}. In particular, we say that $(F', \Phi')$ is {\em minimal} if removal of any element out of $F'$ or $\Phi'$ will violate the condition in the above definition. We do not intend to characterize here minimal pairs for a given jointly distinguished pair $(F,\Phi)$. But instead, we provide below an example for illustration.    
%Note that if $F$ is distinguished and $\Phi$ is codistinguished to~$F$, then $\overline F = F$ and $\overline \Phi_F = \Phi$. In  this case, it is interesting to know whether there exist a proper subset $F'$ of $F$ and a proper subset $\Phi'$ of $\Phi$, with {\em minimal} cardinalities, such that $\overline{F'} = F$ and $\overline{\Phi'}_{F'}  = \Phi$. 

%If the two sets $\bar F$ and $\bar \phi$ are jointly distinguished, then $F$  

\begin{example}\label{exmp:fiandphij}{\em 
We consider again the  vector fields $F := \{L_{X_i}\}^3_{i = 1}$ and the functions 
 $\Phi := \{ \phi_{ij} = \tr(gX_j g^\top X^\top_i)\}^3_{i, j = 1}$ introduced in Example~\ref{emp:example1}. We have shown that~$F$ and~$\Phi$ are jointly distinguished on $\SO(3)$.  Now, we define for each $i = 1,2, 3$, a subset $F_i:= F - \{L_{X_i}\}$ and for each $j = 1,2,3$, a subset $\Phi^j:= \{\phi^{ij}\}^3_{i = 1}$.  Recall that we have the following relationships:   $$[L_{X_i},L_{X_j}] = \det(e_i,e_j,e_k)L_{X_k} \quad \mbox{ and } \quad L_{X_i} \phi^{i'j} = - \det(e_i,e_j,e_k)  \phi^{i'k}.$$ 
It follows that ${\cal L}_{F_i} \equiv F$ for all $i = 1,2,3$, and $F_{{i}_{\cal W}} \Phi^j \equiv \Phi$ for all $1\le i, j \le 3$. Moreover, every such pair $(F_i, \Phi^j)$ is minimal.  %are jointly pre-distinguished for any $i, j = 1,2,3$.  Moreover, the two sets $F_i$ and $\Phi^j$ are minimal in cardinalities.  
}
\end{example}

With the above definition,  
we state the following fact which generalizes Theorem~\ref{thm:mthm1}:  

\begin{theorem}\label{thm:mthm2}
Consider the ensemble system~\eqref{eq:ensemblesystem}. 
Suppose that $\{\rho^2_s\}^r_{s = 1}$ is a separating set and contains an everywhere nonzero function; then, the following hold: 
\begin{enumerate}
 \item[1)] If the set of control vector fields $\{f_i\}^m_{i = 1}$ is pre-distinguished, then system~\eqref{eq:ensemblesystem} is approximately ensemble path-controllable.
 \item[2)] If the set of observation functions $\{\phi^j\}^l_{j = 1}$ is (weakly) pre-codistinguished to $\{f_i\}^m_{i = 1}$, then system~\eqref{eq:ensemblesystem} is (weakly) ensemble observable. If, further, $M$ is compact, then for any initial profile $x_\Sigma(0)$, the set $O(x_\Sigma(0))$ is finite and $|O(x_\Sigma(0))| \le \chi_\phi$. 
 \end{enumerate}
\end{theorem}

Similar to Def.~\ref{def:distinguishedensemblesys}, we have the following:

\begin{definition}\label{def:predistinguishedensemblesys}
An ensemble system~\eqref{eq:ensemblesystem} is a {\bf pre-distinguished} if {\em 1)} the set $\{\rho^2_s\}^r_{s = 1}$ separates points and contains an everywhere nonzero function, and {\em 2)} the set of control vector fields $\{f_i\}^m_{i = 1}$  and the set   of observation functions $\{\phi^j\}^l_{j = 1}$ are (weakly) jointly pre-distinguished.  
\end{definition}

  %Similarly, the proof is divided into two parts---we establish ensemble controllability and observability separately. 
It follows from Theorem~\ref{thm:mthm2} that if a system is pre-distinguished, then it is approximately ensemble path-controllable and (weakly) ensemble observable.  We next have the following remark on the existence of a desired set of parametrization functions that satisfies the assumption of Theorem~\ref{thm:mthm2} (compared to Remark~\ref{rmk:existenceofrho1}): 

\begin{remark}{\em 
We first note that if $\{\rho^2_s\}^r_{s = 1}$ is a separating set, then, for any positive integer~$k$, $\{\rho^k_s\}^r_{s = 1}$ is also a separating set. 
Conversely, if $\{\rho_s\}^r_{s = 1}$ is a separating set and each $\rho_s$ is nonnegative (i.e., $\rho_s(\sigma) \ge 0$ for all $\sigma\in \Sigma$), then $\{\rho^2_s\}^r_{s = 1}$ will be a separating set. Such a set $\{\rho_s\}^r_{s = 1}$ exists for any analytic, compact manifold~$\Sigma$. To see this, we again embed $\Sigma$ into a Euclidean space~$\R^N$. Since $\Sigma$ is compact, one can translate the coordinates, if necessary, such that $\Sigma$ is embedded in the positive orthant of~$\R^N$. Then, by restricting the coordinate functions to $\Sigma$, we obtain a separating set $\{\rho_i(\sigma):= \sigma_i\}^N_{i = 1}$ comprised of all positive functions.   }
\end{remark}

We establish Theorem~\ref{thm:mthm2} in the following subsection.  

\subsection{Analysis and proof of Theorem~\ref{thm:mthm2}}
Let $\bar F =\{\bar f_i\}^{\bar m}_{i =1}$ be such that $\bar F \equiv {\cal L}_F$. Decompose ${\cal L}_F:= \sqcup_{k\ge 0}{\cal L}_F(k)$ where ${\cal L}_F(k)$ is comprised of Lie products of depth~$k$.   
In contrast to~\eqref{eq:skequiv}, we do not necessarily have that ${\cal L}(k)\equiv \bar F$ for all $k \ge 0$. It is possible that each ${\cal L}_F(k)$  is, up to scaling, a proper subset of~$\bar F$. 
To tackle the issue, we first introduce a few definitions:

\begin{definition}
Let ${\cal L}_F$ be projectively finite and $\bar F =\{\bar f_i\}^{\bar m}_{i =1}$ be such that $\bar F \equiv {\cal L}_F$. For each~$i = 1,\ldots, \bar m$, define a set of natural numbers $\mathbb{N}_i$ as follows: if $k\in \N_i$, then there exist a Lie product $\xi\in {\cal L}_F(k)$ and a real number~$\lambda$ such that by evaluating~$\xi$, we have $$\bar f_i = \lambda \xi.$$ 
We call every such sequence $\N_i$ an {\bf indicator sequence for $\bar f_i$}.    
\end{definition}

Similarly, we have the following the counterpart of the above definition:  

\begin{definition}
Let $F_{\cal W}\Phi$ be projectively finite and $\bar \Phi = \{\bar \phi^j\}^{\bar l}_{j = 1}$ be such that $\bar \Phi \equiv F_{\cal W}\Phi$. For each $j = 1,\ldots, \bar l$, define a set of natural numbers $\N^j$ as follows: if $k \in \N^j$, then there exist a word ${\rm w}$ of length~$k$ over the alphabet $\{1,\ldots, m\}$, a function $\phi^{j'}\in \Phi$, and a real number $\lambda$ such that $$\bar \phi^j =  \lambda  f_{\rm w}\phi^{j'}.$$  
We call every such sequence $\N^j$ an {\bf indicator sequence for $\bar \phi^j$}.  
\end{definition}

%Decompose ${\cal  D} = \sqcup_{k\ge 0}{\cal  D}(k)$ where ${\cal  D}(k)$ is comprised of all $f_{\rm w}\phi^j$ for $j = 1,\ldots, l$ and ${\rm w}$ a word of length~$k$. Then,    

Note that if~$F$ and $\Phi$ are (weakly) jointly distinguished, then $\N_i = \N^j = \N$ for all $i = 1,\ldots, m (=\bar m)$ and for all $j = 1,\ldots, l (=\bar l)$. 

\begin{example}{\em 
	Consider the subsets $F_1 = F - \{L_{X_1}\}$ and $\Phi^1 =\{\phi^{i1}\}^3_{i = 1}$ introduced in Example~\ref{exmp:fiandphij}. We have that ${\cal L}_{F_1} \equiv F$ and $F_{1_{\cal W}}\Phi^1 \equiv \Phi$. By computation (with details omitted), the indicator sequences $\N_i$ for $L_{X_i}$ are given by $\N_1 = \{2k + 1\}_{k\ge 0}$ and $\N_2 = \N_3 = \{2k\}_{k \ge 0}$.  The indicator sequences $\N^{ij}$ for $\phi^{ij}$ are given by 
$\N^{i1} = \{2k\}_{k\ge 0}$ and $\N^{i2} = \N^{i3} = \{2k + 1\}_{k\ge 0}$ for all $i = 1,2,3$.  	   }
\end{example}

We call a monotonically increasing sequence~$\{n_k\}^\infty_{k = 0}$, with $n_k\in \N$, an {\bf arithmetic sequence} if there is a positive integer~$\delta$ such that $n_{k + 1} - n_k = \delta$ for all $k \ge 0$. In the above example, each indicator sequence $\N_i$ for $L_{X_i}$ (or $\N^{ij}$ for $\phi^{ij}$) is an arithmetic sequence with $\delta = 2$. We generalize this fact in the following proposition, which will be of great use in the proof of Theorem~\ref{thm:mthm2}: 

\begin{proposition}\label{lem:arithmeticseq}
	Every indicator sequence $\mathbb{N}_i$ for $\bar f_i$ (or $\mathbb{N}^j$ for $\bar \phi^j$) contains an infinite arithmetic sequence as a subsequence.  
\end{proposition}

\begin{proof}
%The proof will be divided into two parts. 
We establish the proposition for $\N_i$ and $\N^j$ subsequently. 

{\em Proof for $\N^i$.} 
We fix an $i = 1,\ldots, \bar m$, and prove that~$\N_i$ contains an arithmetic sequence. Because~$F$ is pre-distinguished, there exists a Lie product $\xi_1\in {\cal L}_F$, with $\dep(\xi_1) \ge 1$, and a real number $\lambda_1$ such that $\lambda_1 \xi_1 = \bar f_i$. Denote by~$f_{i_1}\in F$ the first element that shows up in $\xi_1$ (e.g., $\xi_1 = [f_{i_1}, [f_{i'_1},f_{i''_1}]]$). Applying the same argument, but with $\bar f_{i}$ replaced by $f_{i_1}$, we obtain that $\lambda_2 \xi_2 = f_{i_1}$ for some $\xi_2\in {\cal L}_F$ with $\dep(\xi_2) \ge 1 $ and some $\lambda_2\in \R$. 
%Similarly, we let~$f_{i_2}$ be the first element that shows up in~$\xi_2$.  

Next, we let ${\xi_1} \lhd \xi_2$ be a Lie product in ${\cal L}_F$ defined by replacing the first element~$f_{i_1}$ in~$\xi_1$ with the Lie product~$\xi_2$. For example, if $\xi_1 = [f_{i_1}, [f_{i'_1},f_{i''_1}]]$, then  ${\xi_1} \lhd \xi_2 = [\xi_2, [f_{i'_1},f_{i''_1}]]$. It should be clear that $$\lambda_1\lambda_2 {\xi_1} \lhd \xi_2 = f_i, \quad \mbox{with } \dep({\xi_1} \lhd \xi_2) = \dep({\xi_1}) + \dep({\xi_2}).$$ 
By repeating the above procedure, we obtain {\em 1)} a sequence of Lie products $\{\xi_k\}_{k\ge 1}$, {\em 2)} a sequence of vector fields $\{f_{i_k}\}_{k\ge 1}$ with $f_{i_k}\in F$,  and {\em 3)} a sequence of real numbers $\{\lambda_k\}_{k\ge 1}$ such that the first element in $\xi_k$ is $f_{i_k}$ and $\lambda_k \xi_k = f_{i_{k-1}}$.  
It then follows that 
$$
\alpha_k {\xi_1} \lhd \cdots \lhd \xi_k = f_i\,\, \mbox{ where } \alpha_k:=\prod^k_{l = 1} \lambda_l, \quad \forall k \ge 1.
$$ 
Note that ${\xi_1} \lhd \cdots \lhd \xi_k$ is well defined because the operator~``$\lhd$'' is associative. 

Since each $f_{i_k}$ belongs to the finite set $F$, there is a repetition in the sequence. Without loss of generality, we assume that $f_{i_k} = f_{i_{k'}}$ for some $k' > k \ge 1$.  
We then define a Lie product $\xi$ as follows: 
$$\xi:= \xi_{k + 1}\lhd \cdots \lhd \xi_{k'} \quad \mbox{and} \quad \delta: = \dep(\xi) = \sum^{k'}_{l = k + 1}\dep(\xi_l).$$  
Note that the first element in $\xi$ is~$f_{i_k}$ and  
$\nicefrac{\alpha_{k'}}{\alpha_k} \xi  = f_{i_k}$. In fact, the statement can be strengthened:  
For any given $N \ge 0$, we define $$\xi^N:= \xi\lhd \cdots \lhd \xi$$ 
where the number of copies of $\xi$ in the expression is~$N$. 
If $N = 0$, then we let $\xi^0:=f_{i_k}$.  %which is the first element in $\xi$.  
It should be clear that for any $N \ge 0$, the first element in $\xi^N $ is~$f_{i_k}$ and, moreover, 
$
\nicefrac{\alpha_{k'}^N}{\alpha_k^N} \xi^N  = f_{i_k}
$. 
We further define a Lie product $\xi_0$ as follows: 
$$\xi_0:= \xi_{1}\lhd\cdots \lhd \xi_k \quad \mbox{and} \quad  \delta_0:= \dep(\xi_0) = \sum^k_{l = 1}\dep(\xi_l).$$ 
It then follows that for any $N\ge 0$.
$$
\left(\nicefrac{\alpha_{k'}^N}{\alpha_k^{N - 1}} \right) \xi_0 \lhd \xi^N = \bar f_{i},
$$  
which implies that $\N_i$ contains $\{\delta_0 + N\delta\}_{N\ge 0}$ as a subsequence.

{\em Proof for $\N^j$.} The arguments will be similar to the ones used above. We fix a $j = 1,\ldots, \bar l$, and prove that $\N^j$ contains an arithmetic sequence. %Similar arguments will be used to establish the fact.  
Since~$\Phi$ is pre-codistinguished to~$F$, there exist a word ${\rm w}_1$ of positive length, a function $\phi^{j_1}$ out of~$\Phi$, and a real number~$\mu_1$ such that $\mu_1 f_{{\rm w}_1}\phi^{j_1} = \bar \phi^j$. Applying the same argument, but with $\bar \phi^j$ replaced by $\phi^{j_1}$, we obtain that $\mu_2 f_{{\rm w}_2}\phi^{j_2} = \phi^{j_1}$ for some word ${\rm w}_2$ of positive length, some function $\phi^{j_2}$ out of $\Phi$, and some real number~$\mu_2$. Note, in particular, that $$\mu_1\mu_2f_{{\rm w}_1}f_{{\rm w}_2}\phi^{j_2} = \bar \phi^{j}.$$  

By repeating the procedure, we obtain {\em 1)} a sequence of functions $\{\phi^{j_k}\}_{k\ge 1}$ where each $\phi^{j_k}$ belongs to~$\Phi$, {\em 2)} a sequence of words $\{{\rm w}_k\}_{k \ge 1}$ of positive lengths, and {\em 3)} a sequence of real numbers $\{\mu_k\}_{k\ge 1}$ such that $\mu_k f_{{\rm w}_k} \phi^{j_k} = \phi^{j_{k-1}}$. It then follows that
$$
\beta_k f_{{\rm w}_1} \cdots f_{{\rm w}_k}\phi^{j_k} = \bar \phi^j \, \, \mbox{ where } \beta_k:=\prod^k_{l = 1} \mu_l, \quad \forall k \ge 1.   
$$
Since each $\phi^{j_k}$ belongs to the finite set $\Phi$, there is a repetition in the sequence, say $\phi^{j_{k}} = \phi^{j_{k'}}$ for some $k' > k \ge 1$. It then implies that
$
\nicefrac{\beta_{k'}}{\beta_k} f_{{\rm w}}\phi^{j_{k}} = \phi^{j_k}$ where ${\rm w} := {\rm w}_{k + 1} \cdots {\rm w}_{k'}$ is obtained by concatenation. Denote by~$\delta$ the length ${\rm w}$. For a nonnegative integer~$N$, we let ${\rm w}^N$ be a word obtained by concatenating~$N$ copies of~${\rm w}$. If $N = 0$, then ${\rm w}^N = \varnothing$. 
We further let ${\rm w}_0:= {\rm w}_{1} \cdots {\rm w}_{k}$ and $\delta_0$ be the length of~${\rm w}_0$. It then follows that for any $N \ge 0$, 
$$
\left (\nicefrac{\beta^{N}_{k'}}{\beta^{N - 1}_k} \right ) f_{{\rm w}_0} f_{{\rm w}^N}\phi^{j_k} = \bar \phi^j, 
$$
which implies that $\N^j$ contains $\{\delta_0 + N\delta\}_{N\ge 0}$ as a subsequence. 
\end{proof}

With Prop.~\ref{lem:arithmeticseq} at hand, we now prove Theorem~\ref{thm:mthm2}:    

\begin{proof}  
The proof of Theorem~\ref{thm:mthm2} will be similar to the proof of Theorem~\ref{thm:mthm1}. We emphasize below the difference between the two proofs.  

We first establish item~1 of Theorem~\ref{thm:mthm2}. By repeatedly applying Lie extensions of system~\eqref{eq:ensemblesystem}, we obtain the following:
$$
\dot x_\sigma(t) = f_0(x_\sigma(t),\sigma) + \sum_{l \ge 0}\sum_{\xi \in {\cal L}_F(l)}\sum_{\rp \in {\cal P(l + 1)}} u_{\xi, \rp}(t) \rp(\sigma) \xi(x_\sigma(t)), \quad \forall \sigma\in \Sigma. 
$$
One obtains a $k$th order Lie extended system by truncating the infinite summation over~$l$ and keeping only the terms with~$l \le k$. Because~$F = \{f_i\}^m_{i = 1}$ is pre-distinguished, we let $\bar F =\{\bar f_i\}^{\bar m}_{i =1}$ be such that $\bar F \equiv {\cal L}_F$. 
 Then, by the definition of indicator sequence~$\N_i$ for $\bar f_i$, the above equation can be simplified as follows:
$$
\dot x_\sigma(t) = f_0(x_\sigma(t),\sigma) + \sum^{\bar m}_{i = 1}\sum_{l\in \N_i}\sum_{\rp \in {\cal P(l + 1)}} u_{i, \rp}(t) \rp(\sigma) \bar f_i(x_\sigma(t)), \quad \forall \sigma\in \Sigma.
$$  
%The arguments used in Section~\ref{ssec:paec} imply that   
To establish  ensemble controllability of the  above system (or more precisely, a truncated version after a certain order), it suffices to show that for any $i = 1,\ldots, \bar m$, the %following set of monomials: 
$\R$-span of monomials in $\sqcup_{l\in \N_i}{\cal P(l+1)}$ is dense in ${\rm C}^0(\Sigma)$. We prove this fact below.

We fix an~$i = 1,\ldots, \bar m$. 
By Prop.~\ref{lem:arithmeticseq}, the indicator sequence $\N_i$ contains an infinite arithmetic sequence, which we denote by $\{n_k\}_{k \ge 0}$ with $\delta := n_{k + 1} - n_k > 0$ for all $k \ge 0$. We next define functions on $\Sigma$ as follows: $$\bar \rho_s:=\rho^\delta_s, \quad \forall s = 1,\ldots, r.$$ By the assumption of Theorem~\ref{thm:mthm2}, the set $\{\rho^2_s\}^r_{s = 1}$ is a separating set and contains an everywhere nonzero function, say~$\rho_1$. It follows that $\{\bar \rho_s\}^r_{s = 1}$ is also a separating set with $\bar \rho_1$ an everywhere nonzero function. Thus, the subalgebra generated by $\{\bar \rho_s\}^r_{s = 1}$ is dense in ${\rm C}^0(\Sigma)$. Denote the subalgebra by $\bar {\cal S}$. 
Since $\rho_1$ is everywhere nonzero, the following set: $$\rho^{n_0 + 1}_1 \bar{\cal S}:= \left \{\rho^{n_0 + 1}_1\rp \mid \rp\in \bar{\cal S}\right \}$$ is dense in ${\rm C}^0(\Sigma)$ as well. On the other hand, the $\R$-span of $\sqcup_{l\in \N_i}{\cal P}(l + 1)$ contains $\rho_1^{n_0+1}\bar {\cal S}$ as a subset; indeed, if~$\rp$ is a monomial that can be expressed as $$\rp = \rho^{n_0 + 1}_1 \prod^r_{s = 1}\bar \rho^{k_s}_s$$  with $k_s \ge 0$, then $\rp \in {\cal P}(n_k + 1)$  where $k := \sum^r_{s = 1} k_s$. 
We have thus shown that the $\R$-span of $\sqcup_{l\in \N_i}{\cal P(l+1)}$ is dense in ${\rm C}^0(\Sigma)$.

We now establish item~2 of Theorem~\ref{thm:mthm2}. Let $\bar x_\Sigma(0)$ and $x_\Sigma(0)$ two initial profiles that are output equivalent.  The same arguments in Section~\ref{ssec:peo} can be used here to obtain the following fact: Let~$k\ge 0$ be an arbitrary integer. Let ${\rm w}$ be any word of length~$k$ and~$\rp$ be any monomial of degree~$k$. Then, for any $j = 1,\ldots, l$, we have 
\begin{equation}\label{eq:wcsamericas}
\int_\Sigma (f_{\rm w} \phi^j)(x_\sigma(0)) \rp(\sigma)  d\sigma = \int_\Sigma ( f_{\rm w} \phi^j)(\bar x_\sigma(0))\rp(\sigma)  d\sigma.
\end{equation}
Because $\Phi=\{\phi^j\}^l_{j = 1}$ is (weakly) pre-codistinguished to~$F$, 
we let $\bar \Phi = \{\bar \phi^j\}^{\bar l}_{j = 1}$ be such that $F_{\mathcal{W}}\Phi = \bar \Phi$.  
For each $j = 1,\ldots, \bar l$, we define a function ${\rm q}^j$ on~$\Sigma$ as follows:  
$${\rm q}^j(\sigma):= \bar \phi^j(x_\sigma(0)) - \bar \phi^j(\bar x_\sigma(0)).$$  
By the definition of indicator sequence~$\N^j$ for $\bar \phi^j$, we can simplify~\eqref{eq:wcsamericas} as follows:  
$$
\langle {\rm q}^j, \rp \rangle_{{\rm L}^2} = 0, \quad  \forall \rp\in \sqcup_{l\in \N^j}{\cal P}(l).  
$$
Note that the above expression holds for all $j = 1,\ldots, \bar l$.  
It now suffices to show that the $\R$-span of $\sqcup_{l\in \N^j}{\cal P}(l)$ is dense in ${\rm L}^2(\Sigma)$. This, again, follows from Prop.~\ref{lem:arithmeticseq}; indeed, since $\mathbb{N}^j$ contains an infinite arithmetic sequence, it follows by the same arguments (for $\mathbb{N}_i$) that the $\R$-span of $\sqcup_{l\in \N^j}{\cal P}(l)$ is dense in ${\rm C}^0(\Sigma)$. Because $\Sigma$ is compact, ${\rm C}^0(\Sigma)$ is dense in ${\rm L}^2(\Sigma)$. This  completes the proof.     
\end{proof}

\section{Existence of Distinguished Ensemble Systems}\label{sec:onliegroup}

We have shown in the previous section 
%Theorems~\ref{thm:mthm1}-\ref{thm:mthm3} 
that (weakly) jointly  distinguished vector fields $\{f_i\}^m_{i = 1}$ and functions $\{\phi^j\}^l_{j = 1}$ are key ingredients for an ensemble system to be approximately ensemble path-controllable and (weakly) ensemble observable.  
We address in the section the issue about the existence of these finely structured vector fields and functions for a given manifold~$M$.     
%\begin{problem}
%Given a manifold~$M$, are there  (weakly) jointly distinguished vector fields  and functions? 
%\end{problem} 
Amongst other things, we provide an affirmative answer for the case where $M$ is a connected, semi-simple Lie group:  

\begin{theorem}\label{thm:mthm4}
For any connected semi-simple Lie group $G$, there exist weakly jointly distinguished vector fields $\{f_i\}^m_{i = 1}$ and functions $\{\phi^j\}^l_{j = 1}$ on~$G$. Moreover, if $G$ has a trivial center, then $\{f_i\}^m_{i = 1}$ and $\{\phi^j\}^l_{j = 1}$ are jointly distinguished. 
\end{theorem}

Note that each Lie group $G$ admits the structure of real analytic manifold in a unique way such that multiplication and the inversion are real analytic. In this case the exponential map $\exp: \g\to G$ is also real analytic (see~\cite[Prop.~1.117]{knapp2013lie}). We also note that by Lemmas~\ref{lem:topologicalinv} and~\ref{lem:diffeocodist}, if $G$ admits (weakly) jointly distinguished vector fields and functions, then so does any manifold diffeomorphic to~$G$.

{\em Sketch of proof and organization of the section.} 
The proof of the above existence result is constructive.  We will show that for every connected, semi-simple Lie group~$G$, there exists a distinguished set of left- (or right-) invariant  vector fields. Moreover, there is a selected set of matrix coefficients associated with the adjoint representation such that it is codistinguished to the set of left- (or right-) invariant vector fields.

%there exist a set of distinguished left-invariant (or right-invariant) vector fields on~$G$ and a set of functions, obtained as the matrix coefficients of the adjoint representation, which is codistinguished to the set of left-invariant (or right-invariant) vector fields.  

%We provide below an outline of the remaining section.

More specifically, we address in Section~\S\ref{ssec:distinguishedsetofliealgebra} the existence of distinguished left- (or right-) invariant  vector fields over~$G$. 
Since the set of left-invariant vector fields has been identified with the Lie algebra~$\g$, the existence problem will naturally be addressed on the Lie algebra level. For that, we have recently shown in~\cite{chen2017distinguished} that every semi-simple real Lie algebra admits a distinguished set (see Def.~\ref{def:distinguishedlieset} below). We review in the subsection the result established in~\cite{chen2017distinguished} and provide a sketch of the proof.  
 
We next address in Subsections~\S\ref{ssec:matrixcoeffi} -- \S\ref{ssec:proofofcodist} the existence of (weakly) codistinguished functions to a given set of left- (or right-) invariant vector fields. In particular, we will leverage connections with Lie algebra representation. To see this, we recall that if $F:=\{f_i\}^m_{i = 1}$ and $\Phi:=\{\phi^j\}^l_{j = 1}$ are (weakly) jointly  distinguished, then the following map: 
$$(f, \phi)\in \mathbb{L}_F\times \mathbb{L}_\Phi\mapsto f\phi\in \mathbb{L}_\Phi$$ 
is a finite dimensional Lie algebra representation of~$\mathbb{L}_F$ on~$\mathbb{L}_\Phi$, where $\mathbb{L}_F$ and $\mathbb{L}_\Phi$ are linear spans of $F$ and $\Phi$, respectively. We will use the fact and propose in Section~\S\ref{ssec:matrixcoeffi} a constructive approach for generating  codistinguished functions.  

Then, in Sections~\S\ref{ssec:adjointrepresentationcodist} and~\S\ref{ssec:proofofcodist}, we focus on a special Lie group representation, namely the adjoint representation, and demonstrate that there indeed exists a set of matrix coefficients which is (weakly) codistinguished to a set of distinguished left- (or right-) invariant vector fields. In the case where $G$ is a matrix Lie group, such a set of codistinguished functions (matrix coefficients) can be expressed as follows: $$\{\phi^{ij}(g):=\tr(gX_j g^{-1} X_i^\top)\}^m_{i,j= 1},$$
where $\{X_i\}^m_{i = 1}$ is a selected set of matrices in the associated matrix Lie algebra.  
Note, in particular, that the set of functions $\{\phi^{ij}\}^3_{i,j, = 1}$ on $\SO(3)$ introduced in Example~\ref{emp:example1} takes exactly the same form (with $g^\top = g^{-1}$).   

Finally, in Section~\S\ref{ssec:homogeneousspace}, we address the problem about how to translate distinguished vector fields and codistinguished functions on a Lie group $G$ to any of its homogeneous spaces. For distinguished vector fields, we show that there exists a canonical way of translation. However, for the codistinguished functions, the question of translation remains open; we provide a few preliminary results there. At the end of the subsection, we combine the results together and investigate a simple example in which the unit sphere $S^2 =\SO(3)/\SO(2)$ is considered.   %   

\subsection{Distinguished sets of semi-simple real Lie algebras}\label{ssec:distinguishedsetofliealgebra}
Let $G$ be a semi-simple Lie group, and $\g$ be its Lie algebra. We address in the subsection the existence of distinguished left- (or right-) invariant  vector fields over~$G$. 
Recall that for any $X\in \g$, we have used $L_X$ (resp. $R_X$) to denote the corresponding left- (resp., right-) invariant vector field.  
We also recall that for any $X, Y\in \g$, 
$[L_X, L_Y] = L_{[X, Y]}$ and $[R_X, R_Y] = - R_{[X, Y]}$. It thus suffices to investigate the existence of a ``distinguished set'' on the Lie algebra level. 
Precisely, we first have the following definition: 

\begin{definition}[\cite{chen2017distinguished}]\label{def:distinguishedlieset}
Let $\g$ be a semi-simple real Lie algebra. A finite spanning set $\{{X}_i\}^m_{i = 1}$ of~$\g$ is {\bf distinguished} if 
 for any ${X}_i$ and ${X}_j$, there exist an ${X}_k$ and a real number $\lambda$ such that 
\begin{equation}\label{eq:distlieset}
[X_i, X_j] = \lambda {X}_k.
\end{equation}
Conversely, for any ${X}_k$, there exist ${X}_i$, ${X}_j$, and a {\em nonzero} $\lambda$ such that~\eqref{eq:distlieset} holds.
\end{definition}

%Note that since $\g$ is semi-simple, we have $[\g, \g] = \g$. It follows that if $\{X_i\}^m_{i = 1}$ is distinguished, then for any ${X}_k$, there exist ${X}_i$, ${X}_j$, and a {\em nonzero} $\lambda$ such that~\eqref{eq:distlieset} holds. 

Note that the cardinality of a distinguished set $\{X_i\}^m_{i = 1}$ is, in general, greater than the dimension of~$\g$, i.e., the spanning set $\{X_i\}^m_{i =1}$ may contain a basis of~$\g$ as its proper subset.    
We have investigated in~\cite{chen2017distinguished} the existence of a distinguished set of an arbitrary semi-simple real Lie algebra: %We have the following existence result: 

\begin{proposition}\label{thm:basisliealg}
Every semi-simple real Lie algebra admits a distinguished set.  
\end{proposition} 

The proposition then implies that every semi-simple Lie group admits a set of distinguished left- (or right-) invariant vector fields. 
Since the proposition will be of great use in the paper, we outline below a constructive approach for generating a desired distinguished set.  A complete proof can be found in~\cite{chen2017distinguished}. The proof leverages the structure theory of semi-simple real Lie algebras. 
 A reader not interested in the constructive approach can skip the remainder of the subsection.   

{\em Sketch of proof.} Recall that $\ad(X)(\cdot) := [\cdot, X]$ is the adjoint representation. Denote by $B(X, Y) := \tr(\ad_X\ad_Y)$ the Killing form. Let $\h$ be a Cartan subalgebra of~$\g$, and $\g^\C$ (resp. $\h^\C$) be the complexification of~$\g$ (resp. $\h$).  We let~$\Delta$ be the set of roots. For each $\alpha\in \Delta$, we let $h_\alpha\in \h^\C$ be such that 
$\alpha(H) =  B(h_\alpha, H)$ for all $H\in \h^\C$. Denote by $\langle \alpha, \beta\rangle := B(h_\alpha, h_\beta)$, which is an inner-product defined over the $\R$-span of $\Delta$. %corollary~2.38 by knapp.
We denote by $|\alpha|:= \sqrt{\langle \alpha, \alpha \rangle}$ the length of $\alpha$. 
Let $H_\alpha :=  \nicefrac{2h_\alpha}{ |\alpha|^2 }$. 
For a root $\alpha\in \Delta$, let~$\g_\alpha$ be the corresponding root space (as a one-dimensional subspace of~$\g^\C$ over~$\C$).

Suppose, for the moment, that one aims to obtain a distinguished set for the semi-simple {\em complex} Lie algebra $\g^\C$; then, with slight modification, such a set can be obtained via the Chevalley basis~\cite[Chapter VII]{JH:72}, which we recall below: 

\begin{lemma}\label{theorem:rootspacedecomp}
 There are ${X}_\alpha\in \g^\C_\alpha$, for $\alpha\in \Delta$, such that the following hold:
 \begin{enumerate}
 \item[1)] For any $\alpha\in \Delta$, we have $[{X}_\alpha, {X}_{-\alpha}] = H_\alpha$. %with $B({X}_\alpha, {X}_{-\alpha}) = \nicefrac{2}{|\alpha|^2}$.
 \item[2)] For any two non-proportional roots $\alpha, \beta$, we let $\beta + n\alpha$, with $-q \le n \le p$, be the $\alpha$-string that contains $\beta$. Then,  
 $$ 
 [{X}_\alpha, {X}_\beta] = \left\{
 \begin{array}{ll}
 c_{\alpha, \beta} {X}_{\alpha  + \beta} & \mbox{ if } \alpha + \beta \in \Delta, \\
 0 & \mbox{ otherwise},
 \end{array}
 \right.
 $$
 where  $c_{\alpha, \beta}\in \mathbb{Z}$ with $c^2_{\alpha,\beta} = (q + 1)^2$. 
 \end{enumerate}\,
\end{lemma}

We also note that for any $\alpha, \beta\in \Delta$, $[H_\alpha, {X}_\beta] = \nicefrac{2\langle \alpha, \beta\rangle}{|\alpha |^2} {X}_\beta$ and, moreover, $\nicefrac{2\langle \alpha, \beta\rangle}{|\alpha |^2}\in \mathbb{Z}$. 
It thus follows from Lemma~\ref{theorem:rootspacedecomp} that $$A:= \{ H_\alpha, {X}_\alpha, {X}_{-\alpha} \mid \alpha\in\Delta \}$$ is a distinguished set of~$\g^\C$.  The above arguments have the following implications:

\begin{enumerate}
\item[\em 1)]  A semi-simple complex Lie algebra can also be viewed as a Lie algebra over $\R$. 
We call any such real Lie algebra {\em complex}~\cite[Chapter~VI]{knapp2013lie}. In particular, if the real Lie algebra $\g$ is complex, then the  $\R$-span of $A \cup \mathrm{i}A$, with $A$ defined above is~$\g$. Moreover, since the coefficients $\nicefrac{2\langle \alpha, \beta\rangle}{|\alpha|^2}$ and $c_{\alpha, \beta}$ are all integers (and hence real), the set $A\cup \mathrm{i} A$ is a distinguished set of $\g$.    
\item[\em 2)]  If the Lie algebra $\g$ is obtained as the $\R$-span of $A$ (i.e., $\g$ is a {\em split real form} of $\g^\C$), then $A$ is a distinguished set of~$\g$. For example, every special linear Lie algebra~$\sl(n, \R)$ for $n \ge 2$ can be obtained in this way.     
\end{enumerate}
Thus, the technical difficulty for establishing Prop.~\ref{thm:basisliealg} lies in the case where~$\g$ is neither complex nor a split real form of~$\g^\C$. We deal with such a case below. 

First, recall that a {\em Cartan involution} $\theta: \g \to \g$ is a Lie algebra automorphism, with $\theta^2 = \operatorname{id}$, and moreover, the symmetric bilinear form $B_\theta$, defined as $$B_\theta(X, Y):= -B(X, \theta Y),$$ is positive definite on $\g$. One can extend $\theta$ to $\g^\C$ by $\theta(X + \mathrm{i} Y) = \theta X + \mathrm{i}\theta Y$.

Let the Cartan subalgebra $\h$ be chosen such that it is $\theta$-stable, i.e., $\theta\h = \h$. 
Decompose $ \g=\mathfrak{k}\oplus \pp $, where $ \mathfrak{k}$ (resp. $\pp$) is the $+1$-eigenspace (resp. $-1$-eigenspace) of~$\theta$. Their complexifications will, respectively, be denoted by~$\mathfrak{k}^\C$ and~$\pp^\C$. 
%Since $\theta$ is a Lie algebra automorphism,
%\begin{equation}\label{eq:Z2grading}
%[\mathfrak{k},\mathfrak{k}]\subset \mathfrak{k}, \qquad [\pp,\pp]\subset \mathfrak{k}, \qquad [\mathfrak{k},\pp]\subset \pp.
%\end{equation} 
%If $ X \in \g_0 $, then $ \ad_X^*(\cdot)=-\ad_{\theta X}(\cdot) $, where the adjoint is relative to $ B_{\theta} $. 
We further decompose $ \h=\t\oplus \a $, where $ \t $ and $ \a $ are subspaces of  $ \mathfrak{k} $ and $\pp$, respectively.  It is known that the roots in $\Delta$  are real on $ \a\oplus i\t $.  
% be a Cartan subalgebra of $ \g_0 $, with complexification $ \h=\t\oplus \a $ (all the roots are real on $ \a\oplus i\t $, see Corollary~6.49 of~\citep{knapp2013lie}), where $ \a_0 $ and $ \t_0 $ are abelian subspaces of $ \pp_0 $ and $ \t_0 $, respectively. Let now $ \Delta=\Delta(\g,\h) $ be the set of roots for the $ \g $, the complexification of $ \g_0 $;  
%For later use, it is worth mentioning that one can impose a positive ordering on $ \Delta $ that takes $ i\t_0 $ before $ \a_0 $; we denote $ \Delta $ equipped with this ordering by $ \Delta^+ $ and refer to it as a positive system. 
%Since any Cartan subalgebra is conjugate via an inner-automorphism to a $\theta$-stable Cartan subalgebra~\cite[Proposition 6.59]{knapp2013lie}, from now on, 
%we assume that $\h_0 $ is stable under $ \theta $. 

We say that $ \h$ is \emph{maximally compact} when the dimension of $ \t $ is as large as possible. %If a minimally compact $\h_0$ is such that  $\h_0\subseteq \pp_0$, then $\g_0$ is a {\em split real form} of $\g$. Every semi-simple complex Lie algebra $\g$ has a split real form, which is unique up to isomorphism~\cite[Theorem~5.10]{helgason2001differential}.
Given any $\theta$-stable Cartan subalgebra $\h$, one can obtain a maximally compact Cartan subalgebra by recursively applying the Cayley transformation~\cite[Sec.~VI-7]{knapp2013lie}. We assume in the sequel that $\h$ is maximally compact. We say that a root is \emph{imaginary} (resp. \emph{real}) if it takes imaginary (resp. real) value on $ \h$ and, hence,  vanishes over $\a$ (resp. $\t$). If a root is neither real nor imaginary, then it is said to be {\em complex}. It is known~\cite[Proposition~6.70]{knapp2013lie} that  if $\h$ is maximally compact, then there are no real roots and vice versa.

Note that if $\alpha$ is a root, then $\theta\alpha$ is also a root, defined as $(\theta\alpha)(H) := \alpha (\theta H)$ for any $H\in \h^\C$; indeed, if we let $X_\alpha\in \g_\alpha^\C$, then 
$$
[H, \theta X_\alpha] = \theta[\theta H, X_\alpha] = \alpha(\theta H)\theta X_\alpha = (\theta\alpha) (H) \theta X_\alpha.
$$
%In particular, $\theta X_\alpha\in \g_{\theta \alpha}$. 
Since $\theta$ is Lie algebra automorphism, $B(X, Y) = B(\theta X, \theta Y)$ for all $X, Y\in \g^\C$.  In particular, $\theta\alpha (H) = B(H_\alpha, \theta H) = B(\theta H_\alpha,  H)$, which implies that $H_{\theta \alpha} = \theta H_\alpha$. 
Note that if $\alpha$ is imaginary (which vanishes over $\mathfrak{a}$), then $\theta \alpha = \alpha$. This, in particular, implies that~$\g^\C_\alpha$ is $\theta$-stable. Since $\g^\C_\alpha$ is one dimensional (over $\C$), it must be contained in either~$\mathfrak{k}^\C$ or~$\pp^\C$.  An \emph{imaginary root} $ \alpha$ is said to be {\em compact} (resp. {\em non-compact}) if $ \g^\C_{\alpha} \subseteq \mathfrak{k}^\C $ (resp. $ \g^\C_{\alpha} \subseteq\pp^\C $).  It follows that if  $\alpha$  is compact (resp. non-compact), then $\theta X_\alpha = X_\alpha$ (resp. $\theta X_\alpha = -X_{\alpha}$). Furthermore, one can rescale the $X_\alpha$'s, if necessary, so that Lemma~\ref{theorem:rootspacedecomp} holds and $\theta X_\alpha  = X_{\theta\alpha}$ for $\alpha$ a complex root.  

%we impose the equivalence relation ``$\equiv$'' on the collection of subsets of $\g$: For two subsets $A$ and $A'$, we write $A \equiv A'$ if for any $X\in A$, there is an $X'\in A$ such that $X' = c X$ for some nonzero $c\in \R$, and vice versa. 

For a subset $S\subset \g$, we let ${\cal L}_S$ be the collection of Lie products generated by~$S$. Similarly, we say that ${\cal L}_S$ is projectively finite if there exists a finite subset $\bar S$ of~$\g$ such that ${\cal L}_S \equiv \bar S$. Further, we say that the set~$S$ is {\bf pre-distinguished} if $\bar S$ is a distinguished set of $\g$ (compared with Def.~\ref{def:predistinguishedvecf}). We now have the following fact: 

%[++++++]
% For a subset $S\subset \g$, we let ${\rm Gen}(S)$ be the subset of $\g$ generated by $S$, i.e., it is the smallest subset of $\g$ which  contains~$S$ and is closed under Lie bracket, and let $\overline S$ be a minimal subset of $\g$, which is equivalent to ${\rm Gen}(S)$ (i.e., $\overline S \equiv {\rm Gen}(S)$). Similarly, we have the following definition:  

%\begin{definition}[\cite{chen2017distinguished}]
%A subset $S\subseteq \g$ is {\bf pre-distinguished} if $\overline S$ is finite and is distinguished. 
%\end{definition}

 \begin{proposition}\label{prop:1}
 Let $\g$ be a simple real Lie algebra. Suppose that $\g$ is neither complex nor a split real form of~$\g^{\C}$; then, there are ${X}_\alpha\in \g^{\C}_\alpha$, for $\alpha\in \Delta$, such that the items of Lemma~\ref{theorem:rootspacedecomp}  are satisfied and the following set:  
$$S:=\left \{Y_\alpha:={X}_{\alpha} - \theta {X}_{-\alpha}, \quad Z_\alpha: ={\rm  i} ({X}_\alpha + \theta {X}_{- \alpha}) \mid \alpha \in \Delta\right \}$$
belongs to $\g$. Furthermore, the following hold:
\begin{enumerate}
\item[\em 1)] If the underlying root system of $\g$ is not $G_2$, then the set $S$ is pre-distinguished. 
\item[\em 2)] If the underlying root system of $\g$ is $G_2$, then $\g$ is the compact real form of $\g^{\C}$. Decompose $\Delta = \Delta_{\rm short} \cup \Delta_{\rm long}$ where $\Delta_{\rm short}$ (resp. $\Delta_{\rm long}$) is comprised of short (resp. long) roots. Then, the following set is pre-distinguished:
$$
S^* := \bigcup_{\gamma \in \Delta_{\rm long}} \{Y_\gamma, Z_\gamma\} \cup \bigcup_{\tiny
\begin{array}{l}
\alpha, \beta\in \Delta_{\rm short} \vspace{1pt}\\
\mbox{and } \alpha \neq \pm \beta
\end{array}
}[\{Y_\alpha, Z_\alpha\}, \{Y_\beta, Z_\beta\}].
$$
\end{enumerate} 
  \end{proposition} 

We refer the reader to~\cite{chen2017distinguished} for a complete proof. 
It follows from Prop.~\ref{prop:1} that every semi-simple real Lie algebra admits a distinguished set. This establishes Prop.~\ref{thm:basisliealg}.  %

Note that for a given simple Lie algebra~$\g$, there may exist multiple distinguished subsets of $\g$ such that any two of these sets are {\em not} related by a Lie algebra automorphism. More precisely, we say that two subsets $A$ and $A'$ are of the same class if there is a Lie algebra automorphism $\kappa: \g \to \g$ such that $\kappa(A)\equiv A'$. We defer the analysis in another occasion, but provide below a simple example for the case where $\g = \sl(2, \R)$: 

\begin{example}\label{exmp:sl2case}{\em 
Let $\sl(2, \R) = \{X\in \R^{2\times 2} \mid \tr(X) = 0\}$. Consider two subsets $A$ and $A'$ defined as follows:
$$
A := 
\left\{    
H:= \begin{bmatrix}
1 & 0 \\
0 & -1
\end{bmatrix}, \quad
X:=\begin{bmatrix}
0 & 1 \\
0 & 0
\end{bmatrix}, \quad
Y:=\begin{bmatrix}
0 & 0 \\
1 & 0
\end{bmatrix}
\right \},
$$
and 
$$
A' := 
\left\{    
H' := \begin{bmatrix}
1 & 0 \\
0 & -1
\end{bmatrix}, \quad
X' := \begin{bmatrix}
0 & 1 \\
1 & 0
\end{bmatrix}, \quad
Y' := \begin{bmatrix}
0 & 1 \\
-1 & 0
\end{bmatrix}
\right \}.
$$
By computation, we have that %One verifies that $A$ and $A'$ are distinguished subsets of $\sl(2,\R)$:
$$
\begin{array}{lll}
%\mbox{\rm For } A: & 
[H, X] = 2X, & \quad [H, Y] = -2Y,  & \quad [X, Y] = H;\\ 
{[H', X']} = 2Y', &  \quad [H', Y'] = 2X', & \quad [X', Y'] = -2H'.\\ 
\end{array}
$$
Thus, both $A$ and $A'$ are distinguished. But, they are not of the same class.  }
\end{example}

%A problem we pose here and will pursue in the future is thus the following: Given a simple Lie algebra~$\g$, classify all the classes of distinguished subsets of~$\g$. 

\subsection{Matrix coefficients as codistinguished functions}\label{ssec:matrixcoeffi}
Let $\{X_i\}^m_{i = 1}$ be a distinguished set of~$\g$. We address in the subsection the existence of (weakly) co-distinguished functions on~$G$ to the set of left- (resp., right-) invariant vector fields $\{L_{X_i}\}^m_{i = 1}$ (resp. $\{R_{X_i}\}^m_{i = 1}$). Because of the symmetry, the focus will be mostly on the functions codistinguished to the {\em left-invariant} vector fields. We provide a remark at the end of the subsection to address the existence of codistinguished functions to the right-invariant vector fields. 
%Let ${\rm C}^\omega(G)$ be the infinite dimensional vector space of smooth functions over~$G$. 

To proceed, we first recall that the so-called {\bf right-regular representation} of $G$ on ${\rm C}^\omega(G)$, denoted by $r: G\times {\rm C}^\omega(G)\to {\rm C}^\omega(G)$, is defined by 
$$
(x, \phi) \in G\times {\rm C}^\omega(G) \mapsto  (r(x)\phi)(g) := \phi(gx^{-1}).
$$  
Correspondingly, the induced Lie algebra representation $r_*$ is the negative of the Lie derivative along a left-invariant vector field, i.e., 
$
r_*(X)\phi = -L_X \phi
$.  
Note, in particular, that if $\Phi = \{\phi^j\}^l_{j = 1}$ is codistinguished to $\{L_{X_i}\}^m_{i = 1}$, then $r_*|_{\mathbb{L}_\Phi}$ is a finite dimensional representation of $\g$ on $\mathbb{L}_\Phi$; indeed, we have that 
\begin{multline*}
r_*([X_i, X_{i'}]) \phi^j= r_*(X_{i'})r_*(X_i)\phi^j - r_*(X_{i})r_*(X_{i'})\phi^j = \\
L_{X_{i'}} L_{X_i} \phi^j - L_{X_{i}} L_{X_{i'}} \phi^j = -L_{[X_i,X_{i'}]}\phi^j \in \mathbb{L}_\Phi.
\end{multline*} 
Thus, in order to find a set of codistinguished functions to $\{L_{X_i}\}^m_{i = 1}$,  our  strategy is comprised of two steps as outlined below: 
\begin{enumerate}
\item[\em 1)] Construct a finite dimensional subspace~$\mathbb{L}$ of ${\rm C}^\omega(G)$ such that it is closed under~$r$ so that $r_* |_{\mathbb{L}}$ will be a  Lie algebra representation of~$\g$ on~$\mathbb{L}$; 
\item[\em 2)] Find a finite subset $\Phi = \{\phi^j\}^l_{j = 1}$ out of the space~$\mathbb{L}$ such that it is codistinguished to a certain set of left-invariant vector fields~$\{L_{X_i}\}^m_{i = 1}$.
\end{enumerate}     
We now address, one-by-one, the above two steps.   

Our approach for the first step about constructing a finite dimensional subspace~$\mathbb{L}$ of ${\rm C}^\omega(G)$ is to use matrix coefficients associated with a Lie group representation. Specifically, 
we consider an arbitrary  analytic representation~$\pi$ of~$G$ on a finite dimensional inner-product space~$(V, \langle\cdot, \cdot \rangle)$. Let $\{v_i\}^p_{i = 1}$ be any spanning subset of~$V$. We next define a set of matrix coefficients as follows: 
\begin{equation}\label{eq:matrixcoeffi}
\pi^{ij}(g):= \langle v_i, \pi(g) v_j\rangle \in {\rm C}^\omega(G), \quad 1\le i, j\le p.  
\end{equation}
%and correspondingly,  
%$
%\pi(g) := [\pi^{ij}(g)]_{ij}
%$ with respect to the basis $\{v_1,\ldots, v_k\}$.  
Then, we let $\mathbb{L}_\pi$ be a finite dimensional subspace of ${\rm C}^\omega(G)$ spanned by $\pi^{ij}$: 
$$\mathbb{L}_\pi:=  \left\{ \sum^p_{i,j = 1} c_{ij} \pi^{ij}(g) \mid c_{ij} \in \R \right \}.$$
The following fact is certainly known in the literature. But, for completeness of presentation, we provide a proof after the statement:

\begin{lemma}\label{lem:matrixcoeff}
The vector space $\mathbb{L}_\pi$ is closed under~$r(x)$  
for all $x\in G$, i.e., for any $\phi \in \mathbb{L}_\pi$, 
$r(x) \phi \in\mathbb{L}_\pi$.   
Thus, $r|_{\mathbb{L}_\pi}$ (resp. $r_*|_{\mathbb{L}_\pi}$)
%$$(x, \phi)\in G\times \Phi_{\pi} \mapsto r_x\phi \in \Phi_{\pi} $$ 
is a representation of $G$ (resp. $\g$) on $\mathbb{L}_\pi$. 
\end{lemma}

\begin{proof} 
The lemma follows directly from computation. For any $x\in G$ and any $g\in G$, 
$$
(r(x)\pi^{ij})(g) = \pi^{ij}(gx^{-1}) =  \langle v_i, \pi(gx^{-1}) v_j \rangle =  \langle v_i,  \pi(g)\pi(x^{-1}) v_j  \rangle.
$$
Since $\{v_1,\ldots, v_p\}$ spans $V$, there exist real coefficients $c_{lk}$'s such that 
$$\pi(x^{-1}) v_j = \sum^p_{l,k= 1} c_{lk}\langle v_l, \pi(x^{-1})v_j \rangle v_k = \sum^p_{l,k= 1} c_{lk} \pi^{lj}(x^{-1}) v_k.$$ 
It then follows that 
$$
(r(x)\pi^{ij})(g) = \sum^p_{l,k =1}\left(c_{lk} \pi^{lj}(x^{-1})  \right )\pi^{ik}(g),
$$
which implies that $r(x)\pi^{ij}$ is a linear combination of $\pi^{ik}$ for $k = 1,\ldots, p$.  
\end{proof}

\begin{remark} {\em 
Let $G$ be compact, and the representation $\pi$ of $G$ on $V$ be irreducible and unitary. Further, let $\{v_i\}^p_{i = 1}$ be an orthonormal basis of $V$. Then, by Peter-Weyl Theorem~\cite[Sec.~IV]{knapp2013lie},  the matrix coefficients $\{\pi^{ij}\}^p_{i, j = 1}$ are linearly independent. }
\end{remark}

We now address the second step of our strategy about finding a finite subset $\{\phi^j\}^l_{j = 1}$ out of~$\mathbb{L}_\pi$ so that it is codistinguished to a given set of left-invariant vector fields $\{L_{X_i}\}^m_{i = 1}$. To proceed, we first have the following definition as a dual to Def.~\ref{def:distinguishedlieset}: 

\begin{definition}\label{def:codishv}
Let $\pi$ be a finite dimensional representation of $G$ on $V$, and $\pi_*$ be the corresponding Lie algebra representation. A spanning set $\{v_j\}^p_{j = 1}$ of $V$ 
is {\bf codistinguished} to a subset $\{X_i\}^m_{i =1}$ of~$\g$ if it satisfies the following properties: 
\begin{enumerate}
\item[\em 1)] The set of one-forms $\{d\pi^{ij}_e\}^p_{i,j = 1}$ spans $T^*_eG \approx \g^*$.
\item[\em 2)] For any $X_i$ and $v_j$, there exist a $v_k$ and a real number $\lambda$ such that 
\begin{equation}\label{eq:codistrep}
\pi_*(X_i) v_j = \lambda v_k;
\end{equation}
conversely, for any $v_k$, there exist $X_i$, $v_j$, and a {\em nonzero} $\lambda$ such that~\eqref{eq:codistrep} holds.
\item[\em 3)] For any $g, g'\in G$, if $\pi^{ij}(g) = \pi^{ij}(g')$ for all $1\le i, j \le p$, then $g = g'$.
\end{enumerate}
If only {\em 1)} and {\em 2)} hold, then $\{v_j\}^p_{j =1}$ is {\bf weakly codistinguished} to $\{X_i\}^m_{i = 1}$. 
\end{definition}

%Note that the subset $\{X_i\}^m_{i = 1}$ in the above definite is not necessarily distinguished. 
With the above definition, we now have the following fact: %that will be useful for finding a codistinguished set out of $\mathbb{L}_\pi$: 

\begin{lemma}\label{lem:langlang}
If $\{v_j\}^p_{j = 1}$ is codistinguished to $\{X_i\}^m_{i = 1}$, then the set of matrix coefficients $\{\pi^{ij}\}^p_{i,j = 1}$ is codistinguished to the set of left-invariant vector fields $\{L_{X_i}\}^m_{i = 1}$.
\end{lemma}

\begin{proof}
We show below that if $\{v_j\}^p_{j = 1}$ is codistinguished to $\{X_i\}^m_{i = 1}$, then the three items of Def.~\ref{def:codisth} are satisfied. 
\begin{enumerate}
\item[\em 1)] For the first item of Def.~\ref{def:codisth}, we show that for any $g\in G$, the one-forms $\{d\pi^{ij}_g\}^p_{i, j = 1}$ span $T^*_g G$. 
With slight abuse of notation, we write 
%Note that the two sets one-forms $d\pi^{ij}_g$ can be identified with an element in $\g^*$ as follows: 
$$
d\pi^{ij}_g(X):= d\pi^{ij}_g(gX) =  -\langle v_i, \pi(g) \pi_*(X) v_j\rangle, \quad \forall X\in \g. 
$$ 
In this way, each one-forms $d\pi^{ij}_g$ can be viewed as an element in~$\g^*$. 
But then, the two subspaces of $\g^*$: $\span \{d\pi^{ij}_e\}^p_{i, j = 1}$ and $\span \{d\pi^{ij}_g\}^p_{i, j = 1}$ are isomorphic: 
%\ce{
%a <=>[\ce{\pi(g)}][\ce{\pi(g^{-1})}] b
%}
$$\sum^p_{i, j = 1} c_{ij}\langle v_i, \pi_*(\cdot )v_j \rangle \xrightleftharpoons[\pi(g^{-1})]{\pi(g)} \sum^p_{i, j = 1} c_{ij} \langle v_i, \pi(g)\pi_*(\cdot )v_j \rangle  $$
%\xmapsto{\pi(g^{-1})} \sum^p_{i, j = 1} c_{ij} \langle v_i, \pi_*(\cdot )v_j \rangle $$ 
The first item of Def.~\ref{def:codisth} then follows from the first item of Def.~\ref{def:codishv}.

\item[\em 2)] For the second item of Def.~\ref{def:codisth}, it suffices to show that if $\pi_*(X_i) v_j = \lambda v_k$, then  
$
L_{X_i}\pi^{qj} =  -\lambda \pi^{qk}
$ for any $q = 1,\ldots, p$. This holds because
$$
(L_{X_i}\pi^{qj} )(g) = -\langle v_q,  \pi(g)\pi_*(X_i) v_j \rangle = -\lambda \langle v_q, \pi(g) v_k\rangle = -\lambda \pi^{qk}(g).
$$
\item[\em 3)] The third item of Def.~\ref{def:codisth} directly follows from the third item of Def.~\ref{def:codishv}.
\end{enumerate}
This completes the proof.
\end{proof}

We have thus provided a constructive approach for generating a set of matrix coefficients that is (weakly) codistinguished to a given set of left-invariant vector fields.   
The same approach can be slightly modified to generate a set of functions codistinguished to a set of {\em right-invariant} vector fields. We provide  details in the following remark:    %which will be useful in Subsection~\S\ref{ssec:homogeneousspace} for generating codistinguished functions on homogeneous spaces.  

\begin{remark}\label{rmk:rightinvariantvec}{\em 
%To construct a set of codistinguished functions to the right-invariant vector fields $\{R_{X_i}\}^m_{i = 1}$, 
%Note that the Lie algebra $\g$ can be also identified with {\em right-invariant} vector fields $\{R_X\}_{X\in \g}$. We have  $R_{[X, Y]} = -[R_X, R_Y]$ for all $X, Y\in \g$. 
We first recall that the {\em left-regular representation} of~$G$ is given by 
$$
(x, \phi) \in G\times {\rm C}^\omega(G) \mapsto  (l(x)\phi)(g) := \phi(xg),
$$
The corresponding Lie algebra representation is given by $l_*(X)\phi = R_X\phi$. We again let~$\pi$ be a representation of $G$ on a finite-dimensional inner-product space~$(V, \langle \cdot, \cdot \rangle)$, and $\{v_i\}^p_{i = 1}$ be a spanning set of~$V$. 
%Let $\pi$ be a Lie algebra representation of $\g$ on an inner-product space $V$. 
%Similarly, to come up with a set of codistinguished functions to a set of right-invariant vector fields $\{R_{X_i}\}^m_{i = 1}$, 
We next define functions on~$G$ as follows: %(compare with~\eqref{eq:matrixcoeffi}): 
\begin{equation}\label{eq:modifiedmatrixco}
\tilde \pi^{ij}(g):= \langle v_i,  \pi(g^{-1}) v_j\rangle, \quad \forall 1\le i, j \le p.     
\end{equation}
Let $ \mathbb{L}_{\tilde \pi}$ be the $\R$-span of these $\tilde \pi^{ij}$. The same arguments in the proof of Lemma~\ref{lem:matrixcoeff} can be used here to show that $\mathbb{L}_{\tilde \pi}$ is closed under~$l(x)$ for all $x\in G$. %Moreover, Lemma~\ref{lem:langlang} still holds, i.e., 
Furthermore, if the set~$\{v_j\}^p_{j = 1}$ is chosen to be codistinguished to~$\{X_i\}^m_{i = 1}$, then similar arguments in the proof of Lemma~\ref{lem:langlang} can be used to show that the set of functions $\{\tilde \pi^{ij}\}^p_{i,j = 1}$ is codistinguished to the set of right-invariant vector fields $\{R_{X_i}\}^m_{i = 1}$. }
\end{remark}

In summary,  we have shown in the subsection %Lemma~\ref{lem:matrixcoeff} 
that a finite dimensional representation $\pi$ of~$G$ on an inner-product space $V$ can be used to generate a set of matrix coefficients codistinguished to a given set of left- (or right-) invariant vector fields provided that the assumption of Lemma~\ref{lem:langlang} is satisfied.

%gives rise to a Lie algebra representation $r_*|_{\Phi_\pi}$ (resp. $l_*|_{\tilde \Phi_\pi}$), where $\Phi_\pi$ (resp. $\tilde \Phi_\pi$) is a finite dimensional subspace of ${\rm C}^\omega(G)$ spanned by the matrix coefficients $\pi^{ij}$ defined in~\eqref{eq:matrixcoeffi} (resp.  $\tilde \pi^{ij}$ defined in~\eqref{eq:modifiedmatrixco}). We have further shown %in Lemma~\ref{lem:langlang} 
%that if the subset $\{v_j\}^p_{j = 1}$ is chosen such that it is codistinguished to a selected subset $\{X_i\}^m_{i = 1}$ of $\g$, then the resulting matrix coefficients $\{\pi^{ij}\}^p_{i,j = 1}$ (resp. $\{\tilde \pi^{ij}\}^p_{i,j = 1}$) is codistinguished to $\{L_{X_i}\}^m_{i = 1}$ (resp.  $\{R_{X_i}\}^m_{i = 1}$). 

\subsection{On the adjoint representation}\label{ssec:adjointrepresentationcodist}
We follow the discussions in the previous subsection, and consider here the adjoint representation of~$G$ on~$\g$, i.e., $\pi = \Ad$ and $V = \g$. We show that in this special case, there indeed exists a set of matrix coefficients (weakly) codistinguished to a distinguished set of left- (or right-) invariant vector fields.  

To proceed, we first recall that $B(X, Y) = \tr(\ad_X\ad_Y)$ is the Killing form, $\theta$ is a Cartan involution of~$\g$, and $B_\theta(X, Y) = -B(X, \theta Y)$ is an inner-product on~$\g$ (introduced in Section~\S\ref{ssec:distinguishedsetofliealgebra}). We also recall that by Prop.~\ref{thm:basisliealg}, there exists a distinguished set $\{X_i\}^m_{i = 1}$ out of~$\g$. We fix such a set in the sequel. Note, in particular, that by Def.~\ref{def:distinguishedlieset}, the distinguished set $\{X_i\}^m_{i = 1}$ spans~$\g$. Now, we follow the two-step strategy proposed in the previous section and define a set of matrix coefficients $\{\phi^{ij}\}^m_{i,j = 1}$ as follows: 
\begin{equation}\label{eq:defpsiij}
\phi^{ij} (g):= \Ad^{ij}(g) = B_\theta(\Ad(g)X_j, X_i ), \quad 1\le i, j \le m,  
\end{equation}
which is nothing but specializing~\eqref{eq:matrixcoeffi} to the case of adjoint representation. 
To further illustrate~\eqref{eq:defpsiij}, we take advantage of the following fact~\cite[Prop.~6.28]{knapp2013lie}:

\begin{lemma}\label{lem:realmatrices}
Every semi-simple real Lie algebra $\g$ is isomorphic to a Lie algebra of real matrices that is closed under transpose, with the Cartan involution $\theta$  carried to negative transpose, i.e., $\theta X = - X^\top$ for all $X\in \g$.
\end{lemma}

Note that under the above isomorphism, the $+1$-eigenspace~$\mathfrak{k}$ and the $-1$-eigenspace $\pp$ of~$\theta$ (introduced in Section~\S\ref{ssec:distinguishedsetofliealgebra}) correspond to the subspace of skew-symmetric matrices and the subspace of symmetric matrices, respectively.

We note that for a given semi-simple Lie algebra~$\g$ of real matrices, the Killing form $B(X, Y)$ is linearly proportional to $\tr(XY)$, i.e., $B(X, Y) = c \tr(X Y)$ for a real positive constant~$c$.  
%On the other hand, note that Lemma~\ref{lem:realmatrices} does not apply to the group level, i.e., not every semi-simple Lie group is isomorphic to a matrix Lie group. For example, the metaplectic group is a double cover of the symplectic group, yet is not a matrix Lie group.  
Now, suppose that $G$ is isomorphic to a matrix Lie group; then, it follows from Lemma~\ref{lem:realmatrices} that one can re-write~\eqref{eq:defpsiij} as follows:
\begin{equation}\label{eq:matrixrealization}
\phi^{ij} (g) = c\tr(g X_j g^{-1} X_i^\top).
\end{equation}
In particular, it generalizes the functions $\{\phi^{ij}\}_{1\le i, j \le 3}$ on $\SO(3)$ introduced in Example~\ref{emp:example1} to functions on an arbitrary matrix semi-simple Lie group. 
However, we shall note that not every semi-simple Lie group is isomorphic to a matrix Lie group. Nevertheless, the expression~\eqref{eq:defpsiij} is always valid.  %For example, the metaplectic group is a double cover of the symplectic group, yet is not a matrix Lie group. Thus, one cannot always write~\eqref{eq:defpsiij} in terms of~\eqref{eq:matrixrealization}. 

Recall that a center $Z(G)$ of a group~$G$ is defined such that if $z\in G$, then~$z$ commutes with every group element~$g$ of~$G$, i.e., 
$$Z(G):= \{z\in G \mid zg = gz, \quad \forall g\in G\}.$$ 
Let $\phi:=(\ldots,\phi^{ij},\ldots)$ be the collective of $\phi^{ij}$. Similarly, for any group element $g\in G$, we let $[g]_\phi$ be the pre-image of $\phi(g)$.  
We now have the following result:  

\begin{theorem}\label{thm:distinguished}
Let $\{X_i\}^m_{i = 1}$ be a distinguished  set of~$\g$.  
Then, the set of matrix coefficients $\{\phi^{ij}\}^m_{i,j = 1}$ defined in~\eqref{eq:defpsiij} is weakly codistinguished to  $\{L_{X_i}\}^m_{i = 1}$. Moreover, 
\begin{equation}\label{eq:finiteambiguityG}
[g]_\phi=\{gz \mid z\in Z(G)\}, \quad \forall g\in G.
\end{equation} 
In particular, $\{\phi^{ij}\}^m_{i,j = 1}$ is codistinguished to~$\{L_{X_i}\}^m_{i = 1}$ if and only if $Z(G)$ is trivial. 
\end{theorem} 

\begin{remark}\label{rmkk:distinguished}{\em 
Note that if one aims to construct a set of  functions codistinguished to the right-invariant vector fields $\{R_{X_i}\}^m_{i =1}$; then, by Remark~\ref{rmk:rightinvariantvec}, one can simply define functions as follows:   
\begin{equation}\label{eq:anotherversion}
\tilde \phi^{ij}(g) := B_\theta(\Ad(g^{-1})X_j, X_i), \quad \forall 1\le i, j \le m. 
\end{equation}
If one replaces in the statement $L_{X_i}$ with $R_{X_i}$ and correspondingly, $\phi^{ij}$ with $\tilde \phi^{ij}$, then Theorem~\ref{thm:distinguished} will still hold.  
}
\end{remark}

%Theorem~\ref{thm:mthm4} then directly follows from Theorem~\ref{thm:distinguished}.

%From Theorem~\ref{thm:distinguished}, the center $Z(G)$ can be thought as a measure of ambiguity of the estimation  ensemble system. 
Theorem~\ref{thm:mthm4} then follows from Prop.~\ref{thm:basisliealg} and Theorem~\ref{thm:distinguished}. We establish Theorem~\ref{thm:distinguished} in the next subsection. 
We provide below a brief discussion about the center of~$G$.  
First, note that since $\g$ is semi-simple, the center $Z(G)$  is discrete. If, further, $G$ is compact, then $Z(G)$ is finite.  %Centers of (semi-)simple real Lie groups have been extensively investigated.
 We also note that for each semi-simple real Lie algebra~$\g$, there corresponds a unique (up to isomorphism) connected group $G$ whose Lie algebra is~$\g$ and has a trivial center. So, for these Lie groups, $\{L_{X_i}\}^m_{i = 1}$ and $\{\phi^{ij}\}^m_{i,j = 1}$ are jointly  distinguished. 
We further present below the centers of a few commonly seen connected matrix Lie groups:
\begin{enumerate}
\item[\em 1)] If $G = \SU(n)$ is the special unitary group, then  $Z(G) = \left \{ zI \mid z^n = 1, z\in \C \right \}$.   
\item[\em 2)] If $G = \SL(n, \R)$ is the special linear group or if $G = \SO(n)$ is the special orthogonal group, then 
$$
Z(G) = 
\left\{
\begin{array}{ll}
	\{I\} & \mbox{ if } n \mbox{ is odd}, \\
	\{\pm I\} & \mbox{ if } n \mbox{ is even}. 
\end{array}
\right.
$$
%\{I\}$ for $n$ is odd and $Z(G) = \{\pm I\}$ if $n$ is even.
\item[\em 3)] Similarly, if $G = \SO^+(p,q)$ is the identity component of indefinite orthogonal group $\operatorname{O}(p,q)$ (e.g., the Lorentz group $\operatorname{O}(1,3)$), then 
$$
Z(G) = 
\left\{
\begin{array}{ll}
	\{I\} & \mbox{ if } p + q \mbox{ is odd}, \\
	\{\pm I\} & \mbox{ if } p + q \mbox{ is even}. 
\end{array}
\right.
$$
\item[\em 4)] If $G = \SP(2n, \R)$ is the symplectic group, then $Z(G) = \{\pm I_{2n}\}$. 
\end{enumerate}

\subsection{Analysis and proof of Theorem~\ref{thm:distinguished}}\label{ssec:proofofcodist}  
We establish in the subsection Theorem~\ref{thm:distinguished}. 
By Lemma~\ref{lem:langlang}, it suffices to show that the subset $\{X_i\}^m_{i = 1}$ of~$\g$ is codistinguished to itself with respect to the adjoint representation.  This fact will be established after a sequence of lemmas.  For convenience, we reproduce below the set of functions $\{\phi^{ij}\}^m_{i,j = 1}$ that will be investigated in the subsection: 
$$
\phi^{ij} (g):= \Ad^{ij}(g) = B_\theta(\Ad(g)X_j, X_i ), \quad 1\le i, j \le m.  
$$
We show below that the set $\{\phi^{ij}\}^m_{i,j = 1}$ satisfies the three items of Def.~\ref{def:codishv} under the assumption of Theorem~\ref{thm:distinguished}. The arguments we will use below generalize the ones used in Example~\ref{emp:example1}.  
 For the first item of Def.~\ref{def:codishv}, we have the following fact: 

\begin{lemma}\label{lem:checkitem1}
 The set of one-forms $\{d\phi^{ij}_e\}^m_{i, j = 1}$ spans the cotangent space $T^*_eG \approx \g^*$. 
\end{lemma}

\begin{proof}
%We identify the tangent space $T_e G$ with~$\g$. 
First, note that for any $X\in \g$, we have 
$$
d\phi^{ij}_e (X)= B_\theta ([X_j, X], X_i)  = - B([X_j, X], \theta X_i).
$$
Because the Killing form is adjoint-invariant, i.e.,   
$
%B(\Ad(g) X, \Ad(g) Y) = B(X, Y) \quad \mbox{ and } \quad 
B([X, Y], Z) =  B(X, [Y, Z]) 
$  for any $X, Y, Z\in\g$, 
it follows that 
$$
d\phi^{ij}_e (X) = -B([X_j, X], \theta X_i) = \\ - B(X, [\theta X_i,   X_j])   = B_\theta(X, [X_i, \theta  X_j]), 
$$
where the last equality holds because $\theta$ is a Lie algebra automorphism with $\theta^2 = {\rm id}$ and, hence, $\theta [\theta X_i, X_j] = [X_i, \theta X_j]$.  

For convenience, we let $Y_j := \theta X_j$ for all $j = 1,\ldots, m$. Since $\theta$ is a Lie algebra automorphism and $\{X_i\}^m_{i = 1}$ spans $\g$ (because it is distinguished), the subset $\{Y_j\}^m_{j = 1}$ spans~$\g$ as well. Also, note that  
 $\g$ is semi-simple and, hence, $[\g, \g] = \g$.  This, in particular, implies that $\{\hat X_{ij}:=[X_i, Y_j]\}^m_{i, j = 1}$ is also a spanning set of $\g$. It now remains to show that the set of one-forms  $\{B_\theta(\cdot, \hat X_{ij})\}^m_{i, j = 1}$ spans $\g^*$. 
%$$
%\span\left \{B_\theta(\cdot, \hat X_{ij})\right \} = \g^*.
%$$
%%where $\g^*$ is dual space of $\g$, i.e., it is the collection of linear forms on~$\g$. 
But, this follows from the fact that~$B_\theta$ is positive definite on~$\g$; indeed, a nondegenerate bilinear form induces a linear isomorphism between~$\g$ and~$\g^*$. 
Since the set $\{\hat X_{ij}\}^m_{i,j = 1}$ spans~$\g$, the set of one-forms  $\{B_\theta(\cdot, \hat X_{ij})\}^m_{i, j = 1}$ spans~$\g^*$. 
\end{proof}

%Lemma~\ref{lem:checkitem1} then implies that the set $\{\phi^{ij}\}^m_{i,j = 1}$ defined in~\eqref{eq:defpsiij} satisfies item~1 of Def.~\ref{def:codishv}. 
For the second item of Def.~\ref{def:codishv}, we have the following fact: 

\begin{lemma}\label{lem:damowang}
 If $[X_i, X_j] = \lambda X_k$, then $L_{X_i} \phi^{i'j} =- \lambda \phi^{i'k}$ for all $i' =1,\ldots, m$. 
\end{lemma} 

\begin{proof}
The lemma directly follows from computation: 
$$
(L_{X_i} \phi^{i'j})(g)  = B_\theta(\Ad(g)[X_j, X_i], X_{i'}) = -\lambda B_\theta(\Ad(g)X_k, X_{i'}) = -\lambda \phi^{i'k}(g),
$$  
which holds for any $g\in G$. 
\end{proof}

Combining Lemmas~\ref{lem:checkitem1} and~\ref{lem:damowang}, we have that the set of functions $\{\phi^{ij}\}^m_{i,j = 1}$ is weakly codistinguished to the set  of left-invariant vector fields~$\{L_{X_i}\}^m_{i = 1}$. 
Finally, for~\eqref{eq:finiteambiguityG}, we have the following fact: 

\begin{lemma}\label{lem:checkitem4}
If $\phi^{ij}(g) = \phi^{ij}(g')$ for all $1\le i,j \le m$, then $g^{-1}g' \in  Z(G)$ and vice versa.
\end{lemma}

\begin{proof}
We fix a $j = 1,\ldots, m$ and have the following:      
$$\phi^{ij}(g) - \phi^{ij}(g') = B_\theta(\Ad(g) X_j - \Ad(g') X_j, X_i) = 0, \quad \forall i = 1,\ldots, m.$$
Since $B_\theta$ is positive definite and $\{X_i\}^m_{i = 1}$ spans~$\g$, 
$\Ad(g) X_j = \Ad(g') X_j$. This holds for all $j = 1,\ldots, m$. Using again the fact that $\{X_j\}^m_{j = 1}$ spans~$\g$, we obtain  
$\Ad(g^{-1}g') X = X$ for all $X\in \g$. Thus, $g^{-1}g'$ belongs to the centralizer of the identity component of $G$. 
Since $G$ is connected, this holds if and only if $g^{-1}g'\in Z(G)$.
\end{proof}

%We have thus established Theorem~\ref{thm:distinguished}.

\subsection{On homogeneous spaces}\label{ssec:homogeneousspace}
 Let a group $G$ act on a manifold~$M$. We say that the group action is {\em transitive} if for any $x, y\in M$, there exists a group element $g\in G$ such that $g x = y$. Correspondingly, the manifold~$M$ is said to be a {\bf homogeneous space} of~$G$. Note that any homogeneous space can be identified with the space $G/H$ of left cosets $gH$ for $H$ a closed Lie subgroup of $G$. More specifically, we pick an arbitrary point $x\in M$, and let $H$ be the subgroup of $G$ which leaves~$x$ fixed (i.e., $H$ is the stabilizer of~$x$). Then, $M$ is diffeomorphic to~$G/H$, and we write $M \approx G/H$. The group action can  thus be viewed as a map by sending a pair $(g, g'H)$ to $gg'H$. We also note that the homogeneous space $M$ can be equipped with a unique analytic structure (see~\cite[Thm.~4.2]{helgason2001differential}).  
 
 We address in the subsection the existence of distinguished vector fields and codistinguished functions on homogeneous spaces of a semi-simple Lie group.  
We provide at the end of the subsection a simple example in which the unit sphere $S^2 \approx \SO(3)/\SO(2)$ is considered.

{\em Existence of distinguished vector fields.}
There is a canonical way of translating a distinguished set $\{X_i\}^m_{i = 1}$ of the Lie algebra~$\g$ to a distinguished set of vector fields over a homogeneous space~of $G$. 
Precisely, we define a map $\tau: \g \to \mathfrak{X}(M)$ as follows:   
Let $\exp: \mathfrak{g} \to G$ be the exponential map. 
For a given~$X\in \mathfrak{g}$, we define a vector field $\tau(X)\in \mathfrak{X}(M)$ such that for any $\phi \in {\rm C}^\omega(M)$, the following hold:
\begin{equation}\label{eq:liealghomo}
(\tau(X) \phi)(x) :=  \lim_{t\to 0} \frac{\phi(\exp(tX) x) - \phi(x)}{t}, \quad \forall x\in M.
\end{equation}
%If $M$ is embedded in $\R^n$, then the above definition can be simplified as follows:
%$$
%(\tau(X) \phi)(x) := \lim_{t\to 0} \frac{\exp(tX) x- x}{t}, \quad \forall x\in M.
%$$
Let $X_i$ and $X_j$ be any two elements in $\mathfrak{g}$. It is known~\cite[Chapter~2.3]{helgason2001differential}) that
\begin{equation}\label{eq:liealgetovecfld}
[\tau(X_i), \tau (X_j)] = -\tau([X_i, X_j]),
\end{equation}
which then leads to the following result: 

\begin{proposition}\label{pro:distonhomo}
Let $G$ be a semi-simple Lie group with $\g$ the Lie algebra, and~$M$ be a homogeneous space of~$G$.   
If $\{X_i\}^m_{i = 1}$ is a distinguished set of~$\mathfrak{g}$, then $\{\tau(X_i)\}^m_{i = 1}$ is a distinguished set of vector fields over~$M$.    
\end{proposition}

\begin{proof}
It suffices to show that $\{\tau(X_i)(x)\}^m_{i = 1}$ spans the tangent space $T_x M$ for all $x\in M$. Let $H$ be the stabilizer of~$x$, and $\h$ be the corresponding Lie algebra of~$H$. Since $\{X_i\}^m_{i = 1}$ spans $\g$, there must exist a subset of $\{X_i\}^m_{i = 1}$, say $\{X_i\}^{m'}_{i = 1}$, such that if we let $\mathfrak{m}:= \span\{X_i\}^{m'}_{i = 1}$, then $\g = \mathfrak{m} \oplus \h$. Moreover, the following map:
$$
(a_1, \ldots, a_{m'}) \in \R^{m'} \mapsto \exp\left (\sum^{m'}_{i = 1} a_i X_i \right )x \in M
$$
is locally a diffeomorphism around $0\in \R^{m'}$ to an open neighborhood of $x\in M$. This, in particular, implies that $\{\tau(X_i)(x)\}^{m'}_{i = 1}$ is a basis of the tangent space $T_x M$. 
\end{proof}

We provide below an example for illustration:

%\begin{Example}
%{\em
%Let $\SL(n)$ be the special linear group acting on $\R^n$, which  sends a pair $(L,x)$, with $L\in \SL(n)$ and $x\in \R^n$, to $Lx$. The group action is known to be transitive. Then, the map $\tau:\mathfrak{sl}(n)\to \mathfrak{X}(R^n)$   sends a matrix $L$ of zero trace to the vector field $\tau(L)(x) = Lx$. If we choose the set $A_1$ defined in~\eqref{eq:sln} to be our distinguished set for $\sl(n)$. Then,  
%the corresponding vector fields are given by 
%$$
%\tau(H_{ij})(x) = x_i e_i - x_j e_j, \quad 1\le i< j \le n,
%$$ 
%together with
%$$
%\tau(E_{ij})(x) = x_j e_i, \quad 1\le i\neq j\le n,
%$$
%all of which form a distinguished set of vector fields over the Euclidean space $\R^{n}$.  
%}
%\end{Example}

\begin{example}\label{exmp:soncase}{\em  
Let $\SO(n)$ act on $S^{n-1}$ by sending a pair $(g, x)\in  \SO(n)\times S^{n-1}$ to $g x\in S^{n-1}$. The group action is  transitive. The map $\tau: \so(n) \to \mathfrak{X}(S^{n-1})$ defined in~\eqref{eq:liealghomo} sends a skew-symmetric matrix $\Omega$ to the vector field $\tau(\Omega)(x) = \Omega x$. 
We define a set of skew-symmetric matrices as follows: $\Omega_{ij} := e_ie_j^\top - e_j e_i^\top$ for $1\le i < j\le n$.  
It directly follows from computation that $\{\Omega_{ij}\}_{ 1\le i < j \le n }$ is a distinguished set of $\mathfrak{so}(n)$. Moreover, we have that   
$\tau(\Omega_{ij})(x) = x_j e_i - x_i e_j$ where $x_i$'s are the coordinates of~$x$. By Prop.~\ref{pro:distonhomo}, $\{\tau(\Omega_{ij})\}_{1\le i < j\le n}$ is a set of distinguished vector fields over~$S^{n-1}$. }
\end{example}

{\em Existence of codistinguished functions.} We now discuss how to translate a set of codistinguished functions defined on a Lie group $G$ to a set of codistinguished functions on its homogeneous space~$M\approx G/H$. 
We consider below for the case where the closed subgroup $H$ is compact.

 We say that a function $\phi\in {\rm C}^\omega(G)$ is {\em $H$-invariant} if for any $g\in G$ and $h\in H$, we have $\phi(gh) = \phi(g)$. In particular, if $\phi$ is $H$-invariant, then one can simply define a function $ \psi$~on $M$ by $ \psi(g H) := \phi(g)$. This is well defined because if $gH =g'H$, then $g^{-1}g'$ belongs to $H$ and, hence,  $\phi(g) = \phi(g g^{-1}g') = \phi(g')$. Thus, without any ambiguity,  we can treat an $H$-invariant $\phi$ as a function defined on $M$ as well.   

If a function $\phi$ is {\em not} $H$-invariant, then one can construct an $H$-invariant function by averaging $\phi$ over the subgroup $H$. Since $H$ is  compact, we equip $H$ with the normalized Haar measure~\cite[Ch.~VIII]{knapp2013lie}, i.e., $\int_H {\bf 1}_H dh = 1$. We then define a function on $G$ (and on $M$) by averaging the given function $\phi$ over~$H$ as follows:
\begin{equation}\label{eq:averaging}
\bar \phi(g) := \int_H \phi(g h) dh.
\end{equation} 
It should be clear that $\bar \phi$ is $H$-invariant; indeed, for any $h'\in H$, we have  
$$
\bar \phi(gh') :=\int_H \phi(gh' h) dh = \int_H\phi(g h) d(h'^{-1} h) =  \int_H\phi(g h) d h =  \bar \phi(g).
$$
Note that if $\phi$ itself is $H$-invariant, then $\bar \phi = \phi$. We now have the following fact:

%\begin{remark}
%Let $G$ be a matrix Lie group and $H$ and be a compact subgroup of $G$. Let $\{\phi^{ij}\}^m_{i,j = 1}$ be defined in~\eqref{eq:defpsiij}. Then, we can express each $\bar \phi^{ij}$ explicitly as follows:   
%$$\bar \phi^{ij}(g) := \int_H \tr(ghX_jh^{-1}g^{-1} X^\top_i)dh.$$ 
%The above integral is closely related to the Harish-Chandra-Itzykson-Zuber integral formula~\cite{harish1957differential,itzykson1980planar}: there, one integrates a moment generating function: 
%$$\int_G \exp(t\tr(gXg^{-1}Y))dg$$ over a compact group~$G$ (often, $G = \SU(n)$ or $G = \SO(n)$). %We defer the analysis of the integral to another occasion.  
%%Existing results on the integral formula~\cite{duistermaat1982variation,atiyah1984moment,alfaro1995itzykson} and random matrices~\cite{dyson1962brownian,tao2012random,terrencetaoweb} will be useful for computing~$\bar \phi^{ij}$.  
%\end{remark}

\begin{lemma}\label{lem:passontohomo}
Let $\{\phi^j\}^l_{j = 1}$ be a set of functions on $G$ codistinguished to a set of right-invariant vector fields $\{R_{X_i}\}^m_{i = 1}$.  
If $R_{X_i} \phi^j =\lambda \phi^k$, then $\tau(X_i) \bar \phi^j = \lambda\bar\phi^k$. 
\end{lemma}

\begin{proof}
The lemma directly follows from computation:  
$$
(\tau(X_i) \bar \phi^j) (gH)  = \int_H (R_{X_i}\phi^j)(gh)d h = \lambda \int_H \phi^k(gh) dh = \lambda \bar \phi^k(gH),
$$
which holds for all $gH \in M\approx G/H$.  
\end{proof}

%Note, in particular, that if the two functions $\bar \phi^j$ and $\bar \phi^k$ in the above Lemma are treated as functions on $G/H$, then $\tau(X_i) \bar \phi^j = \lambda\bar\phi^k$. 
%Thus, if $\{\phi^{j}\}^l_{j = 1}$ is (weakly) codistinguished to $\{R_{X_i}\}^m_{i = 1}$, then by Lemma~\ref{lem:passontohomo}, the second item of Def.~\ref{def:codisth} holds for $\{\bar \phi^j\}^l_{j = 1}$ with respect to the induced vector fields $\{\tau(X_i)\}^m_{i = 1}$ over~$M$. 
Thus, if the set of one-forms $\{d\bar \phi^{j}_x\}^l_{j = 1}$ spans $T^*_xM$, then, by Lemma~\ref{lem:passontohomo}, $\{\bar \phi^{j}\}^l_{j = 1}$ is (weakly) codistinguished to $\{\tau(X_i)\}^m_{i = 1}$.  We provide below an example for illustration.

\begin{example} {\em 
Let $\SO(3)$ act $S^2$ by sending $(g, x) \in \SO(3)\times S^2$ to $gx\in S^2$. Let~$H$ be a subgroup of $G$ defined as follows: 
$$
H = \left \{ h(\theta): = 
\begin{bmatrix}
1 & 0 & 0\\
0 & \cos(\theta) & \sin(\theta) \\
0 & -\sin(\theta) & \cos(\theta)
\end{bmatrix}
\mid \theta \in [0,2\pi)\right\} \approx \SO(2).
$$
It follows that $H$ is the stabilizer of the vector $e_1\in S^{2}$.  
Let $\{X_i\}^3_{i = 1}$ and $\{\phi^{ij}\}^3_{i, j = 1}$ be given in Example~\ref{emp:example1}, i.e., 
$$
\left\{
\begin{array}{lr}
X_i = e_je_k^\top - e_k e_j^\top, & \quad \mbox{where } \det(e_i, e_j, e_k) = 1, \vspace{3pt}\\
\phi^{ij}(g) = \tr(g X_j g^\top X_i^\top), & \quad 1\le i, j \le 3.
\end{array}
\right. 
$$
%We have shown that $\{\phi^{ij}\}^3_{i,j = 1}$ is codistinguished to $\{L_{X_i}\}$. In fact, it is codistinguished to both  $\{L_{X_i}\}$ and $\{R_{X_i}\}^3_{i = 1}$. We omit here the computational details. 
%By computation (with details omitted), $\{\phi^{ij}\}^3_{i,j = 1}$ is codistinguished to $\{R_{X_i}\}^3_{i = 1}$ (and it is codistinguished to $\{L_{X_i}\}^3_{i = 1}$ as well).  
%We note that by~\eqref{eq:matrixrealization}, $\phi^{ij}(g)$ is, up to scaling, equal to $B_\theta(\Ad(g) X_j, X_i)$. We also note that   
%by~\eqref{eq:anotherversion}, we have, up to the same scaling, that 
%$$
%\tilde \phi^{ij}(g) = B_\theta(\Ad(g^{-1}) X_j, X_i) = \tr(g^\top X_j g X_i^\top) = \tr(g X_i g^\top X_j^\top)  = \phi^{ji}(g). 
%$$
%Thus, $\{\phi^{ij}\}_{1\le i,j \le 3} = \{\tilde \phi^{ij}\}_{1\le i, j \le 3}$ and, hence,  by %Theorem~\ref{thm:distinguished} and 
%Remark~\ref{rmkk:distinguished}, $\{\phi^{ij}\}_{1\le i, j \le 3}$ is codistinguished to both $\{L_{X_i}\}^3_{i = 1}$ and  $\{R_{X_i}\}^3_{i = 1}$. 
%%In particular, one can apply Lemma~\ref{lem:passontohomo} and obtain that the set of $H$-invariant functions $\{\bar \phi^{ij}\}_{1\le i, j\le 3}$ is codistinguished to the induced vector fields $\{\tau(X_i)\}^3_{i = 1}$ on~$S^2$.  

%We now compute explicitly the induced vector fields $\{\tau(X_i)\}^3_{i = 1}$ and the averaged $H$-invariant functions $\{\bar \phi^{ij}\}_{1\le i, j\le 3}$. 
Because the set $\{X_i\}^3_{i = 1}$ is distinguished in $\so(3)$, by Prop.~\ref{pro:distonhomo}, it induces a distinguished set of vector fields $\{\tau(X_i)\}^3_{i = 1}$ over $S^2$:    
$$
\tau(X_1)(x) = 
\begin{bmatrix}
0 \\
x_3 \\
- x_2 
\end{bmatrix},
\quad 
 \tau(X_2)(x) = 
\begin{bmatrix}
-x_3 \\
0 \\
 x_1 
\end{bmatrix}, \quad
 \tau(X_3)(x)= 
\begin{bmatrix}
x_2 \\
- x_1 \\
0
\end{bmatrix}. 
$$
These vector fields satisfy the following relationship:  
$$
[\tau(X_i), \tau(X_j)] = - \det(e_i,e_j,e_k)\tau(X_k).  
$$

We next compute the averaged $H$-invariant functions $\{\bar \phi^{ij}\}^3_{ i, j = 1}$.  The normalized Haar measure on $H$, in this case, is simply given by $dh =\nicefrac{d\theta}{2\pi}$. It follows that 
$$
\bar \phi^{ij}(gH) = \frac{1}{2\pi}\int^{2\pi}_0 \tr(gh(\theta) X_j h(\theta)^\top g^\top X_i^\top)d\theta. 
$$
To evaluate the above integral,  we first have the following computational result: 
$$
\frac{1}{2\pi}\int^{2\pi}_{0} h(\theta) X_j h(\theta)^\top d\theta = 
\left\{
\begin{array}{ll}
X_1 & \mbox{ if } j = 1, \\
0 & \mbox{ otherwise.}
\end{array}
\right.
$$
Thus, the nonzero $\bar \phi^{ij}$'s are given by
\begin{equation}\label{eq:computephibar}
\bar \phi^{i}(gH):= \bar\phi^{i 1}(gH) =  \tr(gX_1g^\top X_i^\top), \quad \forall i = 1, 2, 3. 
\end{equation}

Each left coset $gH$ corresponds to the point $x  =  g e_1\in S^2$. Note that $ge_1$ is simply the first column of $g$.  We now compute each function $\bar \phi^i(g)$ and express the results using only the coordinates $x_i$ of~$x$. First,  by computation, we obtain
$$
gX_1 g^\top =
\begin{bmatrix}
0 & c_{31} & c_{21} \\
-c_{31} & 0 & c_{11} \\
-c_{21} & -c_{11} & 0
\end{bmatrix},
$$
where each $c_{ij}$ is the $ij$th entry of the  cofactor matrix $[c_{ij}]$ of $g\in \SO(3)$.  Since $g$ is an orthogonal matrix,   $g = [c_{ij}]$. In particular, $(c_{11},c_{21}, c_{31})$ is the first column of~$g$, i.e., $(c_{11},c_{21}, c_{31}) = ge_1 = (x_1,x_2, x_3)$ and, hence,  
$$
gX_1 g^\top = \begin{bmatrix}
0 & x_3 & x_2 \\
-x_3 & 0 & x_1 \\
-x_2 & -x_1 & 0
\end{bmatrix}.
$$  
Thus, the functions $\bar \phi^i$ in~\eqref{eq:computephibar}, for $i = 1, 2,3$, are nothing but twice the coordinate functions, i.e.,  
$$
\bar \phi^i(x) = 2 x_i. 
$$ 
%i.e., the $\bar \phi^i$'s are twice the standard coordinate functions .  
It should be clear that $\{\bar\phi^j\}^3_{j = 1}$ satisfies items~1 and~3 of Def.~\ref{def:codisth}. For item~2, we have that $\tau(X_i)\bar \phi^j = \det(e_i,e_j,e_k)  \bar \phi^k$. Thus, $\{\bar \phi^j\}^3_{j = 1}$ is codistinguished to $\{\tau(X_i)\}^3_{i = 1}$. 
}
\end{example}

\section{Conclusions}
We introduced in the paper a novel class of ensemble systems, namely distinguished ensemble systems. Every such system is comprised of two key components: a set of distinguished vector fields and a set of (weakly) codistinguished functions. We established in Section~\S\ref{sec:controlobserverdual} the fact (Theorem~\ref{thm:mthm1}) that a distinguished ensemble system is   approximately ensemble path-controllability and (weakly) ensemble observable. Specifically, we have shown that if the set of parametrization functions separates points and contains an everywhere nonzero function, then an ensemble system~\eqref{eq:ensemblesystem}, with jointly distinguished control vector fields and observation functions, is ensemble controllable and observable. %We called any such system a distinguished ensemble. 

We further extend in Section~\S\ref{ssec:preensemble} the  result to a pre-distinguished ensemble system (with pre-distinguished vector fields and pre-codistinguished functions). There, we introduced indicator sequences $\mathbb{N}_i$ and $\mathbb{N}^j$ and established the fact that every such sequence contains an infinite arithmetic sequence as a subsequence. This fact was instrumental in establishing ensemble controllability and observability of a pre-distinguished ensemble system. 

We demonstrated in Section~\S\ref{sec:controlobserverdual} that the structure of a (pre-)distinguished ensemble system can significantly simplify the analyses of ensemble controllability and observability. Moreover,  such a structure can be used as a guiding principle for designing dynamics of a general ensemble system that is controllable and observable.         
%In particular, the analyses carried out in the section showed that the structure of a distinguished ensemble system simplifies the controllability/observability problem by reducing it to a problem of finding a separating set of functions on the parametrization space.     

We proposed and addressed in Section~\S\ref{sec:onliegroup} the problem about existence of distinguished vector fields and codistinguished functions on a given manifold~$M$. 
We provided an affirmative answer for the case where $M$ is a connected, semi-simple Lie group~$G$. Specifically, we showed that every such Lie group~$G$ admits a set of distinguished left- (or right-) invariant vector fields, together with a set of  matrix coefficients that is (weakly) codistinguished to the set of left- (or right-) invariant vector fields. 

The proof was constructive. For distinguished vector fields, we identified the space of left-invariant vector fields with the Lie algebra~$\g$ associated with~$G$. The existence problem was thus addressed on the Lie algebra level. More specifically, we leveraged the result established in~\cite{chen2017distinguished} in which we have shown that every semi-simple real Lie algebra admits a distinguished set.  

For codistinguished functions, we showed in Section~\S\ref{ssec:matrixcoeffi} that a selected set of matrix coefficients associated with a Lie group representation can be made codistinguished to a given set of left- (or right-) invariant vector fields. We then focussed in Section~\S\ref{ssec:adjointrepresentationcodist} on a special representation, namely the adjoint representation. There, we constructed explicitly a set of matrix coefficients, i.e., $\{B_\theta(\Ad(g)X_j, X_i)\}^m_{i,j= 1}$, and showed in Section~\S\ref{ssec:proofofcodist} that it is codistinguished to a set of left-invariant vector fields. %$\{L_{X_i}\}^m_{i = 1}$.  
%In particular, the above problem covers the following as a special case: Given a semi-simple real Lie algebra~$\g$, classify all distinguished sets of it. 

%An open problem (Problem.~1) we posed at the end of Section~\S\ref{ssec:adjointrepresentationcodist} is the following:   
%{\em Given a semi-simple Lie group~$G$ with $\g$ the Lie algebra, classify all representations~$\pi$ of~$G$ on~$V$, together with subsets $\{X_i\}^m_{i = 1}\subset \g$ and $\{v_j\}^p_{j = 1}\subset V$, such that $\{v_j\}^p_{j =1}$ is codistinguished to $\{X_i\}^m_{i = 1}$ with respect to~$\pi$.} The solution to the above problem will classify not only distinguished sets of a given Lie algebra~$\g$, but also codistinguished functions on $G$ that can be realized as matrix coefficients.  

Finally, in Section~\S\ref{ssec:homogeneousspace}, we discussed how to translate distinguished vector fields and codistinguished functions from a semi-simple Lie group $G$ to its homogeneous spaces. %For the set of distinguished vector fields, there is a natural way of doing this~\eqref{eq:liealghomo}. The situation is more complicated for codistinguished functions, and is in fact an ongoing research direction.     
We showed that for distinguished vector fields, there is a canonical way of using a distinguished set of the Lie algebra~$\g$ to induce a set of distinguished vector fields on the homogeneous space. But, for codistinguished functions, the situation is more complicated. We proposed an averaging method which uses an existing set of codistinguished functions on $G$ to generate a set of averaged $H$-invariant functions on the homogeneous space $M \approx G/H$. These averaged functions will be codistinguished to the set of induced vector fields provided that the one-forms $\{d\bar\phi^j_x\}^l_{j = 1}$ span the cotangent space $T^*_xM$ for all $x\in M$. We provided at the end of Section~\ref{ssec:homogeneousspace} a simple example regarding $S^2$ to illustrate the averaging method.

%%%%%%%%%%%%%%%%%%%%%%%%%%%%%%%%%%%%%%%%%%%%%%%%%%%%%%%%%%%

%\bibliographystyle{IEEEtran}

\bibliographystyle{IEEEtran}

\bibliography{Structuretheoryensemblecontrol}

%\section*{Appendix}
%\setcounter{subsection}{0}

\end{document}